\documentclass[11pt,letterpaper]{article}

\def\showauthornotes{0}
\def\showtableofcontents{1}
\def\showkeys{0}
\def\showdraftbox{0}
\def\showcolorlinks{1}
\def\usemicrotype{1}
\def\showfixme{1}

\def\arxivmode{1}

\usepackage{etex}

\usepackage[l2tabu, orthodox]{nag}

\usepackage{xspace,enumerate}

\usepackage[dvipsnames]{xcolor}

\usepackage[T1]{fontenc}
\usepackage[full]{textcomp}

\usepackage[american]{babel}

\usepackage{mathtools}

\usepackage{amsthm}

\newtheorem{theorem}{Theorem}[section]
\newtheorem*{theorem*}{Theorem}

\newtheorem{proposition}[theorem]{Proposition}
\newtheorem*{proposition*}{Proposition}
\newtheorem{lemma}[theorem]{Lemma}
\newtheorem*{lemma*}{Lemma}
\newtheorem{corollary}[theorem]{Corollary}
\newtheorem*{conjecture*}{Conjecture}
\newtheorem{fact}[theorem]{Fact}
\newtheorem*{fact*}{Fact}

\newtheorem*{hypothesis*}{Hypothesis}

\theoremstyle{definition}

\newtheorem{problem}[theorem]{Problem}

\theoremstyle{remark}

\newtheorem*{claim*}{Claim}
\newtheorem{remark}[theorem]{Remark}
\newtheorem*{remark*}{Remark}

\newtheorem*{observation*}{Observation}

\usepackage[
letterpaper,
top=1in,
bottom=1in,
left=1in,
right=1in]{geometry}

\ifnum\arxivmode=0
\usepackage{newpxtext} %
\usepackage{textcomp} %
\usepackage[varg,bigdelims]{newpxmath}
\usepackage[scr=rsfso]{mathalfa}%
\usepackage{bm} %
\let\mathbb\varmathbb
\fi

\ifnum\arxivmode=1
\usepackage[varg]{pxfonts} %
\usepackage{textcomp} %
\usepackage[scr=rsfso]{mathalfa}%
\usepackage{bm} %
\fi

\ifnum\showkeys=1
\usepackage[color]{showkeys}
\fi

\ifnum\showcolorlinks=1
\usepackage[
pagebackref,
colorlinks=true,
urlcolor=blue,
linkcolor=blue,
citecolor=OliveGreen,
]{hyperref}
\fi

\ifnum\showcolorlinks=0
\usepackage[
pagebackref,
colorlinks=false,
pdfborder={0 0 0}
]{hyperref}
\fi

\usepackage{prettyref}

\newcommand{\savehyperref}[2]{\texorpdfstring{\hyperref[#1]{#2}}{#2}}

\newrefformat{eq}{\savehyperref{#1}{\textup{(\ref*{#1})}}}
\newrefformat{eqn}{\savehyperref{#1}{\textup{(\ref*{#1})}}}
\newrefformat{lem}{\savehyperref{#1}{Lemma~\ref*{#1}}}
\newrefformat{def}{\savehyperref{#1}{Definition~\ref*{#1}}}
\newrefformat{thm}{\savehyperref{#1}{Theorem~\ref*{#1}}}
\newrefformat{cor}{\savehyperref{#1}{Corollary~\ref*{#1}}}
\newrefformat{cha}{\savehyperref{#1}{Chapter~\ref*{#1}}}
\newrefformat{sec}{\savehyperref{#1}{Section~\ref*{#1}}}
\newrefformat{app}{\savehyperref{#1}{Appendix~\ref*{#1}}}
\newrefformat{tab}{\savehyperref{#1}{Table~\ref*{#1}}}
\newrefformat{fig}{\savehyperref{#1}{Figure~\ref*{#1}}}
\newrefformat{hyp}{\savehyperref{#1}{Hypothesis~\ref*{#1}}}
\newrefformat{alg}{\savehyperref{#1}{Algorithm~\ref*{#1}}}
\newrefformat{rem}{\savehyperref{#1}{Remark~\ref*{#1}}}
\newrefformat{item}{\savehyperref{#1}{Item~\ref*{#1}}}
\newrefformat{step}{\savehyperref{#1}{step~\ref*{#1}}}
\newrefformat{conj}{\savehyperref{#1}{Conjecture~\ref*{#1}}}
\newrefformat{fact}{\savehyperref{#1}{Fact~\ref*{#1}}}
\newrefformat{prop}{\savehyperref{#1}{Proposition~\ref*{#1}}}
\newrefformat{prob}{\savehyperref{#1}{Problem~\ref*{#1}}}
\newrefformat{claim}{\savehyperref{#1}{Claim~\ref*{#1}}}
\newrefformat{relax}{\savehyperref{#1}{Relaxation~\ref*{#1}}}
\newrefformat{red}{\savehyperref{#1}{Reduction~\ref*{#1}}}
\newrefformat{part}{\savehyperref{#1}{Part~\ref*{#1}}}

\newcommand{\Sref}[1]{\hyperref[#1]{\S\ref*{#1}}}

\usepackage{nicefrac}

\ifnum\usemicrotype=1
\usepackage{microtype}
\fi

\ifnum\showauthornotes=1
\newcommand{\Authornote}[2]{{\sffamily\small\color{red}{[#1: #2]}}}
\newcommand{\Authornotecolored}[3]{{\sffamily\small\color{#1}{[#2: #3]}}}
\newcommand{\Authorcomment}[2]{{\sffamily\small\color{gray}{[#1: #2]}}}
\newcommand{\Authorstartcomment}[1]{\sffamily\small\color{gray}[#1: }

\newcommand{\Authorfnote}[2]{\footnote{\color{red}{#1: #2}}}
\newcommand{\Authorfixme}[1]{\Authornote{#1}{\textbf{??}}}
\newcommand{\Authormarginmark}[1]{\marginpar{\textcolor{red}{\fbox{\Large #1:!}}}}
\else
\newcommand{\Authornote}[2]{}
\newcommand{\Authornotecolored}[3]{}
\newcommand{\Authorcomment}[2]{}
\newcommand{\Authorstartcomment}[1]{}

\newcommand{\Authorfnote}[2]{}
\newcommand{\Authorfixme}[1]{}
\newcommand{\Authormarginmark}[1]{}
\fi

\newcommand{\Jnote}{\Authornote{J}}

\ifnum\showfixme=0

\fi

\usepackage{boxedminipage}

\newcommand{\paren}[1]{(#1)}
\newcommand{\Paren}[1]{\left(#1\right)}

\newcommand{\brac}[1]{[#1]}
\newcommand{\Brac}[1]{\left[#1\right]}

\newcommand{\abs}[1]{\lvert#1\rvert}
\newcommand{\Abs}[1]{\left\lvert#1\right\rvert}

\newcommand{\Set}[1]{\left\{#1\right\}}

\newcommand{\norm}[1]{\lVert#1\rVert}
\newcommand{\Norm}[1]{\left\lVert#1\right\rVert}

\newcommand{\iprod}[1]{\langle#1\rangle}
\newcommand{\Iprod}[1]{\left\langle#1\right\rangle}

\newcommand{\Esymb}{\mathbb{E}}
\newcommand{\Psymb}{\mathbb{P}}

\DeclareMathOperator*{\E}{\Esymb}

\DeclareMathOperator*{\ProbOp}{\Psymb}

\renewcommand{\Pr}{\ProbOp}

\newcommand{\Prob}[2][]{\Pr_{{#1}}\Set{#2}}

\newcommand{\ex}[1]{\E\brac{#1}}
\newcommand{\Ex}[2][]{\E_{{#1}}\Brac{#2}}

      \small

\newcommand{\tensor}{\otimes}

\newcommand{\textparen}[1]{\text{(#1)}}

\ifx\because\undefined
\newcommand{\because}[1]{\textparen{because #1}}
\else
\renewcommand{\because}[1]{\textparen{because #1}}
\fi

\newcommand{\lmax}{\lambda_{\max}}

\newcommand{\vbig}{\vphantom{\bigoplus}}

\newcommand{\defeq}{\stackrel{\mathrm{def}}=}

\newcommand{\seteq}{\mathrel{\mathop:}=}

\newcommand{\from}{\colon}

\newcommand{\mper}{\,.}
\newcommand{\mcom}{\,,}

\newcommand\bdot\bullet

\ifx\mathds\undefined %
\DeclareMathOperator{\Ind}{\mathbb{I}}
\else
\DeclareMathOperator{\Ind}{\mathds 1}}
\fi

\DeclareMathOperator{\Tr}{Tr}

\DeclareMathOperator{\poly}{poly}

\DeclareMathOperator{\polylog}{polylog}

\newcommand{\etal}{et al.\xspace}

\newcommand{\N}{\mathbb N}
\newcommand{\R}{\mathbb R}

\newcommand{\cD}{\mathcal D}

\newcommand{\cL}{\mathcal L}

\newcommand{\cN}{\mathcal N}

\renewcommand{\leq}{\leqslant}
\renewcommand{\le}{\leqslant}
\renewcommand{\geq}{\geqslant}
\renewcommand{\ge}{\geqslant}

\ifnum\showdraftbox=1
\newcommand{\draftbox}{\begin{center}
  \fbox{%
    \begin{minipage}{2in}%
      \begin{center}%
          \Large\textsc{Working Draft}\\%
        Please do not distribute%
      \end{center}%
    \end{minipage}%
  }%
\end{center}
\vspace{0.2cm}}
\else
\newcommand{\draftbox}{}
\fi

\let\epsilon=\varepsilon

\numberwithin{equation}{section}

\newcommand\MYcurrentlabel{xxx}

\newcommand{\MYstore}[2]{%
  \global\expandafter \def \csname MYMEMORY #1 \endcsname{#2}%
}

\newcommand{\MYload}[1]{%
  \csname MYMEMORY #1 \endcsname%
}

\newcommand{\MYnewlabel}[1]{%
  \renewcommand\MYcurrentlabel{#1}%
  \MYoldlabel{#1}%
}

\newcommand{\MYdummylabel}[1]{}

\newcommand{\torestate}[1]{%
  \let\MYoldlabel\label%
  \let\label\MYnewlabel%
  #1%
  \MYstore{\MYcurrentlabel}{#1}%
  \let\label\MYoldlabel%
}

\newcommand{\restatetheorem}[1]{%
  \let\MYoldlabel\label
  \let\label\MYdummylabel
  \begin{theorem*}[Restatement of \prettyref{#1}]
    \MYload{#1}
  \end{theorem*}
  \let\label\MYoldlabel
}

\newcommand{\restatelemma}[1]{%
  \let\MYoldlabel\label
  \let\label\MYdummylabel
  \begin{lemma*}[Restatement of \prettyref{#1}]
    \MYload{#1}
  \end{lemma*}
  \let\label\MYoldlabel
}

\newcommand{\restateprop}[1]{%
  \let\MYoldlabel\label
  \let\label\MYdummylabel
  \begin{proposition*}[Restatement of \prettyref{#1}]
    \MYload{#1}
  \end{proposition*}
  \let\label\MYoldlabel
}

\newcommand{\restatefact}[1]{%
  \let\MYoldlabel\label
  \let\label\MYdummylabel
  \begin{fact*}[Restatement of \prettyref{#1}]
    \MYload{#1}
  \end{fact*}
  \let\label\MYoldlabel
}

\newcommand{\restate}[1]{%
  \let\MYoldlabel\label
  \let\label\MYdummylabel
  \MYload{#1}
  \let\label\MYoldlabel
}

\newcommand{\addreferencesection}{
  \phantomsection
  \addcontentsline{toc}{section}{References}
}

\newcommand{\e}{\epsilon}

\let\origparagraph\paragraph
\renewcommand{\paragraph}[1]{\origparagraph{#1.}}

\allowdisplaybreaks

\usepackage{paralist}

\usepackage{comment}

\let\citet\cite

\theoremstyle{definition}
\newtheorem{algo}[theorem]{Algorithm}

\usepackage{relsize}
\usepackage[font=footnotesize]{caption}
\usepackage[flushleft,para]{threeparttable}
\makeatletter
\g@addto@macro\TPT@defaults{\footnotesize}
\makeatother

\DeclareMathOperator{\Span}{Span}

\DeclareMathOperator{\Id}{\mathrm{Id}}

\newcommand{\bS}{\mathbf{S}}
\newcommand{\bM}{{\mathbf M}}
\newcommand{\bT}{\mathbf{T}}

\newcommand{\bA}{\mathbf{A}}

\newcommand{\op}[1]{\left \| #1 \right \|_{\mbox{\scriptsize op}}}
\newcommand{\tO}{{\tilde O}}
\newcommand{\tOmega}{{\tilde \Omega}}
\newcommand{\tTheta}{{\tilde \Theta}}

\newcommand{\signh}{R}

\newcommand{\Event}{\mathcal{E}}

\newcommand{\tallsep}{ \ \bigg{|}\ }
\newcommand{\owl}{n^{-\omega(1)}}

\newcommand{\Eone}{\mathcal{G}_{1}}
\newcommand{\Eother}{\mathcal{G}_{i>1}}
\newcommand{\ampc}{\alpha}
\newcommand{\restc}{(2 + \rho)}

\newcommand{\inn}{ { i \in [n] } }
\newcommand{\jni} { {j \neq i} }

\newcommand{\cconst}{\tfrac{d}{n}} %

\renewcommand{\d}{\mathop{}\!\mathrm{d}}

\newcommand{\Mdiag}{M_{\mathrm{diag}}}
\newcommand{\Mcross}{M_{\mathrm{cross}}}
\newcommand{\Mdiff}{M_{\mathrm{diff}}}
\newcommand{\Msame}{M_{\mathrm{same}}}

\DeclareUrlCommand\email{}

\newcommand{\wovp}{\text{w.ov.p.}\xspace}
\newcommand{\YES}{\textsc{yes}\xspace}
\newcommand{\NO}{\textsc{no}\xspace}

\newcommand{\wovple}{\hspace{-3mm}\stackrel{\text{\wovp}}{\le}}

\newcommand{\Pisym}{\Pi_{\mathrm{sym}}}

\let\pref=\prettyref

\title{Fast spectral algorithms from sum-of-squares proofs: tensor decomposition and planted sparse vectors}

\author{%
\normalsize
Samuel B. Hopkins\thanks{Cornell University.
\protect\email{samhop@cs.cornell.edu}.
Supported by an NSF Graduate Research Fellowship (NSF award no. 1144153) and by David Steurer's NSF CAREER award.}
\and
\normalsize
Tselil Schramm\thanks{UC Berkeley, \protect\email{tschramm@cs.berkeley.edu}.
Supported by an NSF Graduate Research Fellowship (NSF award no 1106400).}
\and
\normalsize
Jonathan Shi\thanks{Cornell University, \protect\email{jshi@cs.cornell.edu}. Supported by David Steurer's NSF CAREER award.}
\and
\normalsize
David Steurer\thanks{Cornell University, \protect\email{dsteurer@cs.cornell.edu}.
Supported by a Microsoft Research Fellowship, a Alfred P. Sloan Fellowship, an NSF CAREER award, and the Simons Collaboration for Algorithms and Geometry.
}
}

\begin{document}

\maketitle
\draftbox
\thispagestyle{empty}

\begin{abstract}
We consider two problems that arise in machine learning applications:
the problem of recovering a planted sparse vector in a random linear subspace and the problem of decomposing a random low-rank overcomplete 3-tensor.
For both problems, the best known guarantees are based on the sum-of-squares method.
We develop new algorithms inspired by analyses of the sum-of-squares method.
Our algorithms achieve the same or similar guarantees as sum-of-squares for these problems but the running time is significantly faster.

For the planted sparse vector problem, we give an algorithm with running time nearly linear in the input size that approximately recovers a planted sparse vector with up to constant relative sparsity in a random subspace of $\mathbb R^n$ of dimension up to $\tilde \Omega(\sqrt n)$.
These recovery guarantees match the best known ones of Barak, Kelner, and Steurer (STOC 2014) up to logarithmic factors.

For tensor decomposition, we give an algorithm with running time close to linear in the input size (with exponent $\approx 1.086$) that approximately recovers a component of a random 3-tensor over $\mathbb R^n$ of rank up to $\tilde \Omega(n^{4/3})$.
The best previous algorithm for this problem due to Ge and Ma (RANDOM 2015) works up to rank $\tilde \Omega(n^{3/2})$ but requires quasipolynomial time.
\end{abstract}

\clearpage

\ifnum\showtableofcontents=1
{
\tableofcontents
\thispagestyle{empty}
 }
\fi

\clearpage

\setcounter{page}{1}

\section{Introduction}

\newcommand{\sos}{SoS\xspace}

The sum-of-squares (\sos) method (also known as the Lasserre hierarchy) \cite{MR931698,parrilo2000structured,MR1748764,MR1814045} is a powerful, semidefinite-programming based meta-algorithm that applies to a wide-range of optimization problems.
The method has been studied extensively for moderate-size polynomial optimization problems that arise for example in control theory and in the context of approximation algorithms for combinatorial optimization problems, especially constraint satisfaction and graph partitioning (see e.g. the survey \cite{DBLP:journals/corr/BarakS14}).
For the latter, the \sos method captures and generalizes the best known approximation algorithms based on linear programming (LP), semidefinite programming (SDP), or spectral methods, and it is in many cases the most promising approach to obtain algorithms with better guarantees---especially in the context of Khot's Unique Games Conjecture \cite{DBLP:conf/stoc/BarakBHKSZ12}.

A sequence of recent works applies the sum-of-squares method to basic problems that arise in unsupervised machine learning:
in particular, recovering sparse vectors in linear subspaces and decomposing tensors in a robust way \cite{DBLP:conf/stoc/BarakKS14, DBLP:conf/stoc/BarakKS15, DBLP:conf/colt/HopkinsSS15, BM15, DBLP:conf/approx/GeM15}.
For a wide range of parameters of these problems, \sos achieves significantly stronger guarantees than other methods, in polynomial or quasi-polynomial time.

Like other LP and SDP hierarchies, the sum-of-squares method comes with a degree parameter $d\in \N$ that allows for trading off running time and solution quality.
This trade-off is appealing because for applications the additional utility of better solutions could vastly outweigh additional computational costs.
Unfortunately, the computational cost grows rather steeply in terms of the parameter $d$:
the running time is $n^{O(d)}$ where $n$ is the number of variables (usually comparable to the instance size).
Further, even when the SDP has size polynomial in the input (when $d = O(1)$), solving the underlying semidefinite programs is prohibitively slow for large instances.

In this work, we introduce spectral algorithms for planted sparse vector, tensor decomposition, and tensor principal components analysis (PCA) that exploit the same high-degree information as the corresponding sum-of-squares algorithms without relying on semidefinite programming, and achieve the same (or close to the same) guarantees.
The resulting algorithms are quite simple (a couple of lines of \textsc{matlab} code) and have considerably faster running times---quasi-linear or close to linear in the input size.

A surprising implication of our work is that for some problems, spectral algorithms can exploit information from larger values of the parameter $d$ without spending time $n^{O(d)}$.
For example, one of our algorithms runs in nearly-linear time in the input size, even though it uses properties that the sum-of-squares method can only use for degree parameter $d\ge 4$.
(In particular, the guarantees that the algorithm achieves are strictly stronger than the guarantees that \sos achieves for values of $d<4$.)

The initial successes of \sos in the machine learning setting gave hope that techniques developed in the theory of approximation algorithms, specifically the techniques of hierarchies of convex relaxations and rounding convex relaxations, could broadly impact the practice of machine learning.
This hope was dampened by the fact that in general, algorithms that rely on solving large semidefinite programs are too slow to be practical for the large-scale problems that arise in machine learning.
Our work brings this hope back into focus by demonstrating for the first time that with some care \sos algorithms can be made practical for large-scale problems.

In the following subsections we describe each of the problems that we consider, the prior best-known guarantee via the \sos hierarchy, and our results.

\subsection{Planted sparse vector in random linear subspace}
The problem of finding a sparse vector planted in a random linear subspace was introduced by Spielman, Wang, and Wright as a way of learning sparse dictionaries \cite{DBLP:journals/jmlr/SpielmanWW12}.
Subsequent works have found further applications and begun studying the problem in its own right \cite{Demanet01122014,DBLP:conf/stoc/BarakKS14, DBLP:conf/nips/QuSW14}.
In this problem, we are given a basis for a $d$-dimensional linear subspace of $\R^n$ that is random except for one planted sparse direction, and the goal is to recover this sparse direction.
The computational challenge is to solve this problem even when the planted vector is only mildly sparse (a constant fraction of non-zero coordinates) and the subspace dimension is large compared to the ambient dimension ($d \ge n^{\Omega(1)}$).

Several kinds of algorithms have been proposed for this problem based on linear programming (LP), basic semidefinite programming (SDP), sum-of-squares, and non-convex gradient descent (alternating directions method).

An inherent limitation of simpler convex methods (LP and basic SDP) \cite{DBLP:journals/jmlr/SpielmanWW12, DBLP:conf/nips/dAspremontGJL04} is that they require the relative sparsity of the planted vector to be polynomial in the subspace dimension (less than $n/\sqrt d$ non-zero coordinates).

Sum-of-squares and non-convex methods do not share this limitation.
They can recover planted vectors with constant relative sparsity even if the subspace has polynomial dimension (up to dimension $O(n^{1/2})$ for sum-of-squares \cite{DBLP:conf/stoc/BarakKS14} and up to $O(n^{1/4})$ for non-convex methods \cite{DBLP:conf/nips/QuSW14}).

\medskip
We state the problem formally:

\begin{problem}[Planted sparse vector problem with ambient dimension $n\in \N$, subspace dimension $d\le n$, sparsity $\e>0$, and accuracy $\eta>0$]
    Given an arbitrary orthogonal basis of a subspace spanned by vectors $v_0,v_1,\ldots,v_{d-1} \in \R^n$, where $v_0$ is a vector with at most $\e n$ non-zero entries and $v_1,\ldots,v_{d-1}$ are vectors sampled independently at random from the standard Gaussian distribution on $\R^n$, output a unit vector~$v\in\R^n$ that has correlation~$\langle v, v_0 \rangle^2 \ge 1-\eta$ with the sparse vector $v_0$.
\end{problem}

\paragraph{Our Results}
Our algorithm runs in nearly linear time in the input size, and matches the best-known guarantees up to a polylogarithmic factor in the subspace dimension \cite{DBLP:conf/stoc/BarakKS14}.

\begin{theorem}[Planted sparse vector in nearly-linear time]
\label{thm:planted-fast}
There exists an algorithm that, for every sparsity $\e>0$, ambient dimension $n$, and subspace dimension $d$ with $d \le \sqrt n / (\log n)^{O(1)}$, solves the planted sparse vector problem with high probability for some accuracy $\eta \le O(\epsilon^{1/4}) + o_{n\to \infty}(1)$.
The running time of the algorithm is $\tO(n d)$.
\end{theorem}

We give a technical overview of the proof in \pref{sec:techniques}, and a full proof in \pref{sec:fps}.

Previous work also showed how to recover the planted sparse vector exactly.
The task of going from an approximate solution to an exact one is a special case of standard compressed sensing (see e.g. \cite{DBLP:conf/stoc/BarakKS14}).

\begin{table}[t]
    \caption{Comparison of algorithms for the planted sparse vector problem with ambient dimension $n$, subspace dimension $d$, and relative sparsity~$\epsilon$.}
    \begin{tabular}{|c|c|c|c|c|}
      \hline
      Reference & Technique & Runtime & Largest $d$ & Largest $\epsilon$ \\
      \hline \hline
      Demanet, Hand \cite{Demanet01122014} & linear programming & poly& any & $\Omega(1/\sqrt{d})$\\
      \hline
      Barak, Kelner, Steurer \cite{DBLP:conf/stoc/BarakKS14} & SoS, general SDP & poly& $\Omega(\sqrt{n})$ & $\Omega(1)$ \\
      \hline
      Qu, Sun, Wright \cite{DBLP:conf/nips/QuSW14} & alternating minimization & $\tO(n^2 d^5)$ & $\Omega(n^{1/4})$ & $\Omega(1)$ \\
      \hline
      this work & SoS, partial traces & $\tO(n d)$  & $\tOmega(\sqrt n)$ & $\Omega(1)$ \\
      \hline
    \end{tabular}
  \label{table:comparison-fps}
\end{table}

\subsection{Overcomplete tensor decomposition}
Tensors naturally represent multilinear relationships in data.
Algorithms for tensor decompositions have long been studied as a tool for data analysis across a wide-range of disciplines (see the early work of Harshman~\cite{H70} and the survey~\cite{DBLP:journals/siamrev/KoldaB09}).
While the problem is NP-hard in the worst-case \cite{H90,HL13}, algorithms for special cases of tensor decomposition have recently led to new provable algorithmic results for several unsupervised learning problems \cite{DBLP:journals/jmlr/AnandkumarGHKT14,DBLP:conf/stoc/BhaskaraCMV14,DBLP:conf/stoc/GoyalVX14,DBLP:journals/jmlr/AnandkumarGHK14} including independent component analysis, learning mixtures of Gaussians \cite{DBLP:conf/stoc/GeHK15}, Latent Dirichlet topic modeling~\cite{DBLP:journals/algorithmica/AnandkumarFHKL15} and dictionary learning~\cite{DBLP:conf/stoc/BarakKS15}.
Some previous learning algorithms can also be reinterpreted in terms of tensor decomposition \cite{Chang96, MosselRoch06, NguyenRegev09}.

A key algorithmic challenge for tensor decompositions is \emph{overcompleteness}, when the number of components is larger than their dimension (i.e., the components are linearly dependent).
Most algorithms that work in this regime require tensors of order $4$ or higher \cite{DBLP:journals/tsp/LathauwerCC07,DBLP:conf/stoc/BhaskaraCMV14}.
For example, the FOOBI algorithm of \cite{DBLP:journals/tsp/LathauwerCC07} can recover up to $\Omega(d^2)$ components given an order-$4$ tensor in dimension $d$ under mild algebraic independence assumptions for the components---satisfied with high probability by random components.
For overcomplete $3$-tensors, which arise in many applications of tensor decompositions, such a result remains elusive.

Researchers have therefore turned to investigate average-case versions of the problem, when the  components of the overcomplete $3$-tensor are random:
Given a $3$-tensor $T\in \cramped{\R^{d^3}}$ of the form
\begin{displaymath}
  T = \sum_{i=1}^n a_i \otimes a_i \otimes a_i\mcom
\end{displaymath}
where $a_1,\ldots,a_n$ are random unit or Gaussian vectors, the goal is to approximately recover the components $a_1,\ldots,a_n$.

Algorithms based on \emph{tensor power iteration}---a gradient-descent approach for tensor decomposition---solve this problem in polynomial time when $n\le C\cdot d$ for any constant $C\ge 1$ (the running time is exponential in $C$) \cite{DBLP:conf/colt/AnandkumarGJ15}.
Tensor power iteration also admits local convergence analyses for up to $n\le \tilde\Omega(d^{1.5})$ components \cite{DBLP:conf/colt/AnandkumarGJ15,DBLP:journals/corr/AnandkumarGJ14b}.
Unfortunately, these analyses do not give polynomial-time algorithms because it is not known how to efficiently obtain the kind of initializations assumed by the analyses.

Recently, Ge and Ma~\cite{DBLP:conf/approx/GeM15} were able to show that a tensor-decomposition algorithm \cite{DBLP:conf/stoc/BarakKS15} based on sum-of-squares solves the above problem for $n \le \tilde\Omega(d^{1.5})$ in quasi-polynomial time~$n^{O(\log n)}$.
The key ingredient of their elegant analysis is a subtle spectral concentration bound for a particular degree-$4$ matrix-valued polynomial associated with the decomposition problem of random overcomplete $3$-tensors.

\begin{table}[t]
  \caption{
    Comparison of decomposition algorithms for overcomplete 3-tensors with $n$ components in dimension $d$.}
  \begin{threeparttable}
    \centering
    \begin{tabular}{|c|c|c|c|c|}
      \hline
      Reference & Technique & Runtime & Largest $n$ & Components \\ \hline \hline
      Anandkumar \etal \cite{DBLP:conf/colt/AnandkumarGJ15}\tnote{a} & tensor power iteration & poly & $C\cdot d$ & incoherent\\ \hline
      Ge, Ma \cite{DBLP:conf/approx/GeM15} & SoS, general SDP & $n^{O(\log n)}$ & $\tOmega(d^{3/2})$  & $\cN(0,\tfrac {1}{d}\Id_d)$ \\ \hline
      this work\tnote{b} & SoS, partial traces & $\tO(nd^{1 + \omega})$ & $\tOmega(d^{4/3})$ & $\cN(0,\tfrac{1}{d}\Id_d)$ \\ \hline
    \end{tabular}
    \begin{tablenotes}
    \item[a] The analysis shows that for every constant $C\ge 1$, the running time is polynomial for $n \le C\cdot d$ components.
The analysis assumes that the components also satisfy other random-like properties besides incoherence.
    \item[b] Here, $\omega\le 2.373$ is the constant so that $d \times d$ matrices can be multiplied in $O(d^\omega)$ arithmetic operations.
    \end{tablenotes}
  \end{threeparttable}
  \label{table:comparison-tdecomp}
\end{table}

\medskip
We state the problem formally:
\begin{problem}[Random tensor decomposition with dimension $d$, rank $n$, and accuracy $\eta$]
Let $a_1,\ldots,a_n \in \R^d$ be independently sampled vectors from the Gaussian distribution $\cN(0,\tfrac 1 d \Id_d)$, and let $\bT \in \cramped{(\R^d)^{\otimes 3}}$ be the 3-tensor $\bT = \sum_{i=1}^n a_i^{\tensor 3}$.
\begin{quote}
\begin{description}
\item[Single component:] Given $\bT$ sampled as above, find a unit vector $b$ that has correlation $\max_i\langle  a_i, b \rangle \ge 1-\eta$ with one of the vectors $a_i$.
\item[All components:] Given $\bT$ sampled as above, find a set of unit vectors $\{b_1,\ldots,b_n\}$ such that $\langle  a_i, b_i \rangle \ge 1-\eta$ for every $i\in [n]$.
\end{description}
\end{quote}
\end{problem}

\paragraph{Our Results}
We give the first polynomial-time algorithm for decomposing random overcomplete 3-tensors with up to $\omega(d)$ components.
Our algorithms works as long as the number of components satisfies $n \le \tilde \Omega(d^{4/3})$, which comes close to the bound $\tilde \Omega(d^{1.5})$ achieved by the aforementioned quasi-polynomial algorithm of Ge and Ma.
For the single-component version of the problem, our algorithm runs in time close to linear in the input size.

\begin{theorem}[Fast random tensor decomposition]
\label{thm:tensor-fast}
There exist randomized algorithms that, for every dimension $d$ and rank $n$ with $d\le n \le d^{4/3}/(\log n)^{O(1)}$, solve the random tensor decomposition problem with probability $1-o(1)$ for some accuracy $\eta\le \tO(n^3/d^4)^{1/2}$.
The running time for the single-component version of the problem is $\tO(\min\{d^{1+\omega}, d^{3.257}\})$, where $d^\omega$ is the time to multiply two $d$-by-$d$ matrices.
The running time for the all-components version of the problem is $\tO(n\cdot \min\{d^{1+\omega}, d^{3.257}\})$.
\end{theorem}
We give a technical overview of the proof in \pref{sec:techniques}, and a full proof in \pref{sec:tdecomp}.

We remark that the above algorithm only requires access to the input tensor with some fixed inverse polynomial accuracy because each of its four steps amplifies errors by at most a polynomial factor (see \pref{alg:four-thirds}).
In this sense, the algorithm is robust.

\subsection{Tensor principal component analysis}
The problem of tensor principal component analysis is similar to the tensor decomposition problem.
However, here the focus is not on the number of components in the tensor, but about recovery in the presence of a large amount of random noise.
We are given as input a tensor $\tau \cdot v^{\tensor 3} + \bA$, where $v \in \R^n$ is a unit vector and the entries of $\bA$ are chosen iid from $\cN(0,1)$.
This \emph{spiked tensor} model was introduced by Montanari and Richard~\cite{DBLP:conf/nips/RichardM14}, who also obtained the first algorithms to solve the model with provable statistical guarantees.
The spiked tensor model was subsequently addressed by a subset of the present authors \cite{DBLP:conf/colt/HopkinsSS15}, who applied the \sos approach to improve the signal-to-noise ratio required for recovery from odd-order tensors.

\medskip
We state the problem formally:

\begin{problem}[Tensor principal components analysis with signal-to-noise ratio $\tau$ and accuracy $\eta$]
    Let $\bT \in (\R^d)^{\tensor 3}$ be a tensor so that $\bT = \tau \cdot v^{\tensor 3} + A$, where $A$ is a tensor with independent standard gaussian entries and $v \in \R^d$ is a unit vector.
    Given $\bT$, recover a unit vector $v' \in \R^d$ such that $\iprod{v',v}\ge 1 - \eta$.
\end{problem}

\paragraph{Our results}
For this problem, our improvements over the previous results are more modest---we achieve signal-to-noise guarantees matching \cite{DBLP:conf/colt/HopkinsSS15}, but with an algorithm that runs in linear time rather than near-linear time (time $O(d^3)$ rather than $O(d^3\polylog d)$, for an input of size $d^3$).

\begin{theorem}[Tensor principal component analysis in linear time]\label{thm:tpca-intro}
  There is an algorithm which solves the tensor principal component analysis problem with accuracy $\eta>0$ whenever the signal-to-noise ratio satisfies $\tau \geq O(n^{3/4}/\eta \cdot \log^{1/2} n)$.
  Furthermore, the algorithm runs in time $O(d^3)$.
\end{theorem}

Though for tensor PCA our improvement over previous work is modest, we include the results here as this problem is a pedagogically poignant illustration of our techniques.
We give a technical overview of the proof in \pref{sec:techniques}, and a full proof in \pref{sec:tpca}.

\begin{table}
    \begin{center}
  \caption{Comparison of algorithms for principal component analysis of 3-tensors in dimension $d$ and with signal-to-noise ratio $\tau$.}
  \begin{tabular}{|c|c|c|c|}
    \hline
    Reference & Technique & Runtime & Smallest $\tau$ \\ \hline \hline
    Montanari, Richard \cite{DBLP:conf/nips/RichardM14} & spectral & {\Large\strut} $\tO(d^3)$ & $n$\\ \hline
    Hopkins, Shi, Steurer \cite{DBLP:conf/colt/HopkinsSS15} & SoS, spectral & {\Large\strut} $\tO(d^3)$ & $\Omega(n^{3/4})$\\ \hline
    this work & SoS, partial traces & $O(d^3)$ & $\tOmega(n^{3/4}) $ \\ \hline
  \end{tabular}
  \label{table:comparison-tpca}
  \end{center}
\end{table}

\subsection{Related work}

Foremost, this work builds upon the \sos algorithms of \cite{DBLP:conf/stoc/BarakKS14, DBLP:conf/stoc/BarakKS15, DBLP:conf/approx/GeM15, DBLP:conf/colt/HopkinsSS15}.
In each of these previous works, a machine learning decision problem is solved using an SDP relaxation for \sos.
In these works, the SDP value is large in the \YES case and small in the \NO case, and the SDP value can be bounded using the spectrum of a specific matrix.
This was implicit in \cite{DBLP:conf/stoc/BarakKS14, DBLP:conf/stoc/BarakKS15}, and in \cite{DBLP:conf/colt/HopkinsSS15} it was used to obtain a fast algorithm as well.
In our work, we design spectral algorithms which use smaller matrices, inspired by the \sos certificates in previous works, to solve these machine-learning problems much faster, with almost matching guarantees.

A key idea in our work is that given a large matrix with information encoded in the matrix's spectral gap, one can often efficiently ``compress'' the matrix to a much smaller one without losing that information.
This is particularly true for problems with planted solutions.
In this way, we are able to improve running time by replacing an $n^{O(d)}$-sized SDP with an eigenvector computation for an $n^{k} \times n^{k}$ matrix, for some $k < d$.

The idea of speeding up LP and SDP hierarchies for specific problems has been investigated in a series of previous works \cite{DBLP:conf/soda/VegaK07, DBLP:conf/focs/BarakRS11, DBLP:conf/focs/GuruswamiS12}, which shows that with respect to local analyses of the sum-of-squares algorithm it is sometimes possible to improve the running time from $n^{O(d)}$ to $2^{O(d)} \cdot n^{O(1)}$.
However, the scopes and strategies of these works are completely different from ours.
First, the notion of local analysis from these works does not apply to the problems considered here.
Second, these works employ the ellipsoid method with a separation oracle inspired by rounding algorithms, whereas we reduce the problem to an ordinary eigenvector computation.

It would also be interesting to see if our methods can be used to speed up some of the other recent successful applications of \sos to machine-learning type problems, such as \cite{BM15}, or the application of \cite{DBLP:conf/stoc/BarakKS14} to tensor decomposition with components that are well-separated (rather than random).
Finally, we would be remiss not to mention that \sos lower bounds exist for several of these problems, specifically for tensor principal components analysis, tensor prediction, and sparse PCA \cite{DBLP:conf/colt/HopkinsSS15, BM15, MW15}.
The lower bounds in the \sos framework are a good indication that we cannot expect spectral algorithms achieving better guarantees.

\section{Techniques}
\label{sec:techniques}

\paragraph{Sum-of-squares method (for polynomial optimization over the sphere)}

The problems we consider are connected to optimization problems of the following form:
Given a homogeneous $n$-variate real polynomial $f$ of constant degree, find a unit vector $x\in\R^n$ so as to maximize $f(x)$.
The sum-of-squares method allows us to efficiently compute upper bounds on the maximum value of such a polynomial $f$ over the unit sphere.

For the case that $k=\deg (f)$ is even, the most basic upper bound of this kind is the largest eigenvalue of a \emph{matrix representation} of $f$. A \emph{matrix representation} of a polynomial $f$ is a symmetric matrix $M$ with rows and columns indexed by monomials of degree $k/2$ so that $f(x)$ can be written as the quadratic form $f(x)=\langle  x^{\otimes k/2} , M x^{\otimes k/2}\rangle$, where $x^{ \otimes k/2}$ is the $k/2$-fold tensor power of $x$.
The largest eigenvalue of a matrix representation $M$ is an upper bound on the value of $f(x)$ over unit $x \in \R^n$ because
\[
    f(x)=\langle  x^{\otimes k/2}, M x^{\otimes k/2} \rangle\le \lmax(M) \cdot \lVert  x^{\otimes k/2} \rVert_2^2=\lmax(M)\mper
\]

The sum-of-squares methods improves on this basic spectral bound systematically by associating a large family of polynomials (potentially of degree higher than $\deg (f)$) with the input polynomial $f$ and computing the best possible spectral bound within this family of polynomials.
Concretely, the sum-of-squares method with degree parameter $d$ applied to a polynomial $f$ with $\deg (f) \le d$ considers the affine subspace of polynomials $\{ f + (1-\lVert  x \rVert_2^2)\cdot g \mid \deg (g) \le d-2\}\subseteq \R[x]$ and minimizes $\lmax(M)$ among all matrix representations
\footnote{
    Earlier we defined matrix representations only for homogeneous polynomials of even degree.
    In general, a matrix representation of a polynomial $g$ is a symmetric matrix $M$ with rows and columns indexed by monomials of degree at most $\ell=\deg(g)/2$ such that $g(x)=\langle  x^{\otimes \le  \ell}, M x^{\otimes \le  \ell} \rangle$ (as a polynomial identity), where $x^{\otimes \le \ell}=\tfrac 1 {\sqrt{\ell+1}}(x^{\otimes 0},x^{\otimes 1},\ldots,x^{\otimes \ell})$ is the vector of all monomials of degree at most $\ell$.
 Note that $\lVert  x^{\otimes \le \ell} \rVert=1$ for all $x$ with $\lVert  x \rVert=1$.
}
$M$ of polynomials in this space.\footnote{
The name of the method stems from the fact that this last step is equivalent to finding the minimum number $\lambda$ such that the space contains a polynomial of the form $\lambda - (g_1^2 + \cdots + g_t^2)$, where $g_1,\ldots,g_t$ are polynomials of degree at most $d/2$.}
The problem of searching through this affine linear space of polynomials and their matrix representations and finding the one of smallest maximum eigenvalue can be solved using semidefinite programming.

Our approach for faster algorithms based on \sos algorithms is to construct specific matrices (polynomials) in this affine linear space, then compute their top eigenvectors.
By designing our matrices carefully, we ensure that our algorithms have access to the same higher degree information that the sum-of-squares algorithm can access, and this information affords an advantage over the basic spectral methods for these problems.
At the same time, our algorithms avoid searching for the best polynomial and matrix representation, which gives us faster running times since we avoid semidefinite programming.
This approach is well suited to average-case problems where we avoid the problem of adversarial choice of input;
in particular it is applicable to machine learning problems where noise and inputs are assumed to be random.

\paragraph{Compressing matrices with partial traces}
A serious limitation of the above approach is that the representation of a degree-$d$ polynomial requires size roughly $n^d$.
Hence, even avoiding the use of semidefinite programming, improving upon running time $O(n^{d})$ requires additional ideas.

In each of the problems that we consider, we have a large matrix (suggested by a \sos algorithm) with a ``signal'' planted in some amount of ``noise''.
We show that in some situations, this large matrix can be compressed significantly without loss in the signal by applying \emph{partial trace} operations.
In these situations, the partial trace yields a smaller matrix with the same signal-to-noise ratio as the large matrix suggested by the \sos algorithm, even in situations when lower degree sum-of-squares approaches are known to fail.

The partial trace $\Tr_{\R^d} \from \R^{d^2 \times d^2}\to \R^{d\times d}$ is the linear operator that satisfies $\Tr_{\R^d} A \otimes B = (\Tr A) \cdot B$ for all $A,B\in \R^{d \times d}$.
To see how the partial trace can be used to compress large matrices to smaller ones with little loss, consider the following problem:
Given a matrix $M \in \R^{d^2 \times d^2}$ of the form $ M = \tau \cdot (v \tensor v)(v \tensor v)^\top + A \tensor B$ for some unit vector $v\in \R^d$ and matrices $A,B\in \R^{d\times d}$, we wish to recover the vector $v$.
(This is a simplified version of the planted problems that we consider in this paper, where $\tau \cdot (v \tensor v)(v \tensor v)^\top$ is the signal and $A\tensor B$ plays the role of noise.)

It is straightforward to see that the matrix $A\tensor B$ has spectral norm $\|A\tensor B\| = \|A\|\cdot \|B\|$, and so when $\tau \gg \|A\|\|B\|$, the matrix $M$ has a noticeable spectral gap, and the top eigenvector of $M$ will be close to $v \tensor v$.
If $|\Tr A | \approx \|A\|$, the matrix $\Tr_{\R^d} M = \tau \cdot vv^\top + \Tr(A) \cdot B$ has a matching spectral gap, and we can still recover $v$, but now we only need to compute the top eigenvector of a $d \times d$ (as opposed to $d^2 \times d^2$) matrix.%
\footnote{In some of our applications, the matrix $M$ is only represented implicitly and has size super linear in the size of the input, but nevertheless we can compute the top eigenvector of the partial trace $\Tr_{\R^d} M$ in nearly-linear time.}

If $A$ is a Wigner matrix (e.g. a symmetric matrix with iid $\pm 1$ entries), then both $\Tr(A),\|A\| \approx \sqrt{n}$, and the above condition is indeed met.
In our average case/machine learning settings the ``noise'' component is not as simple as $A \tensor B$ with $A$ a Wigner matrix.
Nonetheless, we are able to ensure that the noise displays a similar behavior under partial trace operations.
In some cases, this requires additional algorithmic steps, such as random projection in the case of tensor decomposition, or centering the matrix eigenvalue distribution in the case of the planted sparse vector.

It is an interesting question if there are general theorems describing the behavior of spectral norms under partial trace operations.
In the current work, we compute the partial traces explicitly and estimate their norms directly.
Indeed, our analyses boil down to concentrations bounds for special matrix polynomials.
A general theory for the concentration of matrix polynomials is a notorious open problem (see \cite{DBLP:journals/corr/MekaW13}).

Partial trace operations have previously been applied for rounding \sos relaxations.
Specifically, the operation of \emph{reweighing} and \emph{conditioning}, used in rounding algorithms for sum-of-squares such as~\cite{DBLP:conf/focs/BarakRS11,DBLP:conf/soda/RaghavendraT12,DBLP:conf/stoc/BarakKS14, DBLP:conf/stoc/BarakKS15,LR15}, corresponds to applying a partial trace operation to the moments matrix returned by the sum-of-squares relaxation.
\medskip

We now give a technical overview of our algorithmic approach for each problem, and some broad strokes of the analysis for each case.
Our most substantial improvements in runtime are for the planted sparse vector  and overcomplete tensor decomposition problems
(\pref{sec:techniques-planted-fast} and \pref{sec:techniques-tensor-decomp} respectively).
Our algorithm for tensor PCA is the simplest application of our techniques, and it may be instructive to skip ahead and read about tensor PCA first (\pref{sec:techniques-tpca}).

\subsection{Planted sparse vector in random linear subspace}
\label{sec:techniques-planted-fast}
Recall that in this problem we are given a linear subspace $U$ (represented by some basis) that is spanned by a $k$-sparse unit vector $v_0\in \R^d$ and random unit vectors $v_1,\ldots,v_{d-1}\in \R^d$.
The goal is to recover the vector $v_0$ approximately.

\vspace{-4mm}
\paragraph{Background and \sos analysis}
Let $A\in \R^{n\times d}$ be a matrix whose columns form an orthonormal basis for $U$.
Our starting point is the polynomial $f(x)=\lVert A x \rVert_4^4=\sum_{i=1}^n (A x)_i^4$.
Previous work showed that for $d\ll \sqrt n$ the maximizer of this polynomial over the sphere corresponds to a vector close to $v_0$ and that degree-$4$ sum-of-squares is able to capture this fact \cite{DBLP:conf/stoc/BarakBHKSZ12, DBLP:conf/stoc/BarakKS14}.
Indeed, typical random vectors $v$ in $\R^n$ satisfy $\lVert v \rVert_4^4\approx 1/n$ whereas our planted vector satisfies $\lVert  v_0 \rVert_4^4 \ge 1/k \gg 1/n$, and this degree-$4$ information is leveraged by the \sos algorithms.

The polynomial $f$ has a convenient matrix representation $M = \sum_{i=1}^n (a_i a_i^\top)^{\otimes 2}$, where $a_1,\ldots,a_n$ are the rows of the generator matrix $A$.
It turns out that the eigenvalues of this matrix indeed give information about the planted sparse vector $v_0$.
In particular, the vector $x_0\in \R^{d}$ with $A x_0 = v_0$ witnesses that $M$ has an eigenvalue of at least $1/k$ because $M$'s quadratic form with the vector $x_0^{\otimes 2}$ satisfies $\langle x_0^{\otimes 2}, M x_0^{\otimes 2}\rangle = \lVert  v_0 \rVert_4^4 \ge 1/k$.
If we let $M'$ be the corresponding matrix for the subspace $U$ without the planted sparse vector, $M'$ turns out to have only eigenvalues of at most $O(1/n)$ up to a single spurious eigenvalue
with eigenvector far from any vector of the form $x \otimes x$ \cite{DBLP:conf/stoc/BarakBHKSZ12}.

It follows that in order to distinguish between a random subspace with a planted sparse vector (\YES case) and a completely random subspace (\NO case), it is enough to compute the second-largest eigenvalue of a $d^2$-by-$d^2$ matrix (representing the $4$-norm polynomial over the subspace as above).
This decision version of the problem, while strictly speaking easier than the search version above, is at the heart of the matter: one can show that the large eigenvalue for the \YES case corresponds to an eigenvector which encodes the coefficients of the sparse planted vector in the basis.

\vspace{-4mm}
\paragraph{Improvements}
The best running time we can hope for with this basic approach is $O(d^4)$ (the size of the matrix).
Since we are interested in $d\le O(\sqrt n)$, the resulting running time $O(n d^2)$ would be subquadratic but still super-linear in the input size $n \cdot d$ (for representing a $d$-dimensional subspace of $\R^n$).
To speed things up, we use the partial trace approach outlined above.
We will begin by applying the partial trace approach naively, obtaining reasonable bounds, and then show that a small modification to the matrix before the partial trace operation allows us to achieve even smaller signal-to-noise ratios.

In the planted case, we may approximate $M \approx \tfrac{1}{k} (x_0 x_0^\top)^{\tensor 2} + Z$, where $x_0$ is the vector of coefficients of $v_0$ in the basis representation given by $A$ (so that $A x_0 = v_0$), and $Z$ is the noise matrix.
Since $\lVert  x_0 \rVert=1$, the partial trace operation preserves the projector $(x_0x_0^\top)^{\otimes 2}$ in the sense that $\Tr_{\R^d} (x_0x_0^\top)^{\otimes 2}= x_0 x_0^\top$.
Hence, with our heuristic approximation for $M$ above, we could show that the top eigenvector of $\Tr _{\R^d} M$ is close to $x_0$ by showing that the spectral norm bound $\lVert \Tr _{\R^d} Z \rVert \le o(1/k)$.

The partial trace of our matrix $M=\sum_{i=1}^n (a_ia_i^\top)\tensor (a_ia_i^\top)$ is easy to compute directly,
\begin{displaymath}
  N = \Tr_{\R^d} M = \sum_{i=1}^n \lVert  a_i \rVert_2^2 \cdot a_ia_i^\top\mper
\end{displaymath}
In the \YES case (random subspace with planted sparse vector), a direct computation shows that
\[
    \lambda_{\YES} \geq \langle  x_0, N x_0 \rangle\approx \tfrac d n \cdot \left(1 + \tfrac n d \lVert v_0 \rVert_4^4\right)\ge \tfrac d n \left( 1+ \tfrac n {d k} \right).
\]
Hence, a natural approach to distinguish between the \YES case and \NO case (completely random subspace) is to upper bound the spectral norm of $N$ in the \NO case.

In order to simplify the bound on the spectral norm of $N$ in the \NO case, suppose that the columns of $A$ are iid samples from the Gaussian distribution $\cN(0,\tfrac{1}{d}\Id)$ (rather than an orthogonal basis for the random subspace)--\pref{lem:fps-basis-change} establishes that this simplification is legitimate.
In this simplified setup, the matrix $N$ in the \NO case is the sum of $n$ iid matrices $\{\lVert  a_i \rVert^2\cdot a_i a_i^\top\}$, and we can upper bound its spectral norm $\lambda_{\NO}$ by $d/n  \cdot (1+ O(\sqrt { d /n}))$ using standard matrix concentration bounds.
Hence, using the spectral norm of $N$, we will  be able to distinguish between the \YES case and the \NO case as long as  $$\sqrt { d /n} \ll n/(dk) \implies \lambda_{\NO} \ll \lambda_{\YES}.$$
For linear sparsity $k=\e \cdot n$, this inequality is true so long as $d \ll (n/ \e^2)^{1/3}$, which is somewhat worse than the bound $\sqrt n$ bound on the dimension that we are aiming for.

Recall that $\Tr B = \sum_i \lambda_i(B)$ for a symmetric matrix $B$.
As discussed above, the partial trace approach works best when the noise behaves as the tensor of two Wigner matrices, in that there are cancellations when the eigenvalues of the noise are summed.
In our case, the noise terms $(a_i a_i^{\top}) \tensor (a_i a_i^{\top})$ do not have this property, as in fact $\Tr a_i a_i^\top = \|a_i\|^2 \approx d/n$.
Thus, in order to improve the dimension bound, we will center the eigenvalue distribution of the noise part of the matrix.
This will cause it to behave more like a Wigner matrix, in that the spectral norm of the noise will not increase after a partial trace.
Consider the partial trace of a matrix of the form $$M - \alpha\cdot  \Id \otimes \sum_i a_i a_i^\top \mcom$$ for some constant $\alpha>0$.
The partial trace of this matrix is $$N'=\sum_{i=1}^n (\lVert  a_i  \rVert_2^2 - \alpha )\cdot a_i a_i^\top\mper $$
We choose the constant $\alpha\approx d/n$ such that our matrix $N'$ has expectation $0$ in the \NO case, when the subspace is completely random.
In the \YES case, the Rayleigh quotient of $N'$ at $x_0$ simply shifts as compared to $N$, and we have $\lambda_{\YES} \ge \langle  x_0, N' x_0 \rangle \approx \lVert  v_0 \rVert_4^4\ge 1/k$ (see \pref{lem:fps-main-technical} and sublemmas).
On the other hand, in the \NO case, this centering operation causes significant cancellations in the eigenvalues of the partial trace matrix (instead of just shifting the eigenvalues).
In the \NO case, $N'$ has spectral norm $\lambda_{\NO} \le O(d/n^{3/2})$ for $d\ll \sqrt n$ (using standard matrix concentration bounds; again see \pref{lem:fps-main-technical} and sublemmas).
Therefore, the spectral norm of the matrix $N'$ allows us to distinguish between the \YES and \NO case as long as $d/n^{3/2}\ll 1/k$, which is satisfied as long as $k\ll n$ and $d\ll \sqrt n$.
We give the full formal argument in \pref{sec:fps}.

\subsection{Overcomplete tensor decomposition}
\label{sec:techniques-tensor-decomp}
Recall that in this problem we are given a $3$-tensor $T$ of the form $T=\sum_{i=1}^n a_i^{\otimes 3}\in\R^{d^3}$, where $a_1,\ldots,a_n \in \R^d$ are independent random vectors from $N(0,\tfrac 1d \Id)$.
The goal is to find a unit vector $a\in \R^d$ that is highly correlated with one\footnote{
We can then approximately recover all the components $a_1,\ldots,a_n$ by running independent trials of our randomized algorithm repeatedly on the same input.} of the vectors $a_1,\ldots,a_n$.

\vspace{-4mm}
\paragraph{Background}
The starting point of our algorithm is the polynomial $f(x)=\sum_{i=1}^n \langle  a_i,x \rangle^3$.
It turns out that for $n \ll d^{1.5}$ the (approximate) maximizers of this polynomial are close to the components $a_1,\ldots,a_n$, in the sense that $f(x) \approx 1$ if and only if $\max _{i\in [n]}\langle  a_i ,x \rangle^2\approx 1$.
Indeed, Ge and Ma \cite{DBLP:conf/approx/GeM15} show that the sum-of-squares method already captures this fact at degree 12, which implies a quasipolynomial time algorithm for this tensor decomposition problem via a general rounding result of Barak, Kelner, and Steurer \cite{DBLP:conf/stoc/BarakKS15}.

The simplest approach to this problem is to consider the tensor representation of the polynomial $\bT = \sum_{i\in[n]} a_i^{\tensor 3}$, and flatten it, hoping the singular vectors of the flattening are correlated with the $a_i$.
However, this approach is doomed to failure for two reasons: firstly, the simple flattenings of $\bT$ are $d^2 \times d$ matrices, and since $n \gg d$ the $a_i^{\tensor 2}$ collide in the column space, so that it is impossible to determine $\Span\{a_i^{\tensor 2}\}$.
Secondly, even for $n \le d$, because the $a_i$ are random vectors, their norms concentrate very closely about $1$.
This makes it difficult to distinguish any one particular $a_i$ even when the span is computable.

\vspace{-4mm}
\paragraph{Improvements}
We will try to circumvent both of these issues by going to higher dimensions.
Suppose, for example, that we had access to $\sum_{i\in[n]} a_i^{\tensor 4}$.\footnote{
As the problem is defined, we assume that we do not have access to this input, and in many machine learning applications this is a valid assumption, as gathering the data necessary to generate the $4$th order input tensor requires a prohibitively large number of samples.
}
The eigenvectors of the flattenings of this matrix are all within $\Span_{i\in[n]} \{a_i^{\tensor 2}\}$, addressing our first issue, leaving us only with the trouble of extracting individual $a_i^{\tensor 2}$ from their span.
If furthermore we had access to $\sum_{i\in[n]} a_i^{\tensor 6}$,  we could perform a partial random projection $(\Phi \tensor \Id \tensor \Id)\sum_{i\in[n]} a_i^{\tensor 6}$ where $\Phi \in \R^{d\times d}$ is a matrix with independent Gaussian entires, and then taking a partial trace, we end up with
$$
\Tr_{\R^d} \left((\Phi \tensor \Id \tensor \Id)\sum_{i\in[n]} a_i^{\tensor 6}\right)
= \sum_{i \in [n]} \iprod{\Phi, a_i^{\tensor 2}} a_i^{\tensor 4}\mper
$$
With reasonable probability (for exposition's sake, say with probability $1/n^{10}$), $\Phi$ is closer to some $a_i^{\tensor 2}$ than to all of the others so that $\iprod{\Phi, a_i^{\tensor 2}} \ge 100 \iprod{\Phi, a_j^{\tensor 2}}$ for all $j\in[n]$, and then $a_i^{\tensor 2}$ is distinguishable from the other vectors in the span of our matrix, taking care of the second issue .
As we show, a much smaller gap is sufficient to distinguish the top $a_i$ from the other $a_j$, and so the higher-probability event that $\Phi$ is only slightly closer to $a_i$ suffices (allowing us to recover all vectors at an additional runtime cost of a factor of $\tO(n)$).
This discussion ignores the presence of a single spurious large eigenvector, which we address in the technical sections.

Of course, we do not have access to the higher-order tensor $\sum_{i\in[n]} a_i^{\tensor 6}$.
Instead, we can obtain a noisy version of this tensor.
Our approach considers the following matrix representation of the polynomial $f^2$,
\begin{displaymath}
  M = \sum_{i,j} a_i a_j^\top \otimes (a_i a_i ^\top) \otimes (a_j a_j^\top ) \in \R^{d^3 \times d^3}\mper
\end{displaymath}
Alternatively, we can view this matrix as a particular flattening of the Kronecker-squared tensor $\bT^{\otimes 2}$.
It is instructive to decompose $M=\Mdiag + \Mcross$ into its diagonal terms $\Mdiag=\sum_i (a_ia_i^\top )^{\otimes 3}$ and its cross terms $\Mcross=\sum_{i\neq j} a_i a_j^\top \otimes (a_i a_i ^\top) \otimes (a_j a_j^\top )$.
The algorithm described above is already successful for $\Mdiag$; we need only control the eigenvalues of the partial trace of the ``noise'' component, $\Mcross$.
The main technical work will be to show that $\| \Tr_{\R^d}\Mdiag\|$ is small.
In fact, we will have to choose $\Phi$ from a somewhat different distribution---observing that $\Tr_{\R^d}(\Phi \tensor \Id \tensor \Id) = \sum_{i,j} \iprod{a_i, \Phi a_j} \cdot (a_i \tensor a_j)(a_i \tensor a_j)^\top$, we will sample $\Phi$ so that $\iprod{a_i, \Phi a_i} \gg \iprod{a_i, \Phi a_j}$.
We give a more detailed overview of this algorithm in the beginning of
\pref{sec:tdecomp},
explaining in more detail our choice of $\Phi$ and justifying heuristically the boundedness of the spectral norm of the noise.

\vspace{-4mm}
\paragraph{Connection to \sos analysis}
To explain how the above algorithm is a speedup of \sos, we give an overview of the \sos algorithm of \cite{DBLP:conf/approx/GeM15,DBLP:conf/stoc/BarakKS15}.
There, the degree-$t$ \sos SDP program is used to obtain an order-$t$ tensor $\chi_t$ (or a \emph{pseudodistribution}).
Informally speaking, we can understand $\chi_t$ as a proxy for $\sum_{i\in[n]} a_i^{\tensor t}$, so that $\chi_t = \sum_{i\in[n]} a_i^{\tensor t} + N$, where $N$ is a noise tensor.
While the precise form of $N$ is unclear, we know that $N$ must obey a set of constraints imposed by the \sos hierarchy at degree $t$.
For a formal discussion of pseudodistributions, see \cite{DBLP:conf/stoc/BarakKS15}.

To extract a single component $a_i$ from the tensor $\sum_{i\in[n]} a_i^{\tensor t}$,  there are many algorithms which would work (for example, the algorithm we described for $\Mdiag$ above).
However, any algorithm extracting an $a_i$ from $\chi_t$ must be robust to the noise tensor $N$.
For this it turns out the following algorithm will do: suppose we have the tensor $\sum_{i\in[n]} a_i^{\tensor t}$, taking $t = O(\log n)$.
Sample $g_1,\ldots,g_{\log(n) - 2}$ random unit vectors, and compute the matrix $M = \sum_i (\prod_{1 \leq j \leq \log(n) - 2} \iprod{g_j, a_i})\cdot a_i a_i^\top$.
If we are lucky enough, there is some $a_i$ so that every $g_j$ is a bit closer to $a_i$ than any other $a_{i'}$, and $M = a_i a_i^\top +E$ for some $\|E\| \ll 1$.
The proof that $\|E\|$ is small can be made so simple that it applies also to the SDP-produced proxy tensor $\chi_{\log n}$, and so this algorithm is robust to the noise $N$.
This last step is very general and can handle tensors whose components $a_i$ are less well-behaved than the random vectors we consider, and also more overcomplete, handling tensors of rank up to $n =  \tilde\Omega(d^{1.5})$.\footnote{
    It is an interesting open question whether taking $t = O(\log n)$ is really necessary, or whether this heavy computational requirement is simply an artifact of the \sos proof.
}

Our subquadratic-time algorithm can be viewed as a low-degree, spectral analogue of the \cite{DBLP:conf/stoc/BarakKS15} \sos algorithm.
However, rather than relying on an SDP to produce an object close to $\sum_{i\in[n]} a_i^{\tensor t}$, we manufacture one ourselves by taking the Kronecker square of our input tensor.
We explicitly know the form of the deviation of $\bT^{\tensor 2}$ from $\sum_{i\in[n]} a_i^{\tensor 6}$, unlike in \cite{DBLP:conf/stoc/BarakKS15}, where the deviation of the SDP certificate $\chi_{t}$ from $\sum_{i\in[n]} a_i^{\tensor t}$ is poorly understood.
We are thus able to control this deviation (or ``noise'') in a less computationally intensive way, by cleverly designing a partial trace operation which decreases the spectral norm of the deviation.
Since the tensor handled by the algorithm is much smaller---order $6$ rather than order $\log n$---this provides the desired speedup.

\subsection{Tensor principal component analysis}
\label{sec:techniques-tpca}
Recall that in this problem we are given a tensor $\bT = \tau \cdot v^{\tensor 3} + \bA$, where $v \in \R^d$ is a unit vector, $\bA$ has iid entries from $\cN(0,1)$, and $\tau > 0$ is the signal-to-noise ratio.
The aim is to recover $v$ approximately.

\vspace{-4mm}
\paragraph{Background and \sos analysis}
A previous application of \sos techniques to this problem discussed several \sos or spectral algorithms,
including one that runs in quasi-linear time \cite{DBLP:conf/colt/HopkinsSS15}.
Here we apply the partial trace method to a subquadratic spectral \sos algorithm discussed in $\cite{DBLP:conf/colt/HopkinsSS15}$
to achieve nearly the same signal-to-noise guarantee in only linear time.

Our starting point is the polynomial $\bT(x) = \tau \cdot \iprod{v,x}^3 + \iprod{x^{\tensor 3}, \bA}$.
The maximizer of $\bT(x)$ over the sphere is close to the vector $v$ so long as $\tau \gg \sqrt{n}$ \cite{DBLP:conf/nips/RichardM14}.
In \cite{DBLP:conf/colt/HopkinsSS15}, it was shown that degree-4 \sos maximizing this polynomial can recover $v$ with a signal-to-noise ratio of at least $\tOmega(n^{3/4})$, since there exists a suitable \sos bound on the noise term $\iprod{x^{\tensor 3}, \bA}$.

Specifically, let $A_i$ be the $i$th slice of $\bA$, so that $\iprod{x, A_i x}$ is the quadratic form $\sum_{j,k} \bA_{ijk} x_j x_k$.
Then there is a \sos proof that $\bT(x)$ is bounded by $|\bT(x) - \tau \cdot \iprod{v,x}^3| \le f(x)^{1/2}\cdot \|x\|$, where $f(x)$ is the degree-4 polynomial $f(x) = \sum_i \iprod{x, A_i x}^2$.
The polynomial $f$ has a convenient matrix representation: $f(x) = \iprod{x^{\tensor 2}, (\sum_i A_i \tensor A_i)x^{\tensor 2}}$:
since this matrix is a sum of iid random matrices $A_i \tensor A_i$, it is easy to show that this matrix spectrally concentrates to its expectation.
So with high probability one can show that the eigenvalues of $\sum_i A_i \tensor A_i$ are at most $\approx d^{3/2} \log(d)^{1/2}$ (except for a single spurious eigenvector), and it follows that degree-$4$ SoS solves tensor PCA so long as $\tau \gg d^{3/4} \log(d)^{1/4}$.

This leads the authors to consider a slight modification of $f(x)$, given by $g(x) = \sum_i \iprod{x, T_i x}^2$, where $T_i$ is the $i$th slice of $\bT$.
Like $\bT$, the function $g$ also contains information about $v$, and the \sos bound on the noise term in $\bT$ carries over as an analogous bound on the noise in $g$.
In particular, expanding $T_i \tensor T_i$ and ignoring some negligible cross-terms yields
\[
  \sum_i T_i \tensor T_i \approx \tau^2 \cdot (v \tensor v)(v \tensor v)^\top + \sum_i A_i \tensor A_i
\]
Using $v \tensor v$ as a test vector, the quadratic form of the latter matrix can be made at least $\tau^2 - O(d^{3/2} \log(d)^{1/2})$.
Together with the boundedness of the eigenvalues of $\sum_i A_i \tensor A_i$ this shows that when $\tau \gg d^{3/4} \log(d)^{1/4}$ there is a spectral algorithm to recover $v$.
Since the matrix $\sum_i T_i \tensor T_i$ is $d^2 \times d^2$, computing the top eigenvector requires $\tO(d^4 \log n)$ time, and by comparison to the input size $d^3$ the algorithm runs in subquadratic time.
\vspace{-4mm}
\paragraph{Improvements}
In this work we speed this up to a linear time algorithm via the partial trace approach.
As we have seen, the heart of the matter is to show that taking the partial trace of $\tau^2 \cdot (v \tensor v)(v \tensor v)^\top +\sum_i A_i \tensor A_i$ does not increase the spectral noise.
That is, we require that
\[
  \Norm{\Tr_{\R^d} \sum_i A_i \tensor A_i} = \Norm{\sum_i \Tr(A_i) \cdot A_i} \leq O(d^{3/2} \log(d)^{1/2})\mper
\]
Notice that the $A_i$ are essentially Wigner matrices, and so it is roughly true that $|\Tr(A_i)| \approx \|A_i\|$, and the situation is very similar to our toy example of the application of partial traces in \pref{sec:techniques}.

Heuristically, because $\sum_{i\in[n]} A_i \tensor A_i$ and $\sum_{i\in[n]} \Tr(A_i) \cdot A_i$ are random matrices, we expect that their eigenvalues are all of roughly the same magnitude.
This means that their spectral norm should be close to their Frobenius norm divided by the square root of the dimension, since for a matrix $M$ with eigenvalues $\lambda_1,\ldots,\lambda_n$, $\|M\|_F = \sqrt{\sum_{i\in[n]} \lambda_i^2}$.
By estimating the sum of the squared entries, we expect that the Frobenious norm of $\sum_i \Tr(A_i) \cdot A_i$ is less than that of $\sum_i A_i \tensor A_i$ by a factor of $\sqrt d$ after the partial trace, while the dimension decreases by a factor of $d$, and so assuming that the eigenvalues are all of the same order, a typical eigenvalue should remain unchanged.
We formalize these heuristic calculations using standard matrix concentration arguments in \pref{sec:tpca}.

\section{Preliminaries}

\paragraph{Linear algebra}
We will work in the real vector spaces given by $\R^{n}$.
A vector of indeterminates may be denoted $x = (x_1, \dots, x_n)$, although
we may sometimes switch to parenthetical notation for indexing, i.e.
$x = (x(1), \dots, x(n))$ when subscripts are already in use.
We denote by $[n]$ the set of all valid indices for a vector in $\R^n$.
Let $e_i$ be the $i$th canonical basis vector so that $e_i(i) = 1$ and
$e_i(j) = 0$ for $j \ne i$.

For a vectors space $V$, we may denote by
$\cL(V)$ the space of linear operators from $V$ to $V$.
The space orthogonal to a vector $v$ is denoted $v^{\perp}$.

For a matrix $M$, we use $M^{-1}$ to denote its inverse or its
Moore-Penrose pseudoinverse; which one it is will be clear from context.
For $M$ PSD, we write $M^{-1/2}$ for the unique PSD matrix with
$(M^{-1/2})^2 = M^{-1}$.

\paragraph{Norms and inner products}
We denote the usual entrywise inner product by $\iprod{{\cdot}, {\cdot}}$, so that $\iprod{u, v} = \sum_{i \in [n]} u_iv_i$ for $u,v \in \R^n$.
The $\ell_p$-norm of a vector $v \in \R^n$ is given by $\|v\|_p = \cramped{\paren{\sum_{i \in [n]} \cramped{{v_i}^p}}^{1/p}}$, with $\|v\|$ denoting the $\ell_2$-norm by default.
The matrix norm used throughout the paper will be the operator / spectral norm,
denoted by $\lVert  M \rVert = \|M\|_{op} \seteq \max_{x\neq 0} \lVert  M x \rVert / \lVert  x \rVert$.

\paragraph{Tensor manipulation}
Boldface variables will reserved for tensors
$\bT \in \R^{n \times n \times n}$, of which we consider only order-3 tensors.
We denote by $\bT(x,y,z)$ the multilinear function in $x,y,z \in \R^n$
such that $\bT(x,y,z) = \sum_{i,j,k \in [n]} T_{i,j,k} x_i y_j z_k$,
applying $x$, $y$, and $z$ to the first, second, and third modes of the tensor
$\bT$ respectively.
If the arguments are matrices $P$, $Q$, and $R$ instead, this lifts
$\bT(P,Q,R)$ to the unique multilinear tensor-valued function such that
$[\bT(P,Q,R)](x,y,z) = \bT(Px, Qy, Rz)$ for all vectors $x, y, z$.

Tensors may be flattened to matrices in the multilinear way such that
for every $u\in \R^{n \times n}$ and $v \in \R^n$, the tensor $u \otimes v$
flattens to the matrix $uv^\top \in \R^{n^2 \times n}$ with $u$ considered as a vector.
There are 3 different ways to flatten a $3$-tensor $\bT$, corresponding
to the 3 modes of $\bT$.
Flattening may be understood as reinterpreting the indices of a tensor
when the tensor is expressed as an 3-dimensional array of numbers.
The expression $v^{\otimes 3}$ refers to $v \otimes v \otimes v$ for a
vector $v$.

\paragraph{Probability and asymptotic bounds}
We will often refer to collections of independent and identically distributed
(or \emph{iid}) random variables.
The Gaussian distribution with mean $\mu$ and variance $\sigma^2$
is denoted $\cN(\mu, \sigma^2)$.
Sometimes we state that an event happens with overwhelming probability.
This means that its probability is at least $1-n^{-\omega(1)}$.
A function is $\tilde O(g(n))$ if it is $O(g(n))$ up to polylogarithmic factors.

\section{Planted sparse vector in random linear subspace}
\label{sec:fps}

In this section we give a nearly-linear-time algorithm to recover a sparse vector planted in a random subspace.
\begin{problem}\label{prob:planted}
    Let $v_0 \in \R^n$ be a unit vector such that $\|v_0\|_4^4 \ge \frac{1}{\epsilon n}$.
    Let $v_1,\ldots,v_{d-1} \in \R^n$ be iid from $\cN(0,\tfrac 1 n \Id_n)$.
    Let $w_0,\ldots,w_{d-1}$ be an orthogonal basis for $\Span \{v_0,\ldots,v_{d-1} \}$.
    { \bf Given: } $w_0,\ldots,w_{d-1}$
    { \bf Find: } a vector $v$ such that $\iprod{v, v_0}^2 \ge 1 - o(1)$.
\end{problem}

\begin{center}
\fbox{\begin{minipage}{\textwidth}
\begin{center}\textbf{Sparse Vector Recovery in Nearly-Linear Time}\end{center}
\begin{algo}
\label{alg:fast-planted}
Input: $w_0,\ldots,w_{d-1}$ as in \pref{prob:planted}.
Goal: Find $v$ with $\iprod{\hat v, v_0}^2 \geq 1 - o(1)$.
\item
\begin{itemize}
  \item Compute leverage scores $\|a_1\|^2,\ldots,\|a_n\|^2$, where $a_i$ is the $i$th row of the $n \times d$ matrix $S \seteq \Paren{\begin{array}{ccc} w_0 & \cdots & w_{d-1}\end{array}}$.
  \item Compute the top eigenvector $u$ of the matrix
  \[
    A \defeq \sum_{i \in [n]} (\|a_i\|_2^2 - \cconst) \cdot a_i a_i^\top\mper
  \]
\item
  Output $Su$.
\end{itemize}
\end{algo}
\end{minipage}}
\end{center}

\begin{remark}[Implementation of \pref{alg:fast-planted} in nearly-linear time]
  The leverage scores $\|a_1\|^2,\ldots,\|a_n\|^2$ are clearly computable in time $O(nd)$.
  In the course of proving correctness of the algorithm we will show that $A$ has constant spectral gap, so by a standard analysis $O(\log d)$ matrix-vector multiplies suffice to recover its top eigenvector.
  A single matrix-vector multiply $Ax$ requires computing $c_i \seteq (\|a_i\|^2 - \cconst) \iprod{a_i,x}$ for each $i$ (in time $O(nd)$) and summing $\sum_{i \in[n]} c_i x_i$ (in time $O(nd)$).
  Finally, computing $Su$ requires summing $d$ vectors of dimension $n$, again taking time $O(nd)$.
\end{remark}

The following theorem expresses correctness of the algorithm.
\begin{theorem}
  \label{thm:fps-correctness}
  Let $v_0 \in \R^n$ be a unit vector with $\|v_0\|_4^4 \geq \tfrac 1 {\epsilon n}$.
  Let $v_1,\ldots,v_{d-1} \in \R^n$ be iid from $\cN(0,\tfrac 1 n \Id_n)$.
  Let $w_0,\ldots,w_{d-1}$ be an orthogonal basis for $\Span \{v_0,\ldots,v_{d-1} \}$.
  Let $a_i$ be the $i$-th row of the $n \times d$ matrix $S \seteq \Paren{\begin{array}{ccc} w_0 & \cdots & w_{d-1}\end{array}}$.

  When $d \leq n^{1/2}/\polylog(n)$, for any sparsity $\epsilon > 0$,
  \wovp the top eigenvector $u$ of $\sum_{i = 1}^n (\|a_i\|^2 - \tfrac d n) \cdot a_i a_i^\top$ has
  $\iprod{Su,v_0}^2 \geq 1 - O(\epsilon^{1/4}) - o(1)$.
\end{theorem}

We have little control over the basis vectors the algorithm is given.
However, there is a particularly nice (albeit non-orthogonal) basis for the subspace which exposes the underlying randomness.
Suppose that we are given the basis vectors $v_0,\ldots,v_d$, where $v_0$ is the sparse vector normalized so that $\|v_0\| = 1$, and $v_1,\ldots,v_{d-1}$ are iid samples from $\cN(0,\tfrac{1}{n}\Id_n)$.
The following lemma shows that if the algorithm had been handed this good representation of the basis rather than an arbitrary orthogonal one, its output would be the correlated to the vector of coefficients giving of the planted sparse vector (in this case the standard basis vector $e_1$).

\begin{lemma}
  \label{lem:fps-main-technical}
  Let $v_0 \in \R^n$ be a unit vector.
  Let $v_1,\ldots,v_{d-1} \in \R^n$ be iid from $\cN(0,\tfrac 1 n \Id)$.
  Let $a_i$ be the $i$th row of the $n \times d$ matrix $S \seteq \Paren{\begin{array}{ccc} v_0 & \cdots & v_{d-1}\end{array}}$.
  Then there is a universal constant $\epsilon^* > 0$ so that for any $\epsilon \leq \epsilon^*$, so long as $d \leq n^{1/2} / \polylog(n)$, \wovp
  \[
    \sum_{i = 1}^n (\|a_i\|^2 - \cconst) \cdot a_i a_i^\top = \|v_0\|_4^4 \cdot e_1 e_1^\top + M
  \]
  where $e_1$ is the first standard basis vector and $\|M\| \leq O(\|v_0\|_4^3 \cdot n^{-1/4} + \|v_0\|_4^2 \cdot n^{-1/2} + \|v_0\|_4 \cdot n^{-3/4} + n^{-1})$.
\end{lemma}

The second ingredient we need is that the algorithm is robust to exchanging this good basis for an arbitrary orthogonal basis.

\begin{lemma}
  \label{lem:fps-basis-change}
  Let $v_0 \in \R^n$ have $\|v_0\|_4^4 \geq \tfrac 1 {\epsilon n}$.
  Let $v_1,\ldots,v_{d-1} \in \R^n$ be iid from $\cN(0,\tfrac 1 n \Id_n)$.
  Let $w_0,\ldots,w_{d-1}$ be an orthogonal basis for $\Span \{v_0,\ldots,v_{d-1} \}$.
  Let $a_i$ be the $i$th row of the $n \times d$ matrix $S \seteq \Paren{\begin{array}{ccc} v_0 & \cdots & v_{d-1}\end{array}}$.
  Let $a_i'$ be the $i$th row of the $n \times d$ matrix $S' \seteq \Paren{\begin{array}{ccc} w_0 & \cdots & w_{d-1}\end{array}}$.
  Let $A \seteq \sum_i a_i a_i^\top$.
  Let $Q \in \R^{d \times d}$ be the orthogonal matrix so that $S A^{-1/2} = S' Q$, which exists since $S A^{-1/2}$ is orthogonal, and which has the effect
  that $a_i' = Q A^{-1/2} a_i$.
  Then when $d \leq n^{1/2} / \polylog(n)$, \wovp
  \[
    \left \| \sum_{i = 1}^n (\|a_i'\|^2 - \tfrac d n)\cdot a_i' a_i'^\top
    - Q \Paren{\sum_{i = 1}^n (\|a_i\|^2 - \tfrac d n) \cdot a_i a_i^\top} Q^\top \right \|
    \leq O\Paren{\frac 1 n} + o(\|v\|_4^4)
  \]
\end{lemma}

Last, we will need the following fact, which follows from standard concentration.
The proof is in \pref{sec:fast-planted-proofs}.
\begin{lemma}
  \label{lem:fps-covariance-concentration}
  Let $v \in \R^n$ be a unit vector.
  Let $b_1,\ldots,b_n \in \R^{d-1}$ be iid from $\cN(0,\frac 1 n \Id_{d-1})$.
  Let $a_i \in \R^d$ be given by $a_i \seteq (v(i) \, \, b_i)$.
  Then \wovp $\|\sum_{i=1}^n a_i a_i^\top - \Id_d \| \leq \tO(d/n)^{1/2}$.
  In particular, when $d = o(n)$, this implies that \wovp $\| (\sum_{i=1}^n a_i a_i^\top)^{-1} - \Id_d \| \leq \tO(d/n)^{1/2}$
  and $\| (\sum_{i=1}^n a_i a_i^\top)^{-1/2} - \Id_d \| \leq \tO(d/n)^{1/2}$.
\end{lemma}

We are ready to prove \pref{thm:fps-correctness}.

\begin{proof}[Proof of \pref{thm:fps-correctness}]
  Let $b_1,\ldots,b_n$ be the rows of the matrix $S' \seteq \Paren{\begin{array}{ccc} v_0 & \cdots & v_{d-1}\end{array}}$.
  Let $B = \sum_i b_i b_i^\top$.
  Note that $S' B^{-1/2}$ has columns which are an orthogonal basis for $\Span \{ w_0,\ldots,w_{d-1} \}$.
  Let $Q \in \R^{d \times d}$ be the rotation so that $S' B^{-1/2} = S Q$.

  By \pref{lem:fps-main-technical} and \pref{lem:fps-basis-change}, we can write the matrix $A = \sum_{i=1}^n (\|a_i\|_2^2 - \cconst) \cdot a_i a_i^\top$
  as
  \[
    A = \|v_0\|_4^4 \cdot Q e_1 e_1^\top Q^\top + M
  \]
  where \wovp
  \[
    \|M\| \leq O(\|v_0\|_4^3 \cdot n^{-1/4} + \|v_0\|_4^2 \cdot n^{-1/2} + \|v_0\|_4 \cdot n^{-3/4} + n^{-1}) + o(\|v\|_4^4)\mper
  \]
  We have assumed that $\|v_0\|_4^4 \geq (\epsilon n)^{-1}$, and so since
  $A$ is an almost-rank-one matrix (\pref{lem:low-correlation}),
  the top eigenvector $u$ of $A$ has $\iprod{u,Qe_1}^2 \geq 1 - O(\epsilon^{1/4})$,
  so that $\iprod{Su,SQe_1}^2 \geq 1 - O(\epsilon^{1/4})$ by column-orthogonality of $S$.

  At the same time, $S Q e_1 = S' B^{-1/2} e_1$, and by \pref{lem:fps-covariance-concentration},
  $\| B^{-1/2} - \Id\| \leq \tO(d/n)^{1/2}$ \wovp,
  so that $\iprod{S u, S' e_1}^2 \ge \iprod{S u, S Q e_1}^2 - o(1) $.
  Finally, $S'e_1 = v_0$ by definition, so
  $\iprod{S u, v_0}^2 \ge 1 - O(\epsilon^{1/4}) - o(1)$.
\end{proof}

\subsection{Algorithm succeeds on good basis}
We now prove \pref{lem:fps-main-technical}.
We decompose the matrix in question into a contribution from $\|v_0\|_4^4$ and the rest:
explicitly, the decomposition is $\sum (\|a_i\|_2^2 - \tfrac{d}{n}) \cdot a_ia_i^\top =
\sum v(i)^2 \cdot a_ia_i^\top + \sum (\|b_i\|_2^2 - \tfrac{d}{n} \cdot a_ia_i^\top)$.
This first lemma handles the contribution from $\|v_0\|_4^4$.
\begin{lemma}
  \label{lem:fps-good-basis-1}
  Let $v \in \R^n$ be a unit vector.
  Let $b_1,\ldots,b_n \in \R^{d-1}$ be random vectors iid from $\cN(0, \tfrac 1 n \cdot \Id_{d-1})$.
  Let $a_i = (v(i) \; \, b_i) \in \R^d$.
  Suppose $d \leq n^{1/2} / \polylog(n)$.
  Then
  \[
    \sum_{i=1}^n v(i)^2 \cdot a_i a_i^\top = \|v\|_4^4 \cdot e_1 e_1^\top + M'
  \]
  where $\|M'\| \leq O(\|v\|_4^3 n^{-1/4} + \|v\|_4^2 n^{-1/2})$ \wovp.
\end{lemma}

\begin{proof}[Proof of \pref{lem:fps-good-basis-1}]
  We first show an operator-norm bound on the principal submatrix $\sum_{i=1}^n v(i)^2 \cdot b_i b_i^\top$ using the truncated matrix Bernstein inequality \pref{prop:truncated-bernstein}.
  First, the expected operator norm of each summand is bounded:
  \[
    \E v(i)^2 \|b_i\|_2^2 \leq (\max_{j} v(j)^2) \cdot O\Paren{\frac d n} \leq \|v\|_4^2 \cdot O\Paren{\frac d n}\mper
  \]
  The operator norms are bounded by constant-degree polynomials in Gaussian variables, so \pref{lem:truncate} applies to truncate their tails in preparation for application of a Bernstein bound.
  We just have to calculate the variance of the sum, which is at most
  \[
    \left \| \E \sum_{i=1}^n v(i)^4 \|b_i\|_2^2 \cdot b_i b_i^\top \right \| = \|v\|_4^4 \cdot O\Paren{\frac d {n^2}}\mper
  \]
  The expectation $\E \sum_{i=1}^n v(i)^2 \cdot b_i b_i^\top$ is $\tfrac{\|v\|^2}{n} \cdot \Id$.
  Applying a matrix Bernstein bound (\pref{prop:truncated-bernstein}) to the deviation from expectation, we get that \wovp,
  \[
    \left \| \Paren{\sum_{i=1}^n v(i)^2 \cdot b_i b_i^\top} - \frac 1 n \cdot \Id \right \| \leq \|v\|_4^2 \cdot \tO\Paren{\frac d n} \leq O(\|v\|_4^2 n^{-1/2})
  \]
  for appropriate choice of $d \leq n^{-1/2}/\polylog(n)$.
  Hence, by triangle inequality, $\| \sum_{i=1}^n v(i)^2 \cdot b_i b_i^\top \| \leq \|v\|_4^2 n^{-1/2}$ \wovp.

  Using a Cauchy-Schwarz-style inequality (\pref{lem:op-cs-block}) we now show that the bound on this principal submatrix is essentially enough to obtain the lemma.
  Let $p_i, q_i \in \R^d$ be given by
  \[
    p_i \defeq v_0(i) \cdot \Paren{\begin{array}{c} v_0(i) \\ 0 \\ \vdots \\ 0\end{array}}
    \qquad q_i \defeq v_0(i) \cdot \Paren{\begin{array}{c} 0 \\ b_i \end{array}}\mper
  \]
  Then
  \[
    \sum_{i=1}^n v(i)^2 \cdot b_i b_i^\top = \|v\|_4^4 + \sum_{i = 1}^n p_i q_i^\top + q_i p_i^\top + q_i q_i^\top\mper
  \]
  We have already bounded $\sum_{i=1}^n q_i q_i^\top = \sum_{i=1}^n v(i)^2 \cdot b_i b_i^\top$.
  At the same time, $\| \sum_{i=1}^n p_i p_i^\top \| = \|v\|_4^4$.
  By \pref{lem:op-cs-block}, then,
  \[
    \left \| \sum_{i=1}^n p_i q_i^\top + q_i p_i^\top \right \| \leq O(\|v\|_4^3 n^{-1/4})
  \]
  \wovp.
  A final application of triangle inquality gives the lemma.
\end{proof}

Our second lemma controls the contribution from the random part of the leverage scores.
\begin{lemma}
  \label{lem:fps-good-basis-2}
  Let $v \in \R^n$ be a unit vector.
  Let $b_1,\ldots,b_n \in \R^{d-1}$ be random vectors iid from $\cN(0, \tfrac 1 n \cdot \Id_{d-1})$.
  Let $a_i = (v(i) \; \, b_i) \in \R^d$.
  Suppose $d \leq n^{1/2} / \polylog(n)$.
  Then \wovp
  \[
    \left \| \sum_{i=1}^n (\|b_i\|_2^2 - \cconst) \cdot a_i a_i^\top \right \| \leq \|v\|_4^2 \cdot O(n^{-3/4}) + \|v\|_4 \cdot O(n^{-1}) + O(n^{-1})\mper
  \]
\end{lemma}
\begin{proof}
  Like in the proof of \pref{lem:fps-good-basis-1}, $\sum_{i=1}^n (\|b_i\|_2^2 - \cconst) \cdot a_i a_i^\top$ decomposes into a convenient block structure; we will bound each block separately.
  \begin{align}
    \sum_{i=1}^n (\|b_i\|_2^2 - \cconst) \cdot a_i a_i^\top = \sum_{i=1}^n (\|b_i\|_2^2 - \cconst) \cdot \left (
      \begin{array}{cc}
        v(i)^2 & v(i) \cdot b_i^\top\\
        v(i) \cdot b_i & b_i b_i^\top
      \end{array}
      \right )\mper \label{eqn:fps-block}
    \end{align}
  In each block we can apply a (truncated) Bernstein inequality.
  For the large block $\sum_{i=1}^n (\|b_i\|_2^2 - \cconst) b_i b_i^\top$, the choice $\cconst$ ensures that $\E (\|b_i\|_2^2 - \cconst) b_i b_i^\top = O(\tfrac 1 {n^2}) \cdot \Id$.
  The expected operator norm of each summand is small:
  \begin{align*}
    \E \|(\|b_i\|_2^2 - \cconst) b_i b_i^\top\| & = \E |(\|b_i\|_2^2 - \cconst)| \|b_i\|_2^2\\
                                                & \leq (\E (\|b_i\|_2^2 - \cconst)^2)^{1/2} (\E \|b_i\|_2^4)^{1/2} \quad \text{by Cauchy-Schwarz}\\
                                                & \leq O\Paren{\frac{d^{1/2}}{n}} \cdot O\Paren{\frac d n} \quad \text{variance of $\chi^2$ with $k$ degrees of freedom is $O(k)$}\\
                                                & = O\Paren{\frac{d^{3/2}}{n^2}}\mper
  \end{align*}
  The termwise operator norms are bounded by constant-degree polynomials in Gaussian variables, so \pref{lem:truncate} applies to truncate the tails of the summands in preparation for a Bernstein bound.
  We just have to compute the variance of the sum, which is small because we have centered the coefficients:
  \begin{align*}
    & \left \| \sum_i \E (\|b_i\|_2^2 - \cconst)^2 \|b_i\|_2^2 \cdot b_i b_i^\top \right \| \leq O\Paren{\frac{d^2}{n^3}}
  \end{align*}
  by direct computation of $\E (\|b_i\|_2^2 - \cconst)^2 \|b_i\|_2^2 b_i b_i^\top$ using \pref{fact:gaussian-poly}.
  These facts together are enough to apply the matrix Bernstein inequality (\pref{prop:truncated-bernstein}) and conclude that \wovp
  \[
    \left \| \sum_{i=1}^n (\|b_i\|_2^2 - \cconst) \cdot b_i b_i^\top \right \| \leq \tO\Paren{\frac d {n^{3/2}}} \leq O\Paren{\frac 1 n}
  \]
  for appropriate choice of $d \leq n/\polylog(n)$.

  We turn to the other blocks from \pref{eqn:fps-block}.
  The upper-left block contains just the scalar $\sum_{i=1}^n (\|b_i\|_2^2 - \cconst) v(i)^2$.
  By standard concentration each term is bounded: \wovp,
  \[
    (\|b_i\|_2^2 - \cconst) v(i)^2 \leq (\max_i v(i)^2) \cdot \tO\Paren{\frac{d^{1/2}}{n}} \leq \|v\|_4^2 \cdot \tO\Paren{\frac{d^{1/2}}{n}}\mper
  \]
  The sum has variance at most $\sum_{i=1}^n v(i)^4 \E (\| b_i\|_2^2 - \cconst)^2 \leq \|v\|_4^4 \cdot O(d/n^2)$.
  Again using \pref{lem:truncate} and \pref{prop:truncated-bernstein}, we get that \wovp
  \[
    \left | \sum_{i=1}^n (\|b_i\|_2^2 - \cconst) v(i)^2 \right | \leq \|v\|_4^2 \cdot \tO\Paren{\frac{d^{1/2}} n}\mper
  \]

  It remains just to address the block $\sum_{i=1}^n (\|b_i\|_2^2 - \cconst) v(i) \cdot b_i$.
  Each term in the sum has expected operator norm at most
  \[
    (\max_i v(i)^2)^{1/2} \cdot O\Paren{\frac{d}{n^{3/2}}} \leq \|v\|_4 \cdot O\Paren{\frac{d}{n^{3/2}}} \cdot\mcom
  \]
  and once again the since the summands' operator norms are bounded by constant-degree polynomials of Gaussian variables \pref{lem:truncate} applies to truncate their tails in preparation to apply a Bernstein bound.
  The variance of the sum is at most $\|v\|_2^2 \cdot O(d^2/n^3)$, again by \pref{fact:gaussian-poly}.
  Finally, \pref{lem:truncate} and \pref{prop:truncated-bernstein} apply to give that \wovp
  \[
    \left \| \sum_{i=1}^n (\|b_i\|_2^2 - \cconst) v(i) \cdot b_i \right \| \leq  \|v\|_4 \cdot \tO\Paren{\frac {d}{n^{3/2}}} + \tO\Paren{\frac d {n^{3/2}}} = \|v\|_4 \cdot n^{-1} + n^{-1}
  \]
  for appropriate choice of $d \leq n^{1/2} / \polylog(n)$.
  Putting it all together gives the lemma.
\end{proof}

We are now ready to prove \pref{lem:fps-main-technical}
\begin{proof}[Proof of \pref{lem:fps-main-technical}]
  We decompose $\|a_i\|_2^2 = v_0(i)^2 + \|b_i\|_2^2$ and use \pref{lem:fps-good-basis-1} and \pref{lem:fps-good-basis-2}.
  \begin{align*}
    \sum_{i=1}^n (\|a_i\|_2^2 - \cconst) \cdot a_i a_i^\top & = \Paren{ \sum_{i=1}^n v_0(i)^2 \cdot a_i a_i^\top } + \Paren{ \sum_{i=1}^n (\|b_i\|_2^2 - \cconst) \cdot a_i a_i^\top}\\
                                                            & = \|v_0\|_4^4 \cdot e_1 e_1^\top + M\mcom
  \end{align*}
  where
  \[
    \|M\| \leq O(\|v_0\|_4^3 \cdot n^{-1/4} + \|v_0\|_4^2 \cdot n^{-1/2}) + O(\|v_0\|_4 \cdot n^{-1} + n^{-1})\mper
  \]
  Since $\|v_0\|_4^4 \geq (\epsilon n)^{-1}$, we get $\|v_0\|_4^4 / \|M\| \geq \frac 1 {\epsilon^{1/4}}$, completing the proof.
\end{proof}

\subsection{Closeness of input basis and good basis}
We turn now to the proof of \pref{lem:fps-basis-change}.
We recall the setting.
We have two matrices: $M$, which the algorithm computes, and $M'$, which is induced by a basis for the subspace which reveals the underlying randomness and which we prefer for the analysis.
$M'$ differs from $M$ by a rotation and a basis orthogonalization step (the good basis is only almost orthogonal).
The rotation is easily handled.
The following lemma gives the critical fact about the orthogonalization step: orthogonalizing does not change the leverage scores too much.
\footnote{Strictly speaking the good basis does not have leverage scores since it is not orthogonal, but we can still talk about the norms of the rows of the matrix whose columns are the basis vectors.}
\begin{lemma}[Restatement of \pref{lem:ortho-4}]
\label{lem:ortho-4-restatement}
Let $v \in \R^n$ be a unit vector and let $b_1,\ldots,b_n \in \R^{d-1}$ be iid from $\cN(0,\frac 1 n \Id_{d-1})$.
  Let $a_i \in \R^d$ be given by $a_i \seteq (v(i) \, \, b_i)$.
  Let $A \seteq \sum_i a_i a_i^\top$.
  Let $c \in \R^{d-1}$ be given by $c \seteq \sum_i v(i) b_i$.
  Then for every index $i \in [n]$, \wovp,
  \[
    \left | \|A^{-1/2} a_i \|^2 - \|a_i\|^2 \right | \leq \tO\Paren{\frac{ d + \sqrt{n}}{n}} \cdot \|a_i\|^2\mper
  \]
\end{lemma}
The proof again uses standard concentration and matrix inversion formulas, and can be found in \pref{sec:fast-planted-proofs}.
We are ready to prove \pref{lem:fps-basis-change}.

\begin{proof}[Proof of \pref{lem:fps-basis-change}]
  The statement we want to show is
  \[
    \left \| \sum_{i = 1}^n (\|a_i'\|^2 - \tfrac d n)\cdot a_i' a_i'^\top
    - Q \Paren{\sum_{i = 1}^n (\|a_i\|^2 - \tfrac d n) \cdot a_i a_i^\top} Q^\top \right \|
    \leq O\Paren{\frac 1 n} + o(\|v\|_4^4) \mper
  \]
  Conjugating by $Q$ and multiplying by $-1$ does not change the operator norm,
  so that this is equivalent to
  \[
    \left \| \sum_{i = 1}^n (\|a_i\|^2 - \tfrac d n)\cdot a_i a_i^\top
    - Q^\top \Paren{\sum_{i = 1}^n (\|a_i'\|^2 - \tfrac d n) \cdot a_i' a_i'^\top} Q \right \|
    \leq O\Paren{\frac 1 n} + o(\|v\|_4^4) \mper
  \]

  Finally, substituting $a_i' = Q A^{-1/2} a_i$, and using the fact that $Q$ is a rotation,
  it will be enough to show
  \begin{align}
    \label{eqn:fps-2}
    \left \| \Paren{\sum_{i = 1}^n (\|a_i\|^2 - \cconst) \cdot a_i a_i^\top} - A^{-1/2} \Paren{\sum_{i = 1}^n (\|A^{-1/2} a_i\|^2 - \cconst) \cdot a_i a_i^\top} A^{-1/2} \right \|
    \leq O\Paren{\frac 1 n} + o(\|v\|_4^4) \mper
  \end{align}
  We write the right-hand matrix as
  \begin{align*}
    & A^{-1/2} \Paren{\sum_{i = 1}^n (\|A^{-1/2} a_i\|^2 - \cconst) \cdot a_i a_i^\top} A^{-1/2} \\
    & = A^{-1/2} \Paren{\sum_{i = 1}^n (\|A^{-1/2} a_i\|^2 - \|a_i\|^2) \cdot a_i a_i^\top} A^{-1/2} + A^{-1/2} \Paren{\sum_{i = 1}^n (\|a_i\|^2 - \cconst) \cdot a_i a_i^\top} A^{-1/2}\mper
  \end{align*}
  The first of these we observe has bounded operator norm \wovp:
  \begin{align*}
    \left \| A^{-1/2} \Paren{\sum_{i = 1}^n (\|A^{-1/2} a_i\|^2 - \|a_i\|^2) \cdot a_i a_i^\top} A^{-1/2} \right \|
    & \leq \left \| A^{-1/2} \Paren{\sum_{i = 1}^n | \|A^{-1/2} a_i\|^2 - \|a_i\|^2 | \cdot a_i a_i^\top} A^{-1/2} \right \|\\
    & \leq \tO\Paren{\frac{d + \sqrt{n}}{n}} \cdot \left \| \sum_{i = 1}^n \|a_i\|^2 \cdot a_i a_i^\top \right \|\\
    \intertext{where we have used \pref{lem:fps-covariance-concentration} to find that $A^{1/2}$ is close to identity, and \pref{lem:ortho-4-restatement} to simplify the summands}
    & = \tO\Paren{\frac{d + \sqrt{n}}{n}} \cdot \Paren{ \left \| \sum_{i = 1}^n v_0(i)^2 \cdot a_i a_i^\top \right \| + \left \| \sum_{i=1}^n \|b_i\|_2^2 \cdot a_i a_i^\top \right \| }\\
    & \leq \tO\Paren{\frac{d + \sqrt{n}}{n}} \cdot \Paren{O(\|v\|_4^4) + \tO\Paren{\frac d n}}\mcom
  \end{align*}
  using in the last step \pref{lem:fps-good-basis-1} and standard concentration to bound $\sum_{i=1}^n \|b_i\|_2^2 \cdot a_i a_i^\top$ (\pref{lem:fps-covariance-concentration}).
  Thus, by triangle inequality applied to \pref{eqn:fps-2}, we get
  \begin{align*}
    & \left \| \Paren{\sum_{i = 1}^n (\|a_i\|^2 - \cconst) \cdot a_i a_i^\top} - A^{-1/2} \Paren{\sum_{i = 1}^n (\|A^{-1/2} a_i\|^2 - \cconst) \cdot a_i a_i^\top} A^{-1/2} \right \|\\
    & \leq  \tO\Paren{\frac{d + \sqrt{n}}{n}} \cdot \Paren{O(\|v\|_4^4) + \tO\Paren{\frac d n}}
    + \left \| \Paren{\sum_{i = 1}^n (\|a_i\|^2 - \cconst) \cdot a_i a_i^\top} - A^{-1/2} \Paren{\sum_{i = 1}^n (\|a_i\|^2 - \cconst) \cdot a_i a_i^\top} A^{-1/2} \right \|\mper
  \end{align*}
  Finally, since \wovp $\|A^{-1/2} - \Id\| = \tO(d/n)^{1/2}$, we get
  \begin{align*}
    & \left \| \Paren{\sum_{i = 1}^n (\|a_i\|^2 - \cconst) \cdot a_i a_i^\top} - A^{-1/2} \Paren{\sum_{i = 1}^n (\|A^{-1/2} a_i\|^2 - \cconst) \cdot a_i a_i^\top} A^{-1/2} \right \|\\
    & \leq \tO\Paren{\frac{d + \sqrt{n}}{n}} \cdot \Paren{O(\|v\|_4^4) + \tO\Paren{\frac d n}} + \tO\Paren{\frac d n}^{1/2} \cdot \left \| \sum_{i=1}^n (\|a_i\|_2^2 - \cconst) \cdot a_i a_i^\top \right \|\\
    & \leq  \tO\Paren{\frac{d + \sqrt{n}}{n}} \cdot \Paren{O(\|v\|_4^4) + \tO\Paren{\frac d n}} + \tO\Paren{\frac d n}^{1/2} \cdot O(\|v\|_4^4)\mper
  \end{align*}
  using \pref{lem:fps-main-technical} in the last step.
  For appropriate choice of $d \leq n^{-1/2} / \polylog(n)$, this is at most $O(n^{-1}) + o(\|v\|_4^4)$.
\end{proof}

\section{Overcomplete tensor decomposition} \label{sec:tdecomp}

In this section, we give a polynomial-time algorithm for the following problem
when $n \le d^{4/3}/(\polylog d)$:
\begin{problem}
  Given an order-$3$ tensor $\bT = \sum_{i = 1}^n a_i \tensor a_i \tensor a_i$, where
  $a_1,\ldots,a_n \in \R^d$ are iid vectors sampled from $\cN(0,\frac{1}{
  d} \Id)$, find vectors $b_1,\ldots,b_n \in \R^n$ such that for all $i\in[n]$,
  \[
    \iprod{a_i,b_i} \ge 1 - o(1) \mper
  \]
\end{problem}

We give an algorithm that solves this problem, so long as the overcompleteness of the input tensor is bounded such that $n \ll d^{4/3}/\polylog d$.

\begin{theorem}\label{thm:tdecomp-alg}
    Given as input the tensor $\bT = \sum_{i=1}^n a_i \tensor a_i \tensor a_i$ where $a_i \sim \cN(0,\tfrac{1}{d} \Id_d)$ with $d \leq n \leq d^{4/3}/\polylog d$,\footnote{The lower bound $d \leq n$ on $n$, is a matter of technical convenience, avoiding separate concentration analyses and arithmetic in the undercomplete $(n < d)$ and overcomplete $(n \geq d)$ settings.
  Indeed, our algorithm still works in the undercomplete setting (tensor decomposition is easier in the undercomplete setting than the overcomplete one), but here other algorithms based on local search also work \cite{DBLP:conf/colt/AnandkumarGJ15}.
  }
  there is an algorithm which may run in time $\tO(nd^{1 + \omega})$ or $\tO(nd^{3.257})$, where $d^\omega$ is the time to multiply two $d\times d$ matrices, which with probability $1 - o(1)$ over the input $\bT$ and the randomness of the algorithm finds unit vectors $b_1,\ldots, b_n \in \R^d$ such that for all $i\in[n]$,
    \[
      \iprod{a_i, b_i} \ge 1 - \tO\Paren{\frac {n^{3/2}} {d^{2}}}\mper
    \]
\end{theorem}
We remark that this accuracy can be improved from $1 - \tO(n^{3/2}/d^2)$ to an arbitrarily good precision using existing local search methods with local convergence guarantees---see \pref{cor:boost}.

As discussed in \pref{sec:techniques}, to decompose the tensor $\sum_i a_i^{\tensor 6}$ (note we do not actually have access to this input!) there is a very simple tensor decomposition algorithm: sample a random $g \in \R^{d^2}$ and compute the matrix $\sum_i \iprod{g, a_i^{\tensor 2}} (a_ia_i^\top)^{\tensor 2}$.
With probability roughly $n^{-O(\epsilon)}$ this matrix has (up to scaling) the form $(a_i a_i^\top)^{\tensor 2} + E$ for some $\|E\| \leq 1 - \epsilon$, and this is enough to recover $a_i$.

However, instead of $\sum_i a_i^{\tensor 6}$, we have only $\sum_{i,j} (a_i \tensor a_j)^{\tensor 3}$.
Unfortunately, running the same algorithm on the latter input will not succeed.
To see why, consider the extra terms $E' \seteq \sum_{i \neq j} \iprod{g, a_i \tensor a_j} (a_i \tensor a_j)^{\tensor 2}$.
Since $|\iprod{g, a_i \tensor a_j}| \approx 1$, it is straightforward to see that $\|E'\|_F \approx n$.
Since the rank of $E'$ is clearly $d^2$, even if we are lucky and all the eigenvalues have similar magnitudes, still a typical eigenvalue will be $\approx n/d \gg 1$, swallowing the $\sum_i a_i^{\tensor 6}$ term.

A convenient feature separating the signal terms $\sum_i (a_i \tensor a_i)^{\tensor 3}$ from the crossterms $\sum_{i \neq j} (a_i \tensor a_j)^{\tensor 3}$ is that the crossterms are not within the span of the $a_i \tensor a_i$.
Although we cannot algorithmically access $\Span \{ a_i \tensor a_i \}$, we have access to something almost as good: the unfolded input tensor,  $T = \sum_{i\in[n]} a_i (a_i\tensor a_i)^\top$.
The rows of this matrix lie in $\Span\{a_i \tensor a_i\}$, and so for $i\neq j$, $\|T(a_i \tensor a_i)\| \gg \|T(a_i \tensor a_j) \|$.
In fact, careful computation reveals that $\|T(a_i \tensor a_i) \| \geq \tOmega(\sqrt n / d) \|T(a_i \tensor a_j) \|$.

The idea now is to replace $\sum_{i,j} \iprod{g, a_i\tensor a_j} (a_i \tensor a_j)^{\tensor 2}$ with $\sum_{i,j} \iprod{g, T(a_i \tensor a_j)} (a_i \tensor a_j)^{\tensor 2}$, now with $g \sim \cN(0, \Id_d)$.
As before, we are hoping that there is $i_0$ so that $\iprod{g, T(a_{i_0} \tensor a_{i_0})} \gg \max_{j \neq i_0} \iprod{g, T(a_j \tensor a_j)}$.
But now we also require $\|\sum_{i \neq j} \iprod{g, T(a_i \tensor a_j)} (a_i \tensor a_j)(a_i \tensor a_j)^\top \| \ll  \iprod{g, T(a_{i_0} \tensor a_{i_0})} \approx \|T(a_i \tensor a_i)\|$.
If we are lucky and all the eigenvalues of this cross-term matrix have roughly the same magnitude (indeed, we will be lucky in this way), then we can estimate heuristically that
\begin{align*}
  \Norm{\sum_{i \neq j} \iprod{g, T(a_i \tensor a_j)} (a_i \tensor a_j)(a_i \tensor a_j)^\top} & \approx
 \tfrac 1 d \Norm{\sum_{i \neq j} \iprod{g, T(a_i \tensor a_j)} (a_i \tensor a_j)(a_i \tensor a_j)^\top}_F \\
 & \leq \tfrac 1 d \cdot \tfrac {\sqrt n}{d} |\iprod{g, T(a_{i_0} \tensor a_{i_0})}| \Norm{\sum_{i \neq j}  (a_i \tensor a_j)(a_i \tensor a_j)^\top}_F \\
 & \leq \tfrac{n^{3/2}}{d^2} |\iprod{g, T(a_{i_0} \tensor a_{i_0})}|\mcom
\end{align*}
suggesting our algorithm will succed when $n^{3/2} \ll d^{2}$, which is to say $n \ll d^{4/3}$.

The following theorem, which formalizes the intuition above, is at the heart of our tensor decomposition algorithm.
\begin{theorem}
  \label{thm:tensor-decomp-main}
  Let $a_1,\ldots,a_n$ be independent random vectors from $\cN(0,\tfrac{1}{d}\Id_d )$ with $d\le n \le d^{4/3}/(\polylog d)$
  and let $g$ be a random vector from $N(0, \Id_{d})$.
  Let $\Sigma \seteq \E_{x \sim \cN(0, \Id_d)} (xx^\top)^{\tensor 2}$ and let $\signh \seteq \sqrt 2 \cdot (\Sigma^+)^{1/2}$.
  Let $T = \sum_{i\in[n]} a_i (a_i \tensor a_i)^\top$.
  Define the matrix $M \in \R^{d^2 \times d^2}$,
  \begin{displaymath}
    M = \sum_{i,j\in [n]} \iprod{g, T(a_i\tensor a_j)} \cdot (a_i \tensor a_j)(a_i\tensor a_j)^\top\mper
  \end{displaymath}
  With probability $1-o(1)$ over the choice of $a_1,\ldots,a_n$, for every $\polylog d / \sqrt d < \epsilon < 1$, the spectral gap of $RMR$ is at least $\lambda_2/\lambda_1 \le 1 - O(\epsilon)$ and the top eigenvector $u\in \R^{d^2}$ of $R M R$ satisfies,
  with probability $\tilde \Omega(1/n^{O(\epsilon)})$ over the choice of $g$,
  $$
  \max_{i\in [n]}\langle  \signh u, a_i \otimes a_i \rangle^2/ \left(\lVert  u \rVert^2 \cdot \lVert  a_i \rVert^4\right) \ge 1 - \tO\Paren{\frac{n^{3/2}}{\epsilon d^2}}\mper
  $$
  Moreover, with probability $1-o(1)$ over the choice of $a_1,\ldots,a_n$, for every $\polylog d / \sqrt d < \epsilon < 1$ there are events $E_1,\ldots,E_n$ so that $\Pr_g E_i \geq \tOmega(1/n^{1 + O(\epsilon)})$ for all $i\in[n]$ and when $E_i $ occurs,
  $\langle  \signh u, a_i \otimes a_i \rangle^2 / \|u\|^2 \cdot \|a_i\|^4 \ge 1 - \tO\Paren{\frac{n^{3/2}}{\epsilon d^2}}$.
\end{theorem}
We will eventually set $\e = 1/\log n$, which gives us a spectral algorithm for recovering a vector $(1 - \tO(n/d^{3/2}))$-correlated to some $a_i^{\tensor 2}$.
Once we have a vector correlated with each $a_i^{\tensor 2}$, obtaining vectors close to the $a_i$ is straightforward.
We will begin by proving this theorem, and defer the algorithmic details to section \pref{sec:alg-details}.

The rest of this section is organized as follows.
In \pref{sec:proof-main-theorem} we prove \pref{thm:tensor-decomp-main} using two core facts: the Gaussian vector $g$ is closer to some $a_i$ than to any other with good probability, and the noise term $\sum_{i\neq j} \iprod{g, T(a_i \tensor a_j)} (a_i \tensor a_j)(a_i \tensor a_j)^\top$ is bounded in spectral norm.
In \pref{sec:diagonal-terms} we prove the first of these two facts, and in \pref{sec:xterm-bound} we prove the second.
In \pref{sec:alg-details}, we give the full details of our tensor decomposition algorithm, then prove \pref{thm:tdecomp-alg} using \pref{thm:tensor-decomp-main}.
Finally, \pref{sec:tensor-decomp-aux} contains proofs of elementary or long-winded lemmas we use along the way.

\subsection{Proof of \pref{thm:tensor-decomp-main}}
\label{sec:proof-main-theorem}
The strategy to prove \pref{thm:tensor-decomp-main} is to decompose the matrix $M$ into two parts $M=\Mdiag + \Mcross$, one formed by diagonal terms $\Mdiag = \sum_{i\in[n]} \langle  g, T (a_i \tensor a_i) \rangle \cdot (a_i \tensor a_i)(a_i\tensor a_i)^\top$ and one formed by cross terms $\Mcross = \sum_{i\neq j}  \langle  g, T(a_i\tensor a_j) \rangle \cdot (a_i\tensor a_j)(a_i\tensor a_j)^\top $.
We will use the fact that the top eigenvector $\Mdiag$ is likely to be correlated with one of the vectors $a_j^{\otimes 2}$, and also the fact that the spectral gap of $\Mdiag$ is noticeable.

The following two propositions capture the relevant facts about the spectra of $\Mdiag$ and $\Mcross$,
and will be proven in \pref{sec:diagonal-terms} and \pref{sec:xterm-bound}.

\begin{proposition}[Spectral gap of diagonal terms]
\torestate{
  \label{prop:diagonal-terms}
  Let $R = \sqrt{2}\cdot ((\E (xx^\top)^{\tensor 2})^+)^{1/2}$ for $x \sim \cN(0,\Id_d)$.
  Let $a_1,\ldots,a_n$ be independent random vectors from $\cN(0,\tfrac 1d \Id_d)$ with $d\le n \le d^{2-\Omega(1)}$ and let $g \sim \cN(0,\Id_d)$ be independent of all the others.
  Let $T \seteq \sum_{\inn} a_i (a_i \tensor a_i)^\top$.
  Suppose $\Mdiag = \sum_{i\in[n]}  \langle  g, T a_i^{\otimes 2} \rangle \cdot (a_i a_i^\top )^{\otimes 2} $.
  Let also $v_j$ be such that $v_jv_j^{\top} = \iprod{g,Ta_j^{\otimes 2}} \cdot (a_ja_j^{\top})^{\otimes 2}$.
  Then, with probability $1-o(1)$ over $a_1,\ldots,a_n$, for each $\e > \polylog d/\sqrt{d}$ and each $j \in [n]$, the event
  \begin{displaymath}
      E_{j,\epsilon} \defeq \left\{ \Norm{\signh \Mdiag \signh - \e \cdot \signh v_jv_j^{\top} \signh}
      \le \Norm{\signh \Mdiag \signh} - \Paren{\e - \tO\Paren{\sqrt{n}/{d}}} \cdot  \Norm{\signh v_jv_j^{\top} \signh} \right\}
  \end{displaymath}
  has probability at least $\tOmega(1/n^{1+O(\e)})$ over the choice of $g$.}
\end{proposition}

Second, we show that when $n \ll d^{4/3}$ the spectral norm of $\Mcross$ is negligible compared to this spectral gap.
\begin{proposition}[Bound on crossterms]
\torestate{
  \label{prop:cross-terms}
  Let $a_1,\ldots,a_n$ be independent random vectors from $\cN(0,\tfrac 1 d \Id_d)$,
  and let $g$ be a random vector from $\cN(0, \Id_{d})$.
  Let $T \seteq \sum_{i \in [n]} a_i (a_i \tensor a_i)^\top$.
  Let $\Mcross \seteq \sum_{i \neq j \in [n]} \iprod{g, T(a_i \tensor a_j)} a_i a_i^\top \tensor a_j a_j^\top$.
  Suppose $n \geq d$.
  Then with \wovp,
  \[
    \Norm{\Mcross} \leq \tO\Paren{\frac{n^3}{d^4}}^{1/2}\mper
  \]}
\end{proposition}

Using these two propositions we will conclude that the top eigenvector of $RMR$ is likely to be correlated with one of the vectors $a_j^{\otimes 2}$.
We also need two simple concentration bounds; we defer the proof to the appendix.
\begin{lemma}
\torestate{
  \label{lem:signal-negligible}
    Let $a_1,\ldots,a_n$ be independently sampled vectors from $\cN(0,\tfrac{1}{d}\Id_d)$, and let $g$ be sampled from $\cN(0,\Id_d)$.
    Let $T = \sum_i a_i (a_i \tensor a_i)^\top$.
    Then with overwhelming probability, for every $j \in [n]$,
    \[
	\left|\iprod{g, T(a_j \tensor a_j)} - \iprod{g, a_j} \|a_j\|^4 \right | \le \tO\Paren{\frac{\sqrt{n}}{d}}
    \]}
\end{lemma}

\begin{fact}[Simple version of \pref{fact:SIP}]
  \label{fact:SIP-simple}
  Let $x, y \sim \cN(0,\tfrac 1 d \Id)$.
  With overwhelming probability, $\Abs{1 - \|x\|^2} \leq \tO(1/\sqrt{d})$ and $\iprod{x,y}^2 = \tO(1/d)$.
\end{fact}

As a last technical tool we will need a simple claim about the fourth moment matrix of the multivariate Gaussian:
\begin{fact}[simple version of \pref{fact:sigma}]
  \label{fact:sigma-simple}
  Let $\Sigma = \E _{x \sim \cN(0,\Id_d)} (x x^\top)^{\otimes 2}$ and let $R = \sqrt{2}\, (\Sigma^+)^{1/2}$.
  Then $\|R\| = 1$, and for any $v \in \R^d$,
 \[
    \|\signh (v\tensor v) \|_2^2 = \left(1-\tfrac{1}{d+2}\right)\cdot \|v\|^4.
 \]
\end{fact}
We are prepared prove \pref{thm:tensor-decomp-main}.
\begin{proof}[Proof of \pref{thm:tensor-decomp-main}]
Let $d\le n \le d^{4/3}/(\polylog d)$ for some $\polylog d$ to be chosen later.
Let $a_1,\ldots,a_n$ be independent random vectors from $\cN(0,\tfrac 1 d \Id_d )$ and let $g \sim \cN(0, \Id_{d})$ be independent of the others.
Let
\begin{align*}
    \Mdiag &\seteq \sum_{\inn} \iprod{g, T(a_i \tensor a_i)} \cdot (a_i a_i^\top)^{\tensor 2} \quad \text{ and}\quad
    \Mcross &\seteq \sum_{i \neq j \in [n]} \iprod{g, T(a_i \tensor a_j)} \cdot a_i a_i^\top \tensor a_j a_j^\top\mper
\end{align*}
Note that $M \seteq \Mdiag + \Mcross$.

\pref{prop:cross-terms} implies that
\begin{equation}
\label{eq:main-1}
  \Prob{\vbig \Norm{\Mcross} \le  \cramped{\tilde O(  n^{3/2}/d^2)}} \ge 1-d^{-\omega(1)}\mper
\end{equation}
Recall that $\Sigma = \E_{x \sim \cN(0,\Id_d)} (xx^\top)^{\tensor 2}$ and $R = \sqrt 2 \cdot (\Sigma^+)^{1/2}$.
By \pref{prop:diagonal-terms}, with probability $1-o(1)$ over the choice of $a_1,\ldots,a_n$, each of the following events ${E}_{j,\e}$ for $j\in [n]$ and $\e>\polylog(d) / \sqrt d$ has probability at least $\tOmega(1/n^{1+O(\e)})$ over the choice of $g$:
\begin{align*}
  E^0_{j,\e} \colon\quad
   \Norm{\signh
       \Paren{\Mdiag - \e\langle g, T a_j^{\otimes 2} \rangle (a_j a_j^\top)^{\otimes 2}}
    \signh}
   \le \lVert  \signh \Mdiag \signh \rVert
   - (\e - \tO(\sqrt{n}/d)) \cdot \left(\cramped{\abs{\langle g,T{a_j}^{\otimes 2} \rangle}} \cdot \lVert \cramped{\signh {a_j}^{\otimes 2}} \rVert^2 \right)\mper
\end{align*}

Together with \pref{eq:main-1}, with probability $1-o(1)$ over the choice of $a_1,\ldots,a_n$, each of the following events ${E^*}_{j,\e}$ has probability at least $\tOmega(1/n^{1+O(\e)}) - d^{-\omega(1)}\ge \tOmega(1/n^{1+O(\e)})$ over the choice of $g$,
\begin{align*}
  E^*_{j,\e}\colon \quad
  & \Norm{\signh
    \Paren{M - \e \langle g, T a_j^{\otimes 2} \rangle (a_j a_j^\top)^{\otimes 2}}
    \signh} \\
  & \le
    \begin{aligned}[t]
    & \left\lVert \signh \cdot \Mdiag \cdot \signh \right\rVert\\
    & - (\e - \tO(\sqrt{n}/d)) \cdot \left(\cramped{\abs{\langle g,T{a_j}^{\otimes 2} \rangle}} \cdot \lVert \cramped{\signh {a_j}^{\otimes 2}} \rVert^2 \right)
      + \cramped{\tilde O(n^{3/2}/d^2)}
  \end{aligned}\\
  & \le \left\lVert \signh\cdot M \cdot \signh \right\rVert
  - (\e - \tO(\sqrt{n}/{d})) \cdot |\langle g,T{a_j}^{\otimes 2} \rangle|\cdot\lVert  \cramped{ \signh {a_j }^{\otimes 2}} \rVert^2 + \tO(n^{3/2}/d^{2})\mper
\end{align*}
Here, we used that $M = \Mdiag + \Mcross$ and that $\norm{\signh \cdot \Mcross \cdot \signh}\le \norm{\Mcross}$ as $\lVert \signh \rVert\le 1$ (\pref{fact:sigma-simple}).

By standard reasoning about the top eigenvector of a matrix with a spectral gap (recorded in \pref{lem:low-correlation}), the event ${E^*}_{j,\e}$ implies that the top eigenvector $\cramped{u\in \R^{d^2}}$ of $\cramped{\signh \cdot M \cdot \signh}$ satisfies
\begin{align*}
    \left\langle  u, \frac{\signh a_j^{\otimes 2}}{\|\signh a_j^{\tensor 2}\|} \right\rangle^2
    &\ge 1 - \frac{\tO(\sqrt{n}/d)}{\epsilon \|Ra_j^{\tensor 2}\|^2}  - \frac{\tO(n^{3/2}/d^2)}{\epsilon \|Ra_j^{\tensor 2}\|^2|\iprod{g,Ta_j^{\tensor 2}}|} \mper
    \intertext{
Since $\lVert  \cramped{\signh {a_j}^{\otimes 2}} \rVert^2\ge \Omega(\lVert  a_j \rVert^4)$ (by \pref{fact:sigma-simple}),
and since $\|a_j\| \ge 1 - \tO(1/\sqrt{d})$ (by \pref{fact:SIP-simple}),
    }
    &\ge 1 - \tO\Paren{\frac{\sqrt{n}}{\epsilon d}} - \frac{\tO(n^{3/2}/d^2)}{\epsilon \cdot |\iprod{g,Ta_j^{\tensor 2}}|}
\end{align*}

Now, by \pref{lem:signal-negligible} we have that for all $j\in[n]$, $|\iprod{g, Ta_j^{\tensor 2}} - \iprod{g,a_j}\|a_j\|^4| \le \tO(\sqrt{n}/{d})$ with probability $1 - \owl$. By standard concentration (see \pref{fact:SIP} for a proof) $|\iprod{g, a_j}\|a_j\|^4 - 1| \le \tO(1/\sqrt{d})$ for all $j\in[n]$ with probability $1-\owl$.
Therefore with overwhelming probability, the final term is bounded by $\tO(n^{3/2} / \epsilon d^2)$.
A union bound now gives the desired conclusion.

Finally, we give a bound on the spectral gap.
We note that the second eigenvector $w$ has $\iprod{u,w} = 0$, and therefore $$
\Iprod{w, \frac{Ra_j^{\tensor 2}}{\|Ra_j^{\tensor 2}\|}}
= \Iprod{w, \frac{Ra_j^{\tensor 2}}{\|Ra_j^{\tensor 2}\|} - u}
\le \left\| \frac{Ra_j^{\tensor 2}}{\|Ra_j^{\tensor 2}\|}- u\right\|
\le \tO(n^{3/2}/\epsilon d^2).
$$
Thus, using our above bound on $\|R(M - \epsilon \iprod{g,Ta_j^{\tensor 2}}(a_j a_j^\top)^{\tensor 2})R\|$ and the concentration bounds we have already applied for $\|a_j\|$, $\iprod{g,Ta_j^{\tensor 2}}$, and $\|R a_j^{\tensor 2}\|$, we have that
\begin{align*}
    \lambda_2(RMR)
    &= w^\top RMR w\\
    &= w^\top R\left(M  - \epsilon\iprod{g,Ta_j^{\tensor 2}}\cdot (a_j a_j^\top)^{\tensor 2}\right)Rw + \epsilon\iprod{g,Ta_j^{\tensor 2}} \iprod{w, R a_j^{\tensor 2}}^2\\
    &\le \left\|R\left(M  - \epsilon\iprod{g,Ta_j^{\tensor 2}}\cdot (a_j a_j^\top)^{\tensor 2}\right)R\right\| + \tO(n^{3/2}/\epsilon d^2)\\
    &\le 1 - \tO(\epsilon) + \tO(n^{3/2}/\epsilon d^2)\mper
\end{align*}
We conclude that the above events also imply that $\lambda_2(RMR)/\lambda_1(RMR) \le 1 -O(\epsilon)$.
\end{proof}

\subsection{Spectral gap for diagonal terms: proof of \pref{prop:diagonal-terms}}
\label{sec:diagonal-terms}
We now prove that the signal matrix, when preconditioned by $\signh$, has a noticeable spectral gap:
\restateprop{prop:diagonal-terms}
The proof has two parts.
First we show that for $a_1,\ldots,a_n \sim \cN(0, \Id_d)$ the matrix $P \seteq \sum_{\inn} (a_i a_i^\top)^{\tensor 2}$ has tightly bounded spectral norm when preconditioned with $\signh$: more precisely, that $\|\signh P \signh\| \leq 1 + \tO(n/d^{3/2})$.
\begin{lemma}
\torestate{
   \label{lem:bounded-signal}
     Let $a_1,\ldots,a_n\sim \cN(0,\tfrac 1d \Id_d)$ be independent random vectors with $d \leq n$.
     Let $\signh \seteq \sqrt{2}\cdot ((\E (a a^\top)^{\tensor 2})^+)^{1/2} $ for
     ${a \sim \cN(0,\Id_d)}$.
     For $S\subseteq [n]$, let $P_S=\sum_{i\in S} (a_{i} {a_i}^\top)^{\otimes 2}$ and let $\Pi_S$ be the projector into the subspace spanned by $\{\signh a_i^{\otimes 2} \mid i\in S\}$.
     Then, with probability $1-o(1)$ over the choice of $a_1,\ldots,a_n$,
     \begin{displaymath}
          \forall S\subseteq [n].\quad
          \left(1- \tO(n/d^{3/2})\right)\cdot \Pi_S
          \preceq   \signh P_S \signh
          \preceq \left(1 + \tO(n/d^{3/2})\right) \cdot \Pi_S
          \mper
     \end{displaymath}}
\end{lemma}
\begin{remark}
In \cite[Lemma 5]{DBLP:conf/approx/GeM15} a similar lemma to this one is proved in the context of the SoS proof system.
However, since Ge and Ma leverage the full power of the SoS algorithm their proof goes via a spectral bound on a different (but related) matrix.
Since our algorithm avoids solving an SDP we need a bound on this matrix in particular.
\end{remark}

The proof of \pref{lem:bounded-signal} proceeds by standard spectral concentration for tall matrices with independent columns (here the columns are $Ra_i^{\tensor 2}$).
The arc of the proof is straightforward but it involves some bookkeeping; we have deferred it to \pref{sec:bounded-signal}.

We also need the following lemma on the concentration of some scalar random variables involving $\signh$; the proof is straightforward by finding the eigenbasis of $\signh$ and applying standard concentration, and it is deferred to the appendix.
\begin{lemma}
\torestate{
\label{lem:gap-scalar}
Let $a_1,\ldots,a_n \sim \cN(0,\tfrac 1 d \Id_d)$.
Let $\Sigma, R$ be as in \pref{fact:sigma-simple}.
Let $u_i = a_i \tensor a_i$.
With overwhelming probability, every $j \in [n]$ satisfies $\sum_{i \neq j} \iprod{u_j, R^2 u_i}^2 = \tO(n/d^2)$ and $|1 - \|Ru_j\|^2| \leq \tO(1/\sqrt d)$.}
\end{lemma}

The next lemma is the linchpin of the proof of \pref{prop:diagonal-terms}: one of the inner products $\langle  g, T{a_j}^{\otimes 2} \rangle$, is likely to be a $\approx (1 + 1/\log(n))$-factor larger than the maximum of the inner products $\langle  g, T{a_i}^{\otimes 2} \rangle$ over $i\neq j$.
Together with standard linear algebra these imply that the matrix $\Mdiag = \sum_{\inn} \iprod{g, T a_i^{\tensor 2}} (a_i a_i^\top)^{\tensor 2}$ has top eigenvector highly correlated or anticorrelated with some $a_i$.

\begin{lemma}
  \label{lem:g-boost}
  Let $a_1,\ldots,a_n\in\R^d$ be independent random vectors from $\cN(0,\frac{1}{d}\Id_{d})$, and let $g$ be a random vector from $\cN(0,\Id_{d})$.
  Let $T = \sum_{i\in[n]} a_i (a_i \tensor a_i)^\top$.
  Let $\e > 0$ and $j\in[n]$.
  Then with overwhelming probability over $a_1,\ldots,a_n$, the following event $\hat E_{j,\e}$ has probability $1/n^{1+O(\e)+\tO(1/\sqrt{d})}$ over the choice of $g$,
  \[
      \hat E_{j,\e} = \left\{\iprod{g,T a_j^{\otimes 2}} \ge (1+\e)(1-\tO(1/\sqrt{d}))\cdot \max_{i\neq j}\left|\iprod{g,T a_i^{\otimes 2}}\right|\right\}\mper
  \]
\end{lemma}

Now we can prove \pref{prop:diagonal-terms}.
\begin{proof}[Proof of \pref{prop:diagonal-terms}]
\newcommand{\etanorm}{\eta_{\text{norm}}}
\newcommand{\etagap}{\eta_{\text{gap}}}
Let $u_i \seteq a_i^{\tensor 2}$.
Fix $j \in [n]$.
We begin by showing a lower bound on the spectral norm $\|R \Mdiag R \|$.
\begin{align*}
  \|R \Mdiag R\| & = \max_{\|v\|=1} \left| \iprod{v, R\Mdiag Rv} \right| \\
  & \geq \frac{\iprod{Ru_j, (R\Mdiag R)Ru_j}}{\|Ru_j\|^2}\\
  & = \frac 1 {\|Ru_j\|^2} \Paren{\iprod{g,Tu_j} \|Ru_j\|^4 + \iprod{Ru_j, \sum_{i \neq j} \iprod{g,T_i} Ru_iu_i^\top R \cdot Ru_j}}
\end{align*}
From \pref{lem:g-boost}, the random vector $g$ is closer to $Tu_j$ than to all $Tu_i$ for $i \neq j, i \in[n]$ with reasonable probability.
More concretely there is some $\polylog d$ so that as long as $\epsilon > \polylog d /\sqrt{d}$ there is some $\alpha = \Theta(\epsilon)$ with $1 - \epsilon = 1/[(1+\alpha)(1-\tO(d^{-1/2}))]$ so that with \wovp over $a_1,\ldots,a_n$ the following event (a direct consequence of $\hat E_{j,\e}$) has probability
$\tOmega(1/n^{1+O(\alpha) + \tO(d^{-1/2})}) = \tOmega(1/n^{1 + O(\epsilon)})$ over $g$:
\begin{equation}\label{eq:gap-1}
    -(1-\epsilon) |\iprod{g,T u_j}| \cdot
  \Paren{\sum_{i\neq j} R u_i u_i^\top R}
  \preceq
  \sum_{i\neq j} \iprod{g,T u_i} \cdot R u_i u_i^\top R
  \preceq (1-\e) |\iprod{g,T u_j}| \cdot
  \Paren{\sum_{i\neq j} R u_i u_i^\top R}.
\end{equation}

When \pref{eq:gap-1} occurs,
\begin{align}
    \|R \Mdiag R\| & \geq \frac 1 {\|Ru_j\|^2} \Paren{|\iprod{g,Tu_j}| \|Ru_j\|^4 - (1-\epsilon) |\iprod{g, Tu_j}|\iprod{Ru_j, \sum_{i \neq j} Ru_iu_i^\top R \cdot Ru_j}}\nonumber \\
				 & = \frac {|\iprod{g, Tu_j}|} {\|Ru_j\|^2} \Paren{\|Ru_j\|^4 - (1-\epsilon) \sum_{i \neq j} \iprod{ u_j, R^2 u_i}^2 } \nonumber\\
     & \geq \frac{|\iprod{g, Tu_j}|\Paren{1 - \tO(1/\sqrt d) - (1-\e) \tO(n/d^2)} }{1 + \tO(1/\sqrt d)} \quad \text{\wovp over $a_1,\ldots,a_n$ (\pref{lem:gap-scalar})}\nonumber \\
  &\ge |\iprod{g,Tu_j}| \cdot (1-\etanorm ) \mcom\label{eq:gap-2}
\end{align}
where we have chosen some $0 \le \etanorm \le \tO(1/\sqrt{d}) + \tO(n/d^2)$ (since for any $x \in \R$, $(1 +x)(1-x) \le 1$).

Next we exhibit an upper bound on $\|R \Mdiag R - \epsilon \iprod{g, Tu_j} Ru_j u_j^\top R\|$.
Again when \pref{eq:gap-1} occurs,
\begin{align}
   & \|R \Mdiag R - \e \iprod{g, Tu_j} Ru_j u_j^\top R\|\\
   & \quad \quad = \Norm{ (1-\e) \iprod{g,Tu_j} Ru_j u_j^\top R + \sum_{i \neq j} \iprod{g,Tu_i} Ru_i u_i^\top R}\nonumber\\
   & \quad \quad \leq (1-\e) |\iprod{g, Tu_j}| \Norm{\sum_{i \in[n]} R u_i u_i^\top R}\quad \text{when \pref{eq:gap-1} occurs}\nonumber\\
   & \quad \quad \leq (1-\e) |\iprod{g, Tu_j}| (1 + \tO(n/d^{1.5}))\quad \text{w.p. $1 - o(1)$ over $a_1,\ldots,a_n$ by \pref{lem:bounded-signal}} \nonumber \\
   & \quad \quad \leq (1-\e) |\iprod{g,Tu_j}| (1 + \etagap)\label{eq:gap-3}
\end{align}
where we have chosen some $0\le \etagap \le \tO(n/d^{1.5})$.

Putting together \pref{eq:gap-2} and \pref{eq:gap-3} with our bounds on $\etanorm$ and $\etagap$ and recalling the conditions on \pref{eq:gap-1}, we have shown that
\begin{align*}
     \Pr_{a_1,\ldots,a_n}
     \big \{ \Pr_g \{ \|R \Mdiag R - & \e \iprod{g,Tu_j} R u_j u_j^\top R \|
   \le  \|R\Mdiag R\| -  (\e - \tO(\sqrt{n}/{d}))\cdot|\iprod{g,Tu_j}|\cdot \|R u_j\|^2 \} \\
    & \geq \tOmega(1/n^{1 + O(\epsilon)}) \big \}
    \geq 1 - o(1)\mper
\end{align*}
This concludes the argument.
\end{proof}

We now turn to proving that with reasonable probability, $g$ is closer to some $Ta_j^{\tensor 2}$ than all others.

\begin{proof}[Proof of \pref{lem:g-boost}]
  To avoid proliferation of indices, without loss of generality fix $j = 1$.
    We begin by expanding the expression $\iprod{g,Ta_i^{\tensor 2}}$,
    \[
	\iprod{g, Ta_i^{\tensor 2}}
	= \sum_{\ell \in[n]} \iprod{g,a_\ell} \iprod{a_\ell,a_i}^2
	= \|a_i\|^4 \iprod{g,a_i} + \sum_{\ell \neq i} \iprod{g,a_\ell}\iprod{a_\ell,a_i}^2.
    \]
    The latter sum is bounded by
    \[
	\left|\sum_{\ell \neq i} \iprod{g,a_\ell}\iprod{a_\ell,a_i}^2\right| \le \tO\Paren{\frac{\sqrt{n}}{d}}\mcom
    \]
    with overwhelming probability for all $i$ and choices of $g$; this follows from a Bersntein bound, given in \pref{lem:signal-negligible}.

  For ease of notation, let $\hat a_{i} \defeq a_{i}/\|a_{i}\|_2$.
  We conclude from \pref{fact:SIP-simple} that with overwhelming probability, $1 - \tO(1/\sqrt{d}) \le \|a_i\|_2 \le 1 + \tO(1/\sqrt{d})$ for all $i\in[n]$.
  Thus $\|a_{i}\|_2$ is roughly equal for all $i$, and we may direct our attention to $\iprod{g,\hat a_i}$.

  Let $\Eone$ be the event that
  $\sqrt{2\ampc} \log^{1/2} n\le |\iprod{g,\hat a_1}|\le d^{1/4}$ for some $\ampc \leq d^{1/2 - \Omega(1)}$ to be chosen later.
  We note that $\iprod{g,\hat a_1}$ is distributed as a standard gaussian,
  and that $g$ is independent of $a_1,\ldots,a_n$.
  Thus, we can use standard tail
  estimates on univariate Gaussians (\pref{lem:gtail}) to conclude that
  \[
    \Pr\Paren{|\iprod{g,\hat a_1}| \ge \sqrt{2\ampc} \log^{1/2} n}
    = \tTheta (n^{-\ampc}) \ \ \text{and}\ \
    \Pr\Paren{|\iprod{g,\hat a_1}| \ge d^{1/4}} =
    \Theta\Paren{\frac{\exp(-\sqrt{d}/2)}{d^{1/4}}}\mper
  \]
  So by a union bound, $\Pr(\Eone) \ge \tOmega(n^{-\ampc}) - O(e^{-{d^{1/2}}/3})= \tOmega(n^{-\ampc})$.

  Now, we must obtain an estimate for the probability that all other inner products with $g$ are small.
  Let $\Eother$ be the event that $|\iprod{g,\hat a_{i}}| \le \sqrt{\restc} \log^{1/2} n$ for all $i\in[n],i>1$ and for some $\rho$ to be chosen later.
  We will show that conditioned on $\Eone$, $\Eother$ occurs with probability $1 - O(n^{1 -\restc/2})$.
  Define $g_{1} \seteq \iprod{g, \hat a_1} \hat{a_1}$ to be the component of $g$ parallel to $a_1$, let $g_{\perp} \seteq g - g_1$ be the component of $g$ orthogonal to $\hat a_1$, and similarly let $\hat a_2^{\perp}, \ldots, \hat a_n^{\perp}$ be the components of $\hat a_2,\ldots, \hat a_n$ orthogonal to $a_1$.
  Because $g_{\perp}$ is independent of $g_{1}$, even conditioned on $\Eone$ we may apply the standard tail bound for univariate Gaussians (\pref{lem:gtail}), concluding that for all $i>1$,
  \[
    \Pr\Paren{|\iprod{g_{\perp},\hat a_i}|\ge \sqrt{\restc} \log^{1/2} n \, \big | \, \Eone} =
    \tTheta(n^{-\restc/2}).
  \]
  Thus, a union bound over $i\neq 1$ allows us to conclude that conditioned
  on $\Eone$, with
  probability $1- \tO(n^{-\rho/2})$
  every $i\in[n]$ with $i>1$ has $| \iprod{g_{\perp},\hat a_i^{\perp}}| \le \sqrt{\restc} \log^{1/2} n$.

  On the other hand, let $\hat a_2^{\parallel},\ldots,\hat a_n^{\parallel}$
  be the components of the $\hat a_i$ parallel to $\hat a_1$. We compute the
  projection of $\hat a_i$ onto $\hat a_1$. With overwhelming probability,
  \begin{align*}
    \iprod{\hat a_1, \hat a_i}
    &=  \frac{\iprod{a_1,a_i}}{ \|a_1\|_2 \cdot \|a_i\|_2}\\
    &= (1 \pm \tO(1/\sqrt{d})) \cdot \iprod{a_1,a_i} \quad \text{\wovp by $\|a_i\|, \|a_1\| = 1 \pm \tO(1/\sqrt{d})$ (\pref{fact:SIP-simple})}\\
    &= (1 \pm \tO(1/\sqrt{d}))\cdot \tO(1/\sqrt{d}) \quad \text{\wovp by $\iprod{a_1, a_i} = \tO(1/\sqrt{d})$ (\pref{fact:SIP-simple})},
  \end{align*}
  Thus \wovp,
  \begin{align*}
    \iprod{g_1, \hat a_i^{\parallel}}
    &= \iprod{g,\hat a_1} \cdot \iprod{\hat a_1,\hat a_i} \ \le\ \iprod{g,\hat a_1} \cdot \tO(1/\sqrt{d}),
  \end{align*}
  for all $i\in [n]$.
  Now we can analyze $\Eother$.
  Taking a union bound over the overwhelmingly probable events (including $\|a_i\| \leq 1 + \tO(1 / \sqrt d)$) and the event that $\iprod{g_{\perp},a_i}$ is small for all $i$, we have that with probability $1 - \tO(n^{-\rho/2})$,
  for every $i\in[n]$ with $i> 1$,
  \begin{align*}
    |\iprod{g,\hat a_i}|
    &\le |\iprod{g_{\perp}, \hat a_i}| + |\iprod{g_1, \hat a_i}|\\
    &\sqrt{(2 + \rho)}\log^{1/2} n + \tO(1/\sqrt{d}) \cdot \iprod{g, \hat a_1}\\
      &\le \sqrt{(2 + \rho)}\log^{1/2} n + \tO(1/d^{1/4}).
  \end{align*}
  We conclude that
  \begin{align*}
    \Pr\Paren{\Eone,\Eother}
    &= \Pr\Paren{\Eother | \Eone} \cdot
      \Pr\Paren{\Eone}\\
    &\ge (1-O(n^{-\rho/2})) \cdot \tOmega(n^{-\ampc})
  \end{align*}
  Setting $\rho = 2\frac{\log\log n}{\log n}$ and $\alpha = (1+\e)^2(1 + \log\log n/\log n + \tO(1/\sqrt{d}))$, the
  conclusion follows.
\end{proof}

\subsection{Bound for cross terms: proof of \pref{prop:cross-terms}}
\label{sec:xterm-bound}
We proceed to the bound on the cross terms $\Mcross$.
\restateprop{prop:cross-terms}
The proof will use two iterations of Matrix Rademacher bounds.
The first step will be to employ a classical decoupling inequality that has previously been used in a tensor decomposition context
\cite{DBLP:conf/approx/GeM15}.

\begin{theorem}[Special Case of Theorem 1 in \cite{de1995decoupling}]
\label{thm:decoupling}
  Let $\{ s_i \}, \{ t_i \}$ be independent iid sequences of random signs.
Let $\{ M_{ij} \}$ be a family of matrices.
There is a universal constant $C$ so that for every $t > 0$,
\[
  \Pr \Paren{ \op{\sum_{i \neq j} s_i s_j M_{ij} } > t} \leq C\cdot \Pr \Paren{ C \op{\sum_{i \neq j} s_i t_j M_{ij}} > t}\mper
\]
\end{theorem}

Once the simplified cross terms are decoupled, we can use a matrix Rademacher bound on one set of signs.
\begin{theorem}[Adapted from Theorem 4.1.1 in \cite{DBLP:journals/FOCM/Tropp12}\footnote{We remark that Tropp's bound is phrased in terms of $\lambda_{max} \sum_i s_i M_i$.
Since $\lambda_{max} \sum_i s_i M_i = \lambda_{min} \sum_i -s_i M_i$,
and the distribution of $s_i M_i$ is negation-invariant, the result we state here follows from an easy union bound.}
]
  \label{thm:matrix-rad}
  Consider a finite sequence $\{ M_i \}$ of fixed $m \times m$ Hermitian matrices.
  Let $s_i$ be a sequence of independent sign variables.
  Let $\sigma^2 \seteq \| \sum_i M_i^2 \|$.
  Then for every $t \geq 0$,
  \[
    \Pr \Paren{ \op{\sum_i s_i M_i} \geq t} \leq 2 m\cdot e^{-t^2 / 2\sigma^2}\mper
  \]
  Also,
  \[
    \E \left \| \sum_i s_i M_i \right \| \leq \sqrt{8 \sigma^2 \log d}\mper
  \]
\end{theorem}
\begin{corollary}
  \label{cor:rad-and-tensor}
  Let $s_1,\ldots,s_n$ be independent signs in $\{ -1, 1 \}$.
  Let $A_1,\ldots, A_n$ and $B_1, \ldots, B_n$ be Hermetian matrices.
  Then \wovp,
  \[
    \Norm{\sum_i s_i \cdot A_i \tensor B_i } \leq \tO\Paren{\max_i \|B_i\| \cdot \Norm{ \sum_i A_i^2 }^{1/2}}\mper
  \]
\end{corollary}
\begin{proof}
We use a matrix Rademacher bound and standard manipulations:
\begin{align*}
   \Norm{\sum_i s_i \cdot A_i \tensor B_i } \quad & \wovple \quad \tO\Paren{\Norm{\sum_i A_i^2 \tensor B_i^2}}^{1/2}\\
   & \leq \tO\Paren{\Norm{ \sum_i \|B_i\|^2 \cdot (A_i^2 \tensor  \Id)}}^{1/2} \quad \text{since $A_i^2$ is PSD for all $i$}\\
   & \leq \tO\Paren{\max_i \|B_i\|^2 \cdot \Norm{ \sum_i A_i^2}}^{1/2} \quad \text{since $A_i^2 \tensor \Id$ is PSD for all $i$}
\mper \qedhere
\end{align*}
\end{proof}

We also need a few further concentration bounds on matrices which will come up as parts of $\Mcross$.
These can be proved by standard inequalities for sums of independent matrices.
\begin{lemma}[Restatement of \pref{fact:GC} and \pref{lem:Mj-bound}]
  \label{lem:Mj-GC-simple}
  Let $a_1,\ldots,a_n$ be independent from $\cN(0,\tfrac 1 d \Id_d)$ with $n \geq d \polylog(d)$.
  With overwhelming probability, $\tOmega(n/d) \cdot \Id \preceq \sum_{i\in[n]} a_i a_i^\top \preceq \tO(n/d)\cdot \Id$.
  Additionally, if $g \sim \cN(0,\Id_d)$ is independent of the rest, for every $j \in [n]$ \wovp
  \[
    \Norm{\sum_{\substack{ \inn \\ i \neq j}} \iprod{g,a_i} \|a_i\|^2 \iprod{a_i, a_j} \cdot a_i a_i^\top} \leq \tO(n/d^2)^{1/2}\mper
  \]
\end{lemma}

\begin{proof}[Proof of \pref{prop:cross-terms}]
  We expand $\Mcross$:
    \begin{align*}
	\Mcross
	&= \sum_{i\neq j} \iprod{g, T(a_i\tensor a_j)} \cdot a_ia_i^\top \tensor a_j a_j^\top\\
	&= \sum_{i\neq j} \Paren{ \sum_{\ell\in[n]} \iprod{a_\ell,a_i}\iprod{a_\ell,a_j} \iprod{g,a_\ell}} \cdot a_i a_i^\top \tensor a_j a_j^\top\mper
        \intertext{Since the joint distribution of $(a_1,\ldots,a_n)$ is identical to that of $(s_1a_1,\ldots,s_na_n)$, this is distributed identically to}
	\Mcross' &= \sum_{\ell\in[n]} \sum_{i\neq j} s_i s_j s_\ell \iprod{g,a_\ell} \iprod{a_\ell,a_i}\iprod{a_\ell,a_j}  \cdot a_i a_i^\top \tensor a_j a_j^\top\mcom
    \end{align*}
    (where we have also swapped the sums over $\ell$ and $i \neq j$).
    We split $\Mcross'$ into $\Mdiff$, for which $i\neq \ell$ and $j\neq \ell$, and $\Msame$, for which $\ell = i$ or $\ell = j$, and bound the norm of each of these sums separately.
    We begin with $\Msame$.
\begin{align*}
    \Msame & \defeq \sum_{i\neq j} s_i^2 s_j \iprod{g,a_i} \iprod{a_i,a_i}\iprod{a_i,a_j}  \cdot a_i a_i^\top \tensor a_j a_j^\top
     + \sum_{i \neq j} s_j^2 s_i \iprod{g,a_j} \iprod{a_j,a_j}\iprod{a_i,a_j}  \cdot a_i a_i^\top \tensor a_j a_j^\top\mper
\end{align*}
By a union bound and an application of the triangle inequality it will be enough to show that just one of these two sums is $\tO(n^3/d^4)^{1/2}$ \wovp.
We rewrite the left-hand one:
\[
  \sum_{i\neq j} s_i^2 s_j \iprod{g,a_i} \iprod{a_i,a_i}\iprod{a_i,a_j}  \cdot a_i a_i^\top \tensor a_j a_j^\top
   = \sum_{j\in[n]} s_j a_j a_j^\top \tensor \Paren{\sum_{i\neq j} \iprod{g,a_i} \|a_i\|^2 \iprod{a_i,a_j}  \cdot a_i a_i^\top}\mper
\]
Define
\[
  M_j \defeq \sum_{i\neq j} \iprod{g,a_i} \|a_i\|^2 \iprod{a_i,a_j}  \cdot a_i a_i^\top
\]
so that now we need to bound $\sum_{j \in [n]} s_j a_j a_j^\top \tensor M_j$.
By \pref{cor:rad-and-tensor},
\begin{align*}
    \Norm{\sum_{j\in[n]} s_j a_j a_j^\top \tensor M_j}
    \quad & \wovple \quad \tO(\max_j \|M_j\|) \cdot \tO\Paren{\Norm{\sum_{j \in [n]} \|a_j\|^2 a_j a_j^\top}^{1/2}}\\
    & \leq \tO(\max_j \|M_j\|) \cdot \max_j \|a_j\| \cdot \tO\Paren{\Norm{\sum_{j \in [n]} a_j a_j^\top}^{1/2}}\\
    \intertext{In \pref{lem:Mj-GC-simple}, we bound $\max_j \|M_j\| \leq \tO(n/d^2)^{1/2}$ \wovp using a matrix Bernstein inequality. Combining this bound with the concentration of $\|a_j\|$ around 1 (\pref{fact:SIP-simple}), we obtain}
    & \wovple \quad \tO(n/d^2)^{1/2} \cdot \tO(n/d)^{1/2} \\
    & = \tO(n/d^{1.5})\mper
\end{align*}
Having finished with $\Msame$, we turn to $\Mdiff$.
\begin{align*}
    \Norm{\Mdiff}
    &= \Norm{\sum_{\ell\neq i\neq j} s_i s_j s_\ell \iprod{g,a_\ell} \iprod{a_\ell,a_i}\iprod{a_\ell,a_j}  \cdot a_i a_i^\top \tensor a_j a_j^\top}\\
    &= \Norm{\sum_{\ell} s_\ell \iprod{g,a_\ell}  \Paren{\sum_{i\neq \ell} s_i \iprod{a_\ell,a_i}a_i a_i^\top \tensor \Paren{\sum_{j\neq \ell,i}s_j \iprod{a_\ell,a_j}a_j a_j^\top}}}\mper
\end{align*}
Letting $t_1,\ldots,t_n$ and $r_1,\ldots,r_n$ be independent uniformly random signs, by \pref{thm:decoupling}, it will be enough to bound the spectral norm after replacing the second and third occurrences of $s_i$ for $t_i$ and $r_i$.
To this end, we define
\[
  \Mdiff' \defeq \sum_{\ell} s_\ell \iprod{g,a_\ell} \Paren{\sum_{i\neq \ell} t_i \iprod{a_\ell,a_i}a_i a_i^\top \tensor \Paren{\sum_{j\neq \ell,i}r_j \iprod{a_\ell,a_j} a_j a_j^\top}}\mper
\]
Let
\[
  N_\ell \defeq \sum_{i\neq \ell} t_i \iprod{a_\ell,a_i}a_i a_i^\top \tensor \Paren{\sum_{j\neq \ell,i}r_j \iprod{a_\ell,a_j} a_j a_j^\top}
\]
so that we are to bound $\Norm{\sum_{\ell \in [n]} s_\ell \iprod{g,a_\ell} \cdot N_\ell}$.
By a matrix Rademacher bound and elementary manipulations,
\begin{align*}
  \Norm{\sum_{\ell \in [n]} s_\ell \iprod{g, a_\ell} \cdot N_\ell} \quad & \wovple \quad \tO\Paren{\Norm{\sum_{\ell \in [n]} \iprod{g, a_\ell}^2 \cdot N_\ell^2}}^{1/2}\\
  & \leq \tO(\sqrt n) \cdot \max_{\ell \in [n]} |\iprod{g, a_\ell}| \cdot \max_{\ell \in [n]} \|N_\ell\|\\
  & \wovple \tO(\sqrt n) \cdot \max_{\ell \in [n]} \|N_\ell\| \quad \text{since $|\iprod{g, a_i}| \leq \tO(1)$ (\pref{fact:SIP-simple})}\mper
\end{align*}
The rest of the proof is devoted to bounding $\|N_\ell \|$.

We start with \pref{cor:rad-and-tensor} to get
\[
\|N_\ell\| \quad \wovple \quad \tO\Paren{\Paren{ \max_i \Norm{\sum_{j\neq \ell,i} r_j \iprod{a_\ell,a_j}\cdot a_j a_j^\top}}
  \cdot \Norm{\sum_{i \neq \ell} \iprod{a_\ell, a_i}^2 \|a_i\|^2 \cdot a_i a_i^\top}^{1/2}}
\]
We use a matrix Rademacher bound for the left-hand matrix,
\begin{align*}
    \Norm{\sum_{j\neq \ell,i} r_j \iprod{a_\ell,a_j} \cdot a_j a_j^\top}
    \quad & \wovple \quad \tO\Paren{\Norm{\sum_{j\neq \ell,i} \iprod{a_\ell,a_j}^2 \|a_j\|^2 \cdot a_j a_j^\top}}^{1/2}\\
    & \leq \tO\Paren{{\max_{j\neq \ell} \iprod{a_\ell,a_j}^2 \|a_j\|^2 \Norm{\sum_{j} a_j a_j^\top}}}^{1/2}\\
    & \wovple \quad \tO\Paren{\frac{\sqrt{n}}{d}},
\end{align*}
where we have used that $\iprod{a_\ell,a_i}^2$ concentrates around $\frac{1}{d}$ (\pref{fact:SIP-simple}), that $\|a_i\|^2$ concentrates around 1 (\pref{fact:SIP-simple}), and that $\Norm{\sum_{i} a_i a_i^\top}$ concentrates around $\tfrac{n}{d}$ (\pref{lem:Mj-GC-simple}) within logarithmic factors all with overwhelming probability.

For the right-hand matrix, we use the fact that the summands are PSD to conclude that
\begin{align*}
\Norm{\sum_{i \neq \ell} \iprod{a_\ell, a_i}^2 \|a_i\|^2 \cdot a_i a_i^\top}
  & \leq \max_{i \neq \ell} \iprod{a_\ell, a_i}^2 \|a_i\|^2 \cdot \Norm{\sum_{i \neq \ell} a_i a_i^\top}\\
  & \wovple \quad \tO(1/d) \cdot \tO(n/d)\mcom
\end{align*}
using the same concentration facts as earlier.

Putting these together, \wovp
\[
  \|N_\ell\| \leq \tO(\sqrt n /d) \cdot \tO(\sqrt n /d) = \tO(n/d^2)\mper
\]
Now we are ready to make the final bound on $\Mdiff'$.
With overwhelming probability,
\[
  \|\Mdiff'\| \leq \tO(\sqrt n) \cdot \max_{\ell \in [n]} \|N_\ell\| \leq \tO(n^3 /d^4)^{1/2}
\]
and hence by \pref{thm:decoupling}, $\|\Mdiff\| \leq \tO(\sqrt n) \cdot \max_{\ell \in [n]} \|N_\ell\| \leq \tO(n^3 /d^4)^{1/2}$ \wovp.

Finally, by triangle inequality and all our bounds thus far, \wovp
\[
  \|\Mcross\| \leq  \|\Msame\| + \|\Mdiff\| \leq \tO(n/d^{1.5}) + \tO(n^3/d^4)^{1/2} \leq \tO(n^3/d^4)^{1/2}\mper\qedhere
\]
\end{proof}

\subsection{Full algorithm and proof of \pref{thm:tdecomp-alg}}\label{sec:alg-details}
In this subsection we give the full details of our tensor decomposition algorithm.
As discussed above, the algorithm proceeds by constructing a random matrix from the input tensor, then computing and post-processing its top eigenvector.

\begin{center}
\fbox{\begin{minipage}{\textwidth}
\begin{center}\textbf{Spectral Tensor Decomposition (One Attempt)}\end{center}
This is the main subroutine of our algorithm---we will run it $\tO(n)$ times and show that this recovers all of the components $a_1,\ldots,a_n$.
\begin{algo}
\label{alg:four-thirds}
Input: $\bT = \sum_{i = 1}^n a_i \tensor a_i \tensor a_i$.
Goal: Recover $a_i$ for some $i\in[n]$.
\begin{itemize}
\item
  Compute the matrix unfolding $T \in \R^{d^2\times d}$ of $\bT$.
  Then compute a $3$-tensor $\bS \in \R^{d^2 \times d^2 \times d^2}$ by starting with the $6$-tensor $\bT \tensor \bT$, permuting indices, and flattening to a $3$-tensor.
  Apply $T$ in one mode of $\bS$ to obtain $\bM \in \R^{d\tensor d^2 \tensor d^2}$, so that:
  \begin{align*}
    T = \sum_{i\in[n]} a_i(a_i\tensor a_i)^\top\mcom \qquad
    \bS = \bT^{\tensor 2} = \sum_{i,j = 1}^n (a_i \tensor a_j)^{\tensor 3}\mcom
  \end{align*}
  \[
	  \bM = \bS(T, \Id_{d^2}, \Id_{d^2})
      = \sum_{i,j \in [n]} T (a_i \tensor a_j) \tensor (a_i \tensor a_j) \tensor (a_i \tensor a_j) \mper
  \]
\item Sample a vector $g \in \R^{d}$ with iid standard gaussian entries.
  Evaluate $\bM$ in its first mode in the direction of $g$ to obtain $M \in \R^{d^2\times d^2}$:
    \[
	M := \bM(g,\Id_{d^2},\Id_{d^2}) = \sum_{i,j \in [n]} \langle g, T (a_i \tensor a_j)\rangle \cdot (a_i \tensor a_j)(a_i \tensor a_j)^\top
    \]
\item Let $\Sigma \defeq \E[(aa^\top)^{\tensor 2}]$ for ${a\sim\cN(0,\Id_d)}$.
    Let $\signh \defeq \sqrt{2}\cdot (\Sigma^{+})^{1/2}$.
Compute the top eigenvector $u\in \R^{d^2}$ of $\signh M\signh$, and reshape $\signh u$ to a matrix $U \in \R^{d\times d}$.
\item For each of the signings of the top 2 unit left (or right) singular vectors
    ${\pm}u_1, {\pm}u_2$ of $U$,
      check if $\sum_{i \in [n]} \iprod{a_i,{\pm}u_j}^3 \geq 1 - c(n,d)$ where $c(n,d) = \Theta(n/d^{3/2})$ is an appropriate threshold.
      If so, output ${\pm}u_j$.
      Otherwise output nothing.
\end{itemize}
\end{algo}
\end{minipage}}
\end{center}

\pref{thm:tensor-decomp-main} gets us most of the way to the correctness of \pref{alg:four-thirds}, proving that the top eigenvector of the matrix $RMR$ is correlated with some $a_i^{\tensor 2}$ with reasonable probability.
We need a few more ingredients to prove \pref{thm:tdecomp-alg}.
First, we need to show a bound on the runtime of \pref{alg:four-thirds}.
\begin{lemma}
  \label{lem:tdecomp-implementation}
  \pref{alg:four-thirds} can be implemented in time $\tO(d^{1+\omega})$,
  where $d^{\omega}$ is the runtime for multiplying two $d \times d$ matrices.
  It may also be implemented in time $\tO(d^{3.257})$.
\end{lemma}
\begin{proof}
    To run the algorithm, we only require access to power iteration using the matrix $RMR$.
    We first give a fast implementation for power iteration with the matrix $M$, and handle the multiplications with $R$ separately.

    Consider a vector $v \in \R^{d^2}$, and a random vector $g \sim \cN(0,\Id_d)$, and let $V,G \in \R^{d\times d}$ be the reshapings of $v$ and $gT$ respectively into matrices.
    Call $\bT_v = \bT(\Id_d,V,G)$, where we have applied $V$ and $G$ in the second and third modes of $\bT$, and call $T_v$ the reshaping of $\bT_v$ into a $d\times d^2$ matrix.
    We have that
    \[
	T_v = \sum_{i\in[n]}  a_i(V a_i \tensor Ga_i)^\top
    \]
    We show that the matrix-vector multiply $Mv$ can be computed as a flattening of the following product:
    \begin{align*}
	T_v T^\top
	&= \Paren{\sum_{i\in[n]}a_i (Va_i \tensor Ga_i)^\top}\Paren{\sum_{j\in[n]} (a_j \tensor a_j)a_j^\top}\\
	&= \sum_{i,j\in[n]} \iprod{a_j, Va_i}\cdot \iprod{a_j, Ga_i}\cdot a_i a_j^\top\\
	&= \sum_{i,j\in[n]} \iprod{a_i \tensor a_j,v} \cdot \iprod{gT, a_i \tensor a_j}\cdot a_i a_j^\top
    \end{align*}
    Flattening $T_v T^{\top}$ from a $d \times d$ matrix to a vector $v_{TT}\in \R^{d^2}$, we have that
    \[
	v_{TT} = \sum_{i,j\in[n]} \iprod{gT, a_i\tensor a_j} \cdot \iprod{a_i\tensor a_j, v}\cdot a_i \tensor a_j = Mv.
    \]
So we have that $Mv$ is a flattening of the product $T_v T^\top$, which we will compute as a proxy for computing $Mv$ via direct multiplication.

Computing $T_v = \bT(\Id,V,G)$ can be done with two matrix multiplication operations,
both times multiplying a $d^2 \times d$ matrix with a $d \times d$ matrix.
Computing $T_v T^\top$ is a multiplication of a $d \times d^2$ matrix by a $d^2 \times d$ matrix.
Both these steps may be done in time $O(d^{1 + \omega})$, by regarding the $d \times d^2$ matrices
as block matrices with blocks of size $d \times d$.
Alternatively, the asymptotically fastest known algorithm for rectangular matrix multiplication
gives a time of $O(d^{3.257})$ \cite{6375330}.

Now, to compute the matrix-vector multiply $RMRu$ for any vector $u\in \R^{d^2}$, we may first compute $v = Ru$, perform the operation $Mv$ in time $O(d^{1+\omega})$ as described above, and then again multiply by $R$.
The matrix $R$ is sparse: it has $O(d)$ entries per row (see \pref{fact:sigma}), so the multiplication $Ru$ requires time $O(d^3)$.

Performing the update $RMRv$ a total of $O(\log^2 n)$ times is sufficient for convergence, as we have that with reasonable probability, the spectral gap $\lambda_2(RMR)/\lambda_1(RMR) \le 1-O(\tfrac{1}{\log n})$, as a result of applying \pref{thm:tensor-decomp-main} with the choice of $\epsilon = O(\tfrac{1}{\log n})$.

Finally, checking the value of $\sum_i \iprod{a_i, x}^3$ requires $O(d^3)$ operations, and we do so a constant number of times, once for each of the signings of the top 2 left (or right) singular vectors of $U$.
\end{proof}

Next, we need to show that given $u$ with $\iprod{Ru, a_i \tensor a_i}^2 \geq (1 - \tO(n^{3/2}/\epsilon d^2)) \cdot \|u\|^2 \cdot \|a_i\|^4$ we can actually recover the tensor component $a_i$.
Here \pref{alg:four-thirds} reshapes $Ru$ to a $d \times d$ matrix and checks the top two left- or right-singular vectors; the next lemma shows one of these singular vectors must be highly correlated with $a_i$.
(The proof is deferred to \pref{sec:linalg}.)
\begin{lemma}
\torestate{
\label{lem:symflat}
    Let $M \in \R^{d^2 \times d^2}$ be a symmetric matrix with $\|M\| \le 1$, and let $v \in \R^d$ and $u \in \R^{d^2}$ be vectors.
    Furthermore, let $U$ be the reshaping of the vector $Mu \in \R^{d^2}$ to a matrix in $\R^{d \times d}$.
    Fix $c > 0$, and suppose that $\iprod{Mu, v \tensor v}^2 \ge c^2 \cdot\|u\|^2 \cdot \|v\|^4$.
    Then $U$ has some left singular vector $a$ and some right singular vector $b$ such that
    \[
	|\iprod{a,v}|, |\iprod{b,v}| \ge c \cdot \|v\|\mper
    \]
    Furthermore, for any $0 < \alpha < 1$, there are $a',b'$ among the top $\lfloor \tfrac{1}{\alpha c^2}\rfloor$ singular vectors of $U$ with
    \[
	|\iprod{a',v}|, |\iprod{b',v}| \ge \sqrt{1-\alpha} \cdot c \cdot \|v\| \mper
    \]
    If $c \ge \sqrt{\tfrac{1}{2}(1 + \eta)}$ for some $\eta > 0$, then $a,b$ are amongst the top $\lfloor\frac{(1+\eta)}{\eta c^2} \rfloor $ singular vectors.
}
\end{lemma}
Since here $c^2 = 1 - o(1)$, we can choose $\eta = 1 - o(1)$ and check only the top $2$ singular vectors.

Next, we must show how to choose the threshold $c(n,d)$ so that a big enough value $\sum_{i \in [n]} \iprod{a_i, u_j}^3$ is ensures that $u_j$ is close to a tensor component.
The proof is at the end of this section.
(A very similar fact appears in \cite{DBLP:conf/approx/GeM15}. We need a somewhat different parameterization here, but we reuse many of their results in the proof.)
\begin{lemma}\label{lem:tensor-check-success}
    Let $T = \sum_{i\in[n]} a_i \tensor a_i \tensor a_i$ for normally distributed vectors $a_i \sim \cN(0,\tfrac{1}{d}\Id_d)$.
    For all $0 < \gamma, \gamma' < 1$,
    \begin{enumerate}
      \item
    With overwhelming probability, for every $v \in \R^d$ such that $\sum_{i\in[n]} \iprod{a_i, v}^3 \ge 1 -\gamma$,
    \[
	\max_{i\in[n]} |\iprod{a_i,v}| \ge 1 - O(\gamma) - \tO(n/d^{3/2})\mper
    \]

  \item
    With overwhelming probability over $a_1,\ldots,a_n$ if $v \in \R^d$ with $\|v\| = 1$ satisfies $\iprod{v,a_j} \geq 1 - \gamma'$ for some $j$ then $\sum_i \iprod{a_i, v}^3 \geq 1 - O(\gamma') - \tO(n/d^{3/2})$.
\end{enumerate}
\end{lemma}

We are now ready to prove \pref{thm:tdecomp-alg}.

\begin{proof}[Proof of \pref{thm:tdecomp-alg}]
  By \pref{thm:tensor-decomp-main}, with probability $1 - o(1)$ over $a_1,\ldots,a_n$ there are events $E_1,\ldots,E_n$ so that $\Pr_g(E_i) \geq \tO(1/n^{1 + O(\epsilon)})$ such that when event $E_i$ occurs the top eigenvector $u$ of $RMR$ satisfies
    \[
      \frac{\iprod{Ru, a_i \tensor a_i}^2}{\|u\|^2 \cdot \|a_i\|^4} \ge 1 - \tO\Paren{\frac{n^{3/2}}{\epsilon d^2}}\mper
    \]
    For a particular sample $g \sim \cN(0,\Id_d)$, let $u_g$ be this eigenvector.

    The algorithm is as follows.
    Sample $g_1,\ldots,g_r \sim \cN(0,\Id_d)$ independently for some $r$ to be chosen later.
    Compute $Ru_{g_1},\ldots,Ru_{g_r}$, reshape each to a $d \times d$ matrix, and compute its singular value decomposition.
    This gives a family of (right) singular vectors $v_1,\ldots,v_{dr}$.
    For each, evaluate $\sum_i \iprod{a_i, v_j}^3$.
    Let $c(n,d)$ be a threshold to be chosen later.
    Initialize $S \subset \R^d$ to the empty set.
    Examining each $1 \leq j \leq dr$ in turn, add $v_j$ to $S$ if $\sum_i \iprod{a_i, v_j}^3 \geq 1 - c(n,d)$ and for every $v$ already in $S$, $\iprod{v,v_j}^2 \leq 1/2$.
    Output the set $S$.

    Choose $\epsilon = 1/\log n$.
    By \pref{lem:symflat}, when $E_i$ occurs for $g_j$ one of $v \in \{ \pm v_{jr},\ldots, \pm v_{(j+1)r} \}$ has $\iprod{v, a_i} \geq (1 - \tO(n^{3/2}/d^2))(\|u_j\|^2 \cdot \|a_j\|^4)$.
    Then by \pref{lem:tensor-check-success}, when $E_i$ occurs for $g_j$, this $v$ we will have $\sum_i \iprod{a_i, {\pm}v}^3 \geq 1 - \tO(n/d^{3/2})$.
    Choose $c(n,d) = \tilde \Theta(n^{3/2}/d^2)$ so that when $E_i$ occurs for $g_j$, so long as it has not previously occurred for some $j' < j$, the algorithm adds ${\pm}v$ to $S$.

    The events $E_{i}^{(t)}$ and $E_{i}^{(t')}$ are independent for any two executions of the algorithm $t$ and $t'$ and have probability $\tOmega(1/n)$.
    Thus, after $r = \tO(n)$ executions of the algorithm, with high probability for every $i \in [n]$ there is $j \in [r]$ so that $E_i$ occurs for $g_j$.
    Finally, by \pref{lem:tensor-check-success}, the algorithm can never add to $S$ a vector which is not $(1 - \tO(n/d^{3/2}))$-close to some $a_i$.
 \end{proof}

It just remains to prove \pref{lem:tensor-check-success}.

\begin{proof}[Proof of \pref{lem:tensor-check-success}]
  We start with the first claim.
    By \cite[Lemma 2, (proof of) Lemma 8, Theorem 4.2]{DBLP:conf/approx/GeM15}, the following inequalities all hold \wovp.
    \begin{align}
      \sum_{i \in n} \iprod{a_i, x}^4 & \leq 1 + \tO(n/d^{3/2}) \quad \text{for all $\|x\| =1$}\mcom \label{eq:check-1} \\
     \sum_{i \in [n]} \iprod{a_i,x}^6 & \geq 1 - O\Paren{\sum_{i \in [n]} \iprod{a_i,x}^3 -1} - \tO(n/d^{3/2}) \quad \text{for all $\|x\|=1$}\mcom \label{eq:check-2} \\
    \Abs{\sum_{i \in [n]} \iprod{a_i, x}^3} & \leq 1 + \tO(n/d^{3/2}) \quad \text{for all $\|x\|=1$}\mper \label{eq:check-3}
    \end{align}
    To begin,
    \[
	\sum_{i\in[n]} \iprod{a_i, v}^6 \le \Paren{\max_{i\in[n]} \iprod{a_i,v}^2} \cdot \Paren{\sum_{i\in[n]} \iprod{a_i, v}^4}.
    \]
    By \pref{eq:check-1}, this implies
    \begin{equation}
	\max_{i\in[n]} \iprod{a_i,v}^2 \ge (1 - \tO(n/d^{3/2}))\cdot \sum_{i\in[n]} \iprod{v,a_i}^6. \label{eq:two-siv}
    \end{equation}
    Now combining \pref{eq:check-2} with \pref{eq:two-siv} we have
    \begin{align*}
	\max_{i\in[n]} \iprod{a_i,v}^2
        &\ge (1-\tO(n/d^{3/2}))\cdot (1 - O(1 - \sum_i \iprod{a_i, v}^3 )- \tO(n/d^{3/2}))\mper
    \end{align*}
    Together with \pref{eq:check-3} this concludes the of the first claim.

    For the second claim, we note that by \pref{eq:check-3}, and homogeneity, $|\sum_{i \neq j} \iprod{a_i, x}^3| \leq \|x\|^3(1 + \tO(n/d^{3/2})$ \wovp.
    We write $v = \iprod{a_j, x} a_j + x^\perp$, where $\iprod{x^\perp, a_j} = 0$.
    Now we expand
    \begin{align*}
      \sum_i \iprod{a_i,v}^3 & \geq (1 - \gamma')^3 + \sum_{i \neq j} \iprod{\iprod{a_j,x} a_j + x^\perp, a_i}^3\\
                             & = ( 1- \gamma')^3 + \sum_{i \neq j} \iprod{a_j,x}^3 \iprod{a_j, a_i}^3 + 3 \iprod{a_j,x}^2 \iprod{a_j,a_i}^2 \iprod{x^\perp, a_i}\\
                             & + 3 \iprod{a_j,x}\iprod{a_j,a_i} \iprod{x^\perp a_i}^2 + \iprod{x^\perp, a_i}^3\mper
    \end{align*}
    We estimate each term in the expansion:
    \begin{align*}
      \Abs{\sum_{i \neq j} \iprod{a_j,x}^3 \iprod{a_j, a_i}^3} & \leq |\iprod{a_j, x}^3|\sum_{i \neq j} |\iprod{a_j, a_i}|^3 \leq \tO\Paren{\frac{n}{d^{3/2}}} \\
                                                         & \text{\wovp by Cauchy-Schwarz and standard concentration.}\\
      \Abs{\sum_{i \neq j} \iprod{a_j,x}^2 \iprod{a_j, a_i}^2 \iprod{x^\perp, a_i}} & \leq \Paren{\sum_{i \neq j} \iprod{a_j, x}^4 \iprod{a_j, a_i}^4}^{1/2} \Paren{\sum_{i \neq j} \iprod{x^\perp, a_i}^2}^{1/2} \quad \text {by Cauchy-Schwarz}\\
                                                                                    & \leq O(\sqrt n) \cdot \max_{i \neq j} \iprod{a_j, a_i}^2 \cdot \tO\Paren{\frac n d}^{1/2} \quad \text{\wovp by standard concentration.}\\
                                                                                    & \leq \tO\Paren{\frac n {d^{3/2}}} \quad \text{\wovp by standard concentration}\\
        \Abs{\sum_{i \neq j} \iprod{a_j,x} \iprod{a_j, a_i} \iprod{x^\perp, a_i}^2} & \leq O(1) \cdot \max_{i \neq j} |\iprod{a_j, a_i}| \cdot \sum_{i \neq j}\iprod{x^\perp, a_i}^2 \quad \text{\wovp by standard concentration}\\
      & \leq \tO\Paren{\frac 1 {\sqrt d}} \cdot \tO\Paren{\frac n d} \quad \text{\wovp by standard concentration}\\
      & \leq \tO\Paren{\frac n {d^{3/2}}}\\
      \Abs{\sum_{i \neq j} \iprod{x^\perp, a_i}^3} & \leq \gamma' + \tO\Paren{\frac n {d^{3/2}}} \quad \text{\wovp by \pref{eq:check-3} and homogeneity}\mper
    \end{align*}

    Now we estimate
    \[
      \sum_i \iprod{a_i, v}^3 \geq (1 - \gamma')^3 + \sum_{i \neq j} \iprod{a_i,x}^3 \geq (1 - \gamma)^3 - \gamma' - \tO(n/d^{3/2}) \geq 1 - O(\gamma') - \tO(n/d^{3/2})\mper
    \]
    since $\gamma' < 1$.
\end{proof}

\subsubsection{Boosting Accuracy with Local Search}
We remark that \pref{alg:four-thirds} may be used in conjunction with a local search
algorithm to obtain much greater guarantees on the accuracy of the recovered vectors.
Previous progress on the tensor decomposition problem has produced iterative
algorithms that provide local convergence guarantees given a good enough
initialization, but which leave the question of how to initialize the procedure
up to future work, or up to the specifics of an implementation.
In this context, our contribution can be seen as a general method of obtaining
good initializations for these local iterative procedures.

In particular, Anandkumar et al. \cite{DBLP:conf/colt/AnandkumarGJ15} give an algorithm
that combines tensor power iteration and a form of coordinate descent,
which when initialized with the output of \pref{alg:four-thirds}, achieves a linear
convergence rate to the true decomposition within polynomial time.

\begin{theorem}[Adapted from Theorem 1 in \cite{DBLP:conf/colt/AnandkumarGJ15}]\label{thm:noisy-tpi}
    Given a rank-$n$ tensor $\bT = \sum_{i} a_i \tensor a_i \tensor a_i$ with random Gaussian components $a_i \sim \cN(0, \tfrac{1}{d}\Id_d)$.
    There is a constant $c > 0$ so that if a set of unit vectors $\{x_i \in \R^d\}_i$ satisfies
    \[
   \iprod{x_i, a_i} \geq 1-c, \quad\forall i \in [n],
    \]
  then there exists a procedure which with overwhelming probability over
  $\bT$ and for any $\epsilon > 0$, recovers a set of vectors
  $\{\hat{a}_i\}$ such that
    \[
   \iprod{\hat{a}_i, a_i} \geq 1-\epsilon, \quad \forall i \in [n],
    \]
  in time $O(\poly(d) + nd^3\log \epsilon)$.

  \Jnote{}
\end{theorem}

\begin{remark}
Theorem 1 of Anandkumar et al. is stated for random asymmetric tensors,
but the adaptation to symmetric tensors is stated in equations (14) and (27)
in the same paper.

The theorem of Anandkumar et al. allows for a perturbation tensor $\Phi$, which
is just the zero tensor in our setting.
Additionally, the weight ratios specifying the weight of each rank-one
component in the input tensor are $w_{max} = w_{min} = 1$.
Lastly, the initialization conditions are given in terms of the distance between
the intialization vectors and the true vectors $|x_i - a_i|$,
which is related to our measure of closeness $\iprod{x_i,a_i}$ by the
equation $|x_i - a_i|^2 = |x_i|^2 + |a_i|^2 - 2\iprod{x_i,a_i}$.

The linear convergence guarantee is stated in Lemma 12 of Anandkumar et al.

\end{remark}

\begin{corollary}[Corollary of \pref{thm:tdecomp-alg}]\label{cor:boost}
  Given as input the tensor $\bT = \sum_{i=1}^n a_i \tensor a_i \tensor a_i$ where $a_i \sim \cN(0,\tfrac{1}{d} \Id_d)$ with $d \leq n \leq d^{4/3} / \polylog d$, there is a polynomial-time algorithm which with probability $1 - o(1)$ over the input $\bT$ and the algorithm randomness finds unit vectors $\hat{a}_1, \ldots, \hat{a}_n \in \R^d$ such that for all $i\in[n]$,
    \[
      \iprod{\hat{a}_i, a_i} \ge 1 - O\Paren{2^{-n}}\mper
    \]
\end{corollary}
\begin{proof}
  We repeatedly invoke \pref{alg:four-thirds} until we obtain a full set of $n$ vectors as characterized by \pref{thm:tdecomp-alg}.
  Apply \pref{thm:noisy-tpi} to the recovered set of vectors until the desired accuracy is obtained.
\end{proof}

\section{Tensor principal component analysis}\label{sec:tpca}
The Tensor PCA problem in the spiked tensor model is similar to the
setting of tensor decomposition, but here the goal is to recover
a single large component with all smaller components of the
tensor regarded as random noise.

\begin{problem}[Tensor PCA in the Order-$3$ Spiked Tensor Model]
\label{prob:spiked-tensor}
  Given an input tensor $\bT = \tau \cdot v^{\tensor 3} + \bA$, where $v \in \R^n$ is an arbitrary unit vector, $\tau \geq 0$ is the signal-to-noise ratio, and $\bA$ is a random noise tensor with iid standard Gaussian entries, recover the signal $v$ approximately.
\end{problem}
Using the partial trace method,
we give the first linear-time algorithm for this problem that
recovers $v$ for signal-to-noise ratio $\tau = O(n^{3/4}/\poly \log n)$.
In addition, the algorithm requires only $O(n^2)$ auxiliary space (compared to the input size of $n^3$) and uses only one non-adaptive pass over the input.

\subsection{Spiked tensor model}
This spiked tensor model (for general order-$k$ tensors) was introduced by
Montanari and Richard \cite{DBLP:conf/nips/RichardM14},
who also obtained the first algorithms to solve the model with provable
statistical guarantees.
Subsequently, the SoS approach was applied to the model to improve the
signal-to-noise ratio required for odd-order tensors
\cite{DBLP:conf/colt/HopkinsSS15}; for $3$-tensors reducing the requirement
from $\tau = \Omega(n)$ to $\tau = \Omega(n^{3/4} \log(n)^{1/4})$.

Using the linear-algebraic objects involved in the analysis of the SoS
relaxation, the previous work has also described algorithms with guarantees
similar to those of the SoS SDP relaxation,
while requiring only nearly subquadratic or linear time \cite{DBLP:conf/colt/HopkinsSS15}.

The algorithm here improves on the previous results by use of the partial
trace method,
simplifying the analysis and improving the runtime by a factor of $\log n$.

\subsection{Linear-time algorithm}
\begin{center}
\fbox{\begin{minipage}{\textwidth}
\begin{center}\textbf{Linear-Time Algorithm for Tensor PCA}\end{center}
\begin{algo}
\label{alg:tpca}
Input: $\bT = \tau \cdot v^{\tensor 3} + \bA$.
Goal: Recover $v'$ with $\iprod{v,v'} \geq 1 - o(1)$.
\begin{itemize}
\item
Compute the partial trace $M \seteq \Tr_{\R^n} \sum_i T_i \tensor T_i \in \R^{n \times n}$, where $T_i$ are the first-mode slices of $\bT$.
\item Output the top eigenvector $v'$ of $M$.
\end{itemize}
\end{algo}
\end{minipage}}
\end{center}

\begin{theorem}
  \label{thm:tpca-success}
  When $\bA$ has iid standard Gaussian entries and $\tau \geq C n^{3/4} \log(n)^{1/2}/\epsilon$ for some constant $C$,
  \pref{alg:tpca} recovers $v'$ with $\iprod{v,v'} \geq 1 - O(\epsilon)$ with high probability over $\bA$.
\end{theorem}
\begin{theorem}
  \label{thm:tpca-speed}
  \pref{alg:tpca} can be implemented in linear time and sublinear space.
\end{theorem}
These theorems are proved by routine matrix concentration results, showing
that in the partial trace matrix, the signal dominates the noise.

To implement the algorithm in linear time it is enough to show that this
(sublinear-sized) matrix has constant spectral gap;
then a standard application of the matrix power method
computes the top eigenvector.

\begin{lemma}
  \label{lem:tpca-concentration}
  For any $v$, with high probability over $\bA$, the following occur:
  \begin{align*}
    \left \| \sum_i \Tr(A_i) \cdot A_i \right \| & \leq O(n^{3/2} \log^2 n)\\
        \left \| \sum_i v(i) \cdot A_i \right \| & \leq O(\sqrt n \log n)\\
    \left \|\sum_i \Tr(A_i) v(i) \cdot vv^\top \right \| & \leq O(\sqrt n \log n)\mper
  \end{align*}
\end{lemma}
The proof may be found in \pref{app:tensor-pca}.

\begin{proof}[Proof of \pref{thm:tpca-success}]
  We expand the partial trace $\Tr_{\R^n} \sum_i T_i \tensor T_i$.
  \begin{align*}
    \Tr_{\R^n} \sum_i T_i \tensor T_i & = \sum_i \Tr(T_i) \cdot T_i\\
    & = \sum_i \Tr(\tau \cdot v(i) vv^\top + A_i) \cdot (\tau \cdot v(i) vv^\top + A_i)\\
    & = \sum_i (\tau v(i)\|v\|^2 + \Tr(A_i)) \cdot (\tau \cdot v(i) vv^\top + A_i)\\
    & = \tau^2 vv^\top + \tau \Paren{ \sum_i v(i) \cdot A_i + \sum_i \Tr(A_i) v(i) vv^\top} + \sum_i \Tr(A_i) \cdot A_i\mper
  \end{align*}
  Applying \pref{lem:tpca-concentration} and the triangle inequality, we see that
  \[
    \left \|\tau \Paren{ \sum_i v(i) \cdot A_i + \sum_i \Tr(A_i) v(i) vv^\top} + \sum_i \Tr(A_i) \cdot A_i \right \| \leq O(n^{3/2} \log n)
  \]
  with high probability.
  Thus, for appropriate choice of $\tau = \Omega(n^{3/4} \sqrt{(\log n)/\epsilon})$,
  the matrix $\Tr_{\R^n} \sum_i T_i \tensor T_i$ is close to rank one, and the result follows by standard manipulations.
\end{proof}
\begin{proof}[Proof of \pref{thm:tpca-speed}]
  Carrying over the expansion of the partial trace from above
  and setting $\tau = O(n^{3/4} \sqrt{(\log n)/\epsilon})$,
  the matrix $\Tr_{\R^n} \sum_i T_i \tensor T_i$ has a spectral gap ratio
  equal to $\Omega(1/\epsilon)$ and so the matrix
  power method finds the top eigenvector in $O(\log (n/\epsilon))$ iterations.
  This matrix has dimension $n \times n$, so a single iteration takes $O(n^2)$ time, which is sublinear in the input size $n^3$.
  Finally, to construct $\Tr_{\R^n} \sum_i T_i \tensor T_i$ we use
  \[
    \Tr_{\R^n} \sum_i T_i \tensor T_i = \sum_i \Tr(T_i) \cdot T_i
  \]
  and note that to construct the right-hand side it is enough to examine each entry of $\bT$ just $O(1)$ times and perform $O(n^3)$ additions.
  At no point do we need to store more than $O(n^2)$ matrix entries at the same time.
\end{proof}

\section*{Acknowledgements}
We would like to thank Rong Ge for very helpful discussions.
We also thank Jonah Brown Cohen, Pasin Manurangsi and Aviad Rubinstein for helpful comments in the preparation of this manuscript.

\addreferencesection
\bibliographystyle{amsalpha}
\bibliography{bib/mr,bib/dblp,bib/scholar,bib/sos-speedups,bib/zblatt}

\appendix

\section{Additional preliminaries}

\subsection{Linear algebra}
\label{sec:linalg}

Here we provide some lemmas in linear algebra.

This first lemma is closely related to the sos Cauchy-Schwarz from \cite{DBLP:conf/stoc/BarakKS14}, and the proof is essentially the same.
\begin{lemma}[PSD Cauchy-Schwarz]
  \label{lem:op-cs-block}
  Let $M \in \R^{d \times d}$, $M \succeq 0$ and symmetric.
  Let $p_1, \ldots, p_n, q_1,\ldots,q_n \in \R^d$.
  Then
  \[
    \iprod{M, \sum_{i=1}^n p_i q_i^\top} \leq \iprod{M, \sum_{i=1}^n p_i p_i^\top}^{1/2} \iprod{M, \sum_{i=1}^n q_i q_i^\top}^{1/2}\mper
  \]
\end{lemma}
In applications, we will have $\sum_i p_i q_i$ as a single block of a larger block matrix
containing also the blocks $\sum_i p_i p_i^\top$ and $\sum_i q_i q_i^\top$.
\begin{proof}
  We first claim that
  \[
    \iprod{M, \sum_{i=1}^n p_i q_i^\top} \leq \frac 1 2 \iprod{M, \sum_{i=1}^n p_i p_i^\top}
    + \frac 1 2 \iprod{M, \sum_{i=1}^n q_i q_i^\top}\mper
  \]
  To see this, just note that the right-hand side minus the left is exactly
  \[
    \iprod{M, \sum_{i=1}^n (p_i - q_i)(p_i - q_i)^\top} = \sum_i (p_i - q_i)^\top M (p_i - q_i) \geq 0\mper
  \]
  The lemma follows now be applying this inequality to
  \[
    p_i' = \frac{p_i}{\iprod{M, \sum_{i=1}^n p_i p_i^\top}^{1/2}} \qquad
    q_i' = \frac{q_i}{\iprod{M, \sum_{i=1}^n q_i q_i^\top}^{1/2}}\mper\qedhere
  \]
\end{proof}

\begin{lemma}[Operator Norm Cauchy-Schwarz for Sums]
\label{lem:op-cs}
  Let $A_1,\ldots,A_m, B_1,\ldots,B_m$ be real random matrices.
  Then
  \[
    \left \| \sum_i \E A_i B_i \right \| \leq \left \| \sum_i \E A_i^\top A_i \right \|^{1/2} \left \| \sum_i \E B_i^\top B_i \right \|^{1/2}\mper
  \]
\end{lemma}
\begin{proof}
  We have for any unit $x, y$,
  \begin{align*}
    x^\top \sum_i \E A_i B_i x & = \sum_i \E \iprod{A_i x, B_i y}\\
    & \leq \sum_i \E \|A_i x\| \|B_i y\|\\
    & \leq \sum_i (\E \|A_i x\|^2)^{1/2} (\E \|B_i x\|^2)^{1/2}\\
    & \leq \sqrt{\sum_i \E \|A_i x\|^2 } \sqrt{\sum_i \E \|B_i y\|^2}\\
    & = \sqrt{\E x^\top \sum_i A_i^\top A_i x} \sqrt{\E y^\top \sum_i B_i^\top B_i y}\\
    & \leq \left \| \sum_i \E A_i^\top A_i \right \|^{1/2} \left \| \sum_i \E B_i^\top B_i \right \|^{1/2}\mper
  \end{align*}
  where the nontrivial inequalities follow from Cauchy-Schwarz for expectations, vectors and scalars, respectively.
\end{proof}

The followng lemma allows to argue about the top eigenvector of matrices with spectral gap.

\begin{lemma}[Top eigenvector of gapped matrices]
  \label{lem:low-correlation}
  Let $M$ be a symmetric $r$-by-$r$ matrix and let $u,v$ be a vectors in $\R^r$ with $\lVert  u \rVert=1$.
  Suppose $u$ is a top singular vector of $M$ so that $\abs{\langle  u,M u \rangle} = \lVert  M \rVert$ and $v$ satisfies for some $\e>0$,
  \begin{displaymath}
    \lVert  M  - v v ^\top \rVert \le \lVert  M \rVert - \e \cdot \lVert v \rVert^2
  \end{displaymath}
  Then, $\langle  u,v \rangle^2 \ge \e \cdot \lVert  v \rVert^2$.
\end{lemma}

\begin{proof}

We lower bound the quadratic form of $M- vv^\top$ evaluated at $u$ by
\begin{displaymath}
  \lvert \langle  u, (M - v v^\top) u \rangle\rvert
  \ge \lvert \langle  u,Mu \rangle \rvert -  \langle  u,v \rangle^2
  = \lVert  M  \rVert - \langle  u,v \rangle^2\mper
\end{displaymath}
At the same time, this quadratic form evaluated at $u$ is upper bounded by $\lVert  M \rVert - \e \cdot \lVert  v \rVert^2$.
It follows that $\langle  u,v \rangle^2 \ge \e\cdot \lVert  v \rVert^2$ as desired.

\end{proof}

The following lemma states that a vector in $\R^{d^2}$ which is close to a
symmetric vector $v^{\tensor 2}$, if flattened to a matrix, has top eigenvector
correlated with the symmetric vector.
\restatelemma{lem:symflat}
\begin{proof}
    Let $\hat{v} = v / \|v\|$.
    Let $(\sigma_i, a_i, b_i)$ be the $i$th singular value, left and right (unit) singular vectors of $U$ respectively.

Our assumptions imply that
\[
    \left| \hat v^\top U \hat v \right| = \left| \iprod{M u, \hat v\tensor \hat v} \right| \ge c \cdot \| u\|.
\]
Furthermore, we observe that $\|U\|_F = \|M u\| \le \|M \| \cdot \|u\|$, and that therefore $\|U\|_F \le \|u \|$.
We thus have that,
\[
    c \cdot \|u\|
    \le \left| \hat v^\top U \hat v \right|
    = \left| \sum_{i \in [d]} \sigma_i \cdot \iprod{\hat v, a_i}\iprod{\hat v,b_i} \right|
    \le \|u\|\cdot \sqrt{\sum_{i \in [d]} \iprod{\hat v, a_i}^2\iprod{\hat v, b_i}^2},
\]
where to obtain the last inequality we have used Cauchy-Schwarz and our bound on $\|U\|_F$.
We may thus conclude that
\begin{align}
    c^2
    &\le \sum_{i \in [d]} \iprod{\hat v, a_i}^2 \iprod{\hat v, b_i}^2
    \le \max_{i \in [d]} \iprod{a_i,\hat v}^2\cdot \sum_{i\in[d]} \iprod{b_i,\hat v}^2 = \max_{i\in[d]} \iprod{a_i,\hat v}^2\mcom\label{eq:large-cor}
\end{align}
where we have used the fact that the left singular values of $U$ are orthonormal.
The argument is symmetric in the $b_i$.

Furthermore, we have that
\[
    c^2 \cdot \|u\|^2
    \le \left| \hat v^\top U \hat v \right|^2
    = \left| \sum_{i \in [d]} \sigma_i \cdot \iprod{\hat v, a_i}\iprod{\hat v,b_i} \right|^2
\le \left(\sum_{i \in [d]} \sigma_i^2 \iprod{\hat v, a_i}^2\right) \cdot \left(\sum_{i\in[d]}\iprod{\hat v, b_i}^2\right)
= \sum_{i \in [d]} \sigma_i^2 \iprod{\hat v, a_i}^2,
\]
where we have applied Cauchy-Schwarz and the orthonormality of the $b_i$.
In particular,
\[ \sum_{i \in [d]} \sigma_i^2 \iprod{\hat v, a_i}^2 \ge c^2 \|u\|^2 \ge c^2 \|U\|_F^2 \mper \]
On the other hand,
let $S$ be the set of $i\in[d]$ for which $\sigma_i^2 \le \alpha c^2 \|U\|_F^2$.
By substitution,
\[
    \sum_{i \in S} \sigma_i^2 \iprod{\hat v, a_i}^2
    \le \alpha c^2 \|U\|_F^2 \sum_{i\in S} \iprod{\hat v, a_i}^2
    \le \alpha c^2 \|U\|_F^2 \mcom
\]
where we have used the fact that the right singular vectors are orthonormal.
The last two inequalities imply that $S \ne [d]$.
Letting $T = [d] \setminus S$, it follows from subtraction that
\[
    (1-\alpha) c^2 \|U\|_F^2
    \le \sum_{i\in T} \sigma_i^2 \iprod{\hat v, a_i}^2
    \le \max_{i\in T} \iprod{\hat v, a_i}^2 \sum_{i\in T} \sigma_i^2
    = \max_{i\in T} \iprod{\hat v, a_i}^2 \|U\|_F^2 \mcom
\]
so that $\max_{i \in T} \iprod{\hat v, a_i}^2 \ge (1-\alpha)c^2$.
Finally,
\[
    |T| \cdot \alpha c^2 \|U\|_F^2
    \le |T| \cdot \min_{i \in T} \sigma_i^2
    \le \sum_{i\in[d]} \sigma_i^2
    = \|U\|_F^2 \mcom
\]
so that $|T| \le \lfloor \frac{1}{\alpha c^2}\rfloor$.
Thus, one of the top $\lfloor \frac{1}{\alpha c^2}\rfloor$ right singular vectors $a$ has correlation $|\iprod{\hat v, a}| \ge \sqrt{(1-\alpha)}c$.
The same proof holds for the $b$.

Furthermore, if $c^2 > \tfrac{1}{2}(1 + \eta)$ for some $\eta > 0$, and $(1-\alpha)c^2 > \tfrac{1}{2}$, then by \pref{eq:large-cor} it must be that $\max_{i\in T} \iprod{\hat v, a_i}^2 = \max_{i\in [d]} \iprod{\hat v, a_i}^2$, as $\hat v$ cannot have square correlation larger than $\tfrac{1}{2}$ with more than one left singular vector.
Taking $\alpha = \tfrac{\eta}{1+\eta}$ guarantees this.
The conclusion follows.
\end{proof}
\subsection{Concentration tools}
We require a number of tools from the literature on concentration of measure.

\subsubsection{For scalar-valued polynomials of Gaussians}
We need the some concentration bounds for certain polynomials of Gaussian random variables.

The following lemma gives standard bounds on the tails of a standard gaussian
variable---somewhat more precisely than other bounds in this paper.
Though there are ample sources, we repeat the proof here for reference.
\begin{lemma}
    \label{lem:gtail}
    Let $X \sim \cN(0,1)$. Then for $t > 0$,
    \[
	\Pr\Paren{X > t} \le \frac{e^{-t^2/2}}{t\sqrt{2\pi}},
    \]
    and
    \[
	\Pr\Paren{ X > t} \ge
	\frac{e^{-t^2/2}}{\sqrt{2\pi}}\cdot\left(\frac{1}{t} - \frac{1}{t^3}\right).
    \]
\end{lemma}
\begin{proof}
    To show the first statement, we apply an integration trick,
    \begin{align*}
	\Pr\Paren{X > t}
	&= \frac{1}{\sqrt{2\pi}} \int_{t}^{\infty} e^{-x^2/2} dx\\
	&\le \frac{1}{\sqrt{2\pi}}  \int_{t}^{\infty} \frac{x}{t} e^{-x^2/2} dx\\
	&= \frac{e^{-t^2/2}}{t\sqrt{2\pi}},
    \end{align*}
    where in the third step we have used the fact that $\frac{x}{t} \le x$ for
    $t\ge x$. For the second statement, we integrate by parts and repeat the
    trick,
    \begin{align*}
	\Pr\Paren{X > t}
	&= \frac{1}{\sqrt{2\pi}} \int_{t}^{\infty} e^{-x^2/2} dx\\
	&= \frac{1}{\sqrt{2\pi}} \int_{t}^{\infty} \frac{1}{x} \cdot x e^{-x^2/2} dx\\
	&= \frac{1}{\sqrt{2\pi}}\left[-\frac{1}{x}e^{-x^2/2} \cdot \right]_{t}^{\infty} -
	\frac{1}{\sqrt{2\pi}} \int_{t}^{\infty} \frac{1}{x^2} \cdot e^{-x^2/2} dx\\
	&\ge \frac{1}{\sqrt{2\pi}}\left[-\frac{1}{x}e^{-x^2/2} \cdot \right]_{t}^{\infty} -
	\frac{1}{\sqrt{2\pi}} \int_{t}^{\infty} \frac{x}{t^3} \cdot e^{-x^2/2} dx\\
	&= \frac{1}{\sqrt{2\pi}}
	\left(\frac{1}{t} - \frac{1}{t^3}\right)e^{-t^2/2}.
    \end{align*}
    This concludes the proof.
\end{proof}

The following is a
small modification of Theorem 6.7 from \cite{janson1997gaussian} which follows from Remark 6.8 in the same.
\begin{lemma}
  \label{lem:gaus-polys}
  For each $\ell \geq 1$ there is a universal constant $c_\ell > 0$
  such that for every $f$ a degree-$\ell$ polynomial of standard Gaussian random variables $X_1,\ldots,X_m$ and $t \geq 2$,
  \[
    \Pr(|f(X)| > t\E |f(X)|) \leq e^{-c_\ell t^{2/\ell}}\mper
  \]
  The same holds (with a different constant $c_\ell$) if $\E |f(x)|$ is replaced by $(\E f(x)^2)^{1/2}$.
\end{lemma}

In our concentration results, we will need to calculate the expectations of
multivariate Gaussian polynomials, many of which share a common form. Below we
give an expression for these expectations.
\begin{fact}
    \label{fact:gaussian-poly}
Let $x$ be a $d$-dimensional vector with independent identically distributed gaussian entries
with variance $\sigma^2$. Let $u$ be a fixed unit vector.
Then setting $X = (\| x\|^2 - c)^p \|x\|^{2m} xx^T$, and setting $U =(\| x\|^2 -
c)^p \|x\|^{2m} uu^T$,  we have
\begin{align*}
    \ex{ X}
    &= \left(\sum_{0 \le k \le p} \binom{p}{k} (-1)^k c^k
(d+2)\cdots(d + 2p +2 m - 2k) \sigma^{2(p + m - k + 1)}\right) \cdot \Id,
\intertext{and}
    \ex{ U}
    &= \left(\sum_{0 \le k \le p} \binom{p}{k} (-1)^k c^k
d(d+2)\cdots(d + 2p + 2m - 2k -2 ) \sigma^{2(p + m - k)}\right) \cdot uu^T
\end{align*}
\end{fact}
\begin{proof}
    \begin{align*}
    \E[X]
    &= \E[(\| x\|^2 - c)^p \|x\|^{2m} x_1^2 ] \cdot \Id
    \\
    &= \Id \cdot \sum_{0 \le k \le p} \binom{p}{k} (-1)^k c^k
    \E\left[\left(\sum_{\ell \in [d]} x_i^2 \right)^{p + m - k}
    x_1^2\right]
\intertext{Since $\left(\sum_{i \in [d]} x_i^2 \right)^{p + m - k}$ is
symmetric in $x_1,\ldots,x_d$, we have }
    &= \Id \cdot \frac{1}{d}\sum_{0 \le k \le p} \binom{p}{k} (-1)^k c^k
    \E\left[\left(\sum_{i \in [d]} x_i^2 \right)^{p + m - k + 1} \right]\\
    \intertext{We have reduced the computation to a question of the moments of a Chi-squared variable with $d$ degrees of freedom. Using these moments,}
    &= \Id \cdot \frac{1}{d}\sum_{0 \le k \le p} \binom{p}{k} (-1)^k c^k
    d(d+2)\cdots(d + 2p + 2m - 2k) \sigma^{2(p + m - k + 1)}\\
    &= \Id \cdot \left(\sum_{0 \le k \le p} \binom{p}{k} (-1)^k c^k
(d+2)\cdots(d + 2p + 2m - 2k) \sigma^{2(p + m - k + 1)}\right).
    \end{align*}
    A similar computation yields the result about $\ex{U}$.
\end{proof}

\subsubsection{For matrix-valued random variables}
On several occasions we will need to apply a Matrix-Bernstein-like theorem to a
sum of matrices with an unfortunate tail. To this end, we prove a ``truncated
Matrix Bernstein Inequality.''
Our proof uses an standard matrix Bernstein inequality as a black box.
The study of inequalities of this variety---on tails of sums of independent matrix-valued random variables---
was initiated by Ahlswede and Winter \cite{DBLP:journals/tit/AhlswedeW02}.
The excellent survey of Tropp \cite{DBLP:journals/FOCM/Tropp12} provides many results of this kind.

In applications of the following the operator norms of the summands $X_1,\ldots,X_n$ have well-behaved tails and so the truncation is a routine formality.
Two corollaries following the proposition and its proof capture truncation for all the matrices we encounter in the present work.

\begin{proposition}[Truncated Matrix Bernstein]
    \label{prop:truncated-bernstein}
    Let $X_1,\ldots, X_n \in \R^{d_1 \times d_2}$ be independent random
    matrices, and suppose that
    \[
	\Pr\left[ \left\Vert X_i - \E[X_i]\right\Vert_{op} \ge \beta \right]
	\le p \text{ for all } i \in [n].
    \]
    Furthermore, suppose that for each $X_i$,
    \[
	\left\Vert \E[X_i] -
	\E[X_i \Ind\left[\Vert X_i \Vert_{op} < \beta\right]] \right\Vert
	\le q.
    \]
    Denote
    \[
	\sigma^2
	= \max\left\{
	    \left\Vert \sum_{i\in[n]}\E\left[ X_i X_i^T \right]-\E\left[
	X_i\right]\E\left[X_i^T \right]\right\Vert_{op},
	    \left\Vert \sum_{i\in[n]}\E\left[ X_i^T X_i \right]-\E\left[
	X_i\right]^T\E\left[X_i \right]\right\Vert_{op}
    \right\}.
    \]
    Then for $X = \sum_{i\in[n]} X_i$, we have
    \[
	\Pr\left[
	    \Vert X - \E[X]\Vert_{op} \ge t
	\right]
	\le  n \cdot p
	+ (d_1 + d_2) \cdot \exp\left(\frac{-(t-nq)^2}{2(\sigma^2 + \beta (t-nq) / 3)}\right).
    \]

    \begin{proof}
        For simplicity we start by centering the variables $X_i$.
        Let $\tilde{X_i} = X_i - \E X_i$ and $\tilde X = \sum_{i \in [n]} \tilde X_i$
	The proof proceeds by a straightforward application of
	the noncommutative Bernstein's Inequality.
	We define variables $Y_1,\ldots, Y_n$, which are the truncated
	counterparts of the $\tilde X_i$s in the following sense:
	\[
	    Y_i = \begin{cases}
		\tilde X_i & \Vert \tilde X_i \Vert_{op} < \beta,\\
		0 & \text{otherwise.}
	    \end{cases}
	\]
	Define $Y = \sum_{i\in[n]} Y_i$.
        We claim that
        \begin{align}
          & \op{\sum_i \E Y_i Y_i^T - \E[Y_i]\E[Y_i]^T} \leq \op{\sum_i \E \tilde X_i \tilde X_i^T} \leq \sigma^2 \text{ and } \label{eqn:trunc-dom} \\
          & \op{\sum_i \E Y_i^T Y_i - \E[Y_i]^T \E[Y_i]} \leq \op{\sum_i \E \tilde X_i^T \tilde X_i} \leq \sigma^2\mcom \label{eqn:trunc-dom2}
        \end{align}
        which, together with the fact that $\|Y_i\| \leq \beta$ almost surely,
        will allow us to apply the non-commutative Bernstein's inequality to $Y$.
        To see \pref{eqn:trunc-dom} (\pref{eqn:trunc-dom2} is similar), we expand $\E Y_i Y_i^T$ as
        \[
          \E Y_i Y_i^T = \Pr\left [\op{\tilde X_i} < \beta \right ]\E \left
	  [\tilde X_i \tilde X_i^T\ \bigg{|}\  \op{\tilde X_i} < \beta \right ]\mper
        \]
        Additionally expanding $\E \left [ \tilde X_i \tilde X_i^T \right ]$ as
        \[
          \E \left [ \tilde X_i \tilde X_i^T \right ]
          = \Pr\left [\op{\tilde X_i} < \beta \right ]\E \left [\tilde X_i
	  \tilde X_i^T \ \bigg{|}\  \op{\tilde X_i} < \beta \right ]
           + \Pr\left [\op{\tilde X_i} \geq \beta \right ]\E \left [\tilde X_i
	   \tilde X_i^T\ \bigg{|}\  \op{\tilde X_i} \geq \beta \right ]\mcom
        \]
        we note that $\E[ \tilde X_i \tilde X_i^T\  | \ \op{\tilde X_i} \geq \beta ]$ is PSD.
        Thus, $\E[Y_i Y_i^T] \succeq \E[X_i X_i^T]$.
        But by definition $\E[Y_i Y_i^T]$ is still PSD
        (and hence $\op{\sum_i \E[Y_i Y_i^T]}$ is given by the maximum eigenvalue of $\E[Y_i Y_i^T]$), so
        \[
          \op{\sum_i \E Y_i Y_i^T} \leq \op{\sum_i \E \tilde X_i \tilde X_i^T}\mper
        \]
        Also PSD are $\E[Y_i]\E[Y_i]^T$ and $\E[(Y_i - \E[Y_i])(Y_i - \E[Y_i])^T] = \E[Y_i Y_i^T] - \E[Y_i]\E[Y_i]^T$.
        By the same reasoning again, then, we get $\op{\sum_i \E Y_i Y_i^T - \E[Y_i]\E[Y_i]^T} \leq \op{\sum_i \E[Y_i Y_i^T]}$.
        Putting this all together gives \pref{eqn:trunc-dom}.

	Now we are ready to apply the non-commutative Bernstein's inequality to $Y$.
        We have
	\[
	    \Pr\left[ \Vert Y - \E[Y] \Vert_{op} \ge \alpha \right]
	    \le
	    (d_1 + d_2)\cdot \exp\left(\frac{-\alpha^2/2}{\sigma^2 + \beta
	    \cdot \alpha / 3}\right).
	\]
	Now, we have
	\begin{align*}
	    \Pr\left[ \Vert X - \E[X] \Vert_{op} \ge t \right]
	    &=
	    \Pr\left[ \Vert X - \E[X] \Vert_{op} \ge t \ \vert \ X = Y\right]
	    \cdot \Pr\left[ X = Y\right]
	    \\ +
        & \Pr\left[ \Vert X - \E[X] \Vert_{op} \ge t \ \vert \ X \neq Y\right]
	    \cdot \Pr\left[ X \neq Y\right],\\
	    &\le
	    \Pr\left[ \Vert X - \E[X] \Vert_{op} \ge t \ \vert \ X = Y\right]
	    + n \cdot p
	\end{align*}
        by a union bound over the events $\{ X_i \neq Y_i \}$.
	It remains to bound the conditional probability
	$\Pr\left[ \Vert X - \E[X] \Vert_{op} \ge t \ \vert \ X = Y\right]$.
	By assumption, $\Vert \E[X] - \E[Y] \Vert_{op} \le nq$, and so by the
	triangle inequality,
	\[
	    \Vert X - \E[X] \Vert_{op}
	    \le \Vert X - \E[Y]\Vert_{op} + \Vert \E[Y] - \E[X]\Vert_{op} \le
	    \Vert X - \E[Y]\Vert_{op} + nq.
	\]
	Thus,
	\begin{align*}
	\Pr\left[ \Vert X - \E[X] \Vert_{op} \ge t \ \vert \ X = Y\right]
	&\le
	\Pr\left[ \Vert X - \E[Y] \Vert_{op} + nq \ge t \ \vert \ X = Y\right]\\
	&=
	\Pr\left[ \Vert Y - \E[Y] \Vert_{op} \ge t - nq \ \vert \ X = Y\right].
	\end{align*}
	Putting everything together and setting $\alpha = t - nq$,
	\[
	    \Pr[\Vert X - \E[X] \Vert_{op} \ge t]
	    \le n\cdot p + (d_1 + d_2)\cdot
	    \exp\left(\frac{-(t-nq)^2/2}{\sigma^2 + \beta (t-nq)/3}\right),
	\]
	as desired.
    \end{proof}

\end{proposition}

The following lemma helps achieve the assumptions of \pref{prop:truncated-bernstein} easily
for a useful class of thin-tailed random matrices.
\begin{lemma} \label{lem:truncate}
    Suppose that $X$ is a matrix whose entries are polynomials of constant degree $\ell $ in unknowns $x$, which we evaluate at independent Gaussians.
     Let $f(x) := \|X\|_{op}$ and $g(x) := \|XX^T\|_{op}$, and either $f$ is itself a polynomial in $x$
     of degree at most $2\ell$ or $g$ is a polynomial in $x$ of degree at most $4\ell$.
     Then if $\beta = R\cdot \alpha$ for $\alpha \ge \min\{\Ex{|f(x)|}, \sqrt{\Ex{g(x)}}\}$ and $R = \polylog(n)$,
    \begin{align}
	\Pr\left(\|X\|_{op} \ge \beta \right) \le n^{-\log n},\label{eq:normbd}\\
	\intertext{ and }
	\Ex{\|X \cdot \Ind\{\|X\|_{op} \ge \beta \}\|_{op}} \le (\beta + \alpha) n^{-\log n}\mper \label{eq:expbd}
    \end{align}
\end{lemma}

\begin{proof}
    We begin with \pref{eq:normbd}. Either $f(x)$ is a polynomial of degree at
    most $2\ell$, or $g(x)$ is a polynomial of degree at most $4\ell$ in
    gaussian variables.
    We can thus use \pref{lem:gaus-polys} to obtain the following bound,
    \begin{align}
	\Pr\left(|f(x)| \ge t \alpha \right)
	\le \exp\left(-c t^{1/(2\ell)} \right),\label{eq:B-tail}
    \end{align}
    where $c$ is a universal constant. Taking $t = R = \polylog(n)$
    gives us \pref{eq:normbd}.

    We now address \pref{eq:expbd}.
To this end, let $p(t)$ and $P(t)$ be the probability
density function and cumulative density function of $\|X\|_{op}$, respectively. We apply Jensen's inequality and instead bound
\begin{align*}
\| \Ex{X \Ind\{\|X\|_{op} \ge \beta \}} \|
&\le  \Ex{\| X  \|_{op} \Ind\{\|X\|_{op} \ge \beta \}}
\ =  \int_{0}^{\infty} t \cdot \Ind\{t \ge \beta \} p(t) dt
\intertext{since the indicator is 0 for $t\le \beta$,}
&=  \int_{\beta}^{\infty} (-t) (-p(t)) dt
\intertext{integrating by parts,}
&=  -t \cdot (1-P(t))\bigg{|}_{\beta}^{\infty} + \int_{\beta}^{\infty} (1-P(t)) dt
\intertext{
  and using the equality of $1-P(t)$ with $\Pr(\|X\|_{op} \ge t)$ along with
\pref{eq:normbd},}
&\le \beta n^{-\log n} + \int_{\beta}^{\infty} \Pr(\|X\|_{op} \ge t)  dt
\intertext{
    Applying the change of variables $t = \alpha s$ so as to apply \pref{eq:B-tail},}
&=  \beta n^{-\log n} + \alpha \int_{R}^{\infty} \Pr(\|X\|_{op} \ge \alpha s)  ds 
\\ &\le \beta n^{-\log n} + \alpha \int_{R}^{\infty} \exp(-c s^{1/(2\ell)}) ds
    \intertext{Now applying a change of variables so $s = (\frac{u \log
    n}{c})^{2\ell}$,}
    & = \beta n^{-\log n} +
    \alpha \int_{\tfrac{c R^{1/(2\ell)}}{\log n} }^\infty
      n^{-u} \cdot
      2\ell \left(\frac{\log n}{c}\right)^{2\ell} u^{2\ell -1}
    du \\
    & \le \beta n^{-\log n} + \alpha \int_{\tfrac{c R^{1/(2\ell)}}{\log n}}^\infty   n^{-u/2}  du \mcom \\
\intertext{where we have used the assumption that $\ell$ is constant.
We can approximate this by a geometric sum,}
&\le \beta n^{-\log n} + \alpha
\sum_{u = \tfrac{c R^{1/(2\ell)}}{\log n} }^\infty   n^{-u/2}  \\
& \le \beta n^{-\log n} +
\alpha\cdot n^{-c R^{1/(2\ell)}/ (2 \log n)}  \\
\end{align*}
Evaluating at $R = \polylog n$ for a sufficiently large polynomial in the log gives us
\[
\Ex{\|X \cdot \Ind\{\|X\|_{op} \ge \beta\}\|_{op}} \le (\beta + \alpha) n^{-\log n},
\]
as desired.
\end{proof}

\section{Concentration bounds for planted sparse vector in random linear subspace}
\label{sec:fast-planted-proofs}

\begin{proof}[Proof of \pref{lem:fps-covariance-concentration}]
  Let $c \seteq \sum_{i=1}^n v(i) b_i$.
  The matrix in question has a nice block structure:
  \[
    \sum_{i=1}^n a_i a_i^\top
    = \left ( \begin{array}{cc} \|v\|_2^2 & c^\top \\ c & \sum_{i=1}^n b_i b_i^\top \end{array} \right )\mper
  \]
  The vector $c$ is distributed as $\cN(0,\tfrac 1 n \Id_{d-1})$ so by standard concentration has
  $\|c\| \leq \tO(d/n)^{1/2}$ \wovp.
  By assumption, $\|v\|_2^2 = 1$.
  Thus by triangle inequality \wovp
  \[
    \left \| \sum_{i=1}^n a_i a_i^\top - \Id_d \right \| \leq \tO\Paren{\frac d n}^{1/2}
    + \left \| \sum_{i=1}^n b_i b_i^\top - \Id_{d-1} \right \|\mper
  \]
  By \cite[Corollary 5.50]{DBLP:journals/corr/abs-1011-3027} applied to the subgaussian vectors $n b_i$, \wovp
  \[
    \left \| \sum_{i=1}^n b_i b_i^\top - \Id_{d-1} \right \| \leq O\Paren{\frac d n}^{1/2}
  \]
  and hence $\| \sum_{i=1}^n a_i a_i^\top - \Id_d \| \leq \tO(d/n)^{1/2}$ \wovp.
  This implies $\| (\sum_{i=1}^n a_i a_i^\top)^{-1} - \Id_d \| \leq \tO(d/n)^{1/2}$
  and $\| (\sum_{i=1}^n a_i a_i^\top)^{-1/2} - \Id_d \| \leq \tO(d/n)^{1/2}$
  when $d = o(n)$ by the following facts applied to the eigenvalues of $\sum_{i=1}^n a_i a_i^\top$.
  For $0 \leq \epsilon < 1$,
  \begin{align*}
    (1 + \epsilon)^{-1} & = 1 - O(\epsilon) \quad \text{ and } \quad (1 - \epsilon)^{-1} = 1 + O(\epsilon)\mcom\\
  (1 + \epsilon)^{-1/2} & = 1 - O(\epsilon) \quad \text { and } \quad (1 - \epsilon)^{-1/2} = 1 + O(\epsilon)\mper
  \end{align*}
  These are proved easily via the identity $(1 + \epsilon)^{-1} = \sum_{k =1}^\infty \epsilon^k$ and similar.
\end{proof}

\subsection*{Orthogonal subspace basis}

\begin{lemma}
\label{lem:ortho-1}
  Let $a_1,\ldots,a_n\in \R^d$ be independent random vectors from $N(0,\tfrac 1 n \Id)$ with $d\le n$ and let $A=\sum_{i=1}^n a_i a_i^\top$.
  Then for every unit vector $x\in \R^d$, with overwhelming probability $1-d^{-\omega(1)}$,
  \begin{displaymath}
    \Abs{ \langle  x , A^{-1} x \rangle -\lVert  x \rVert^2 }
    \le \tilde O\left(\frac {d + \sqrt n } n \right)\cdot \lVert  x \rVert^2\mper
  \end{displaymath}
\end{lemma}

\begin{proof}
Let $x \in \R^d$.
By scale invariance, we may assume $\lVert  x \rVert =1$.

By standard matrix concentration bounds, the matrix $B=\Id - A$ has spectral norm $\lVert  B \rVert\le \tilde O(d/n)^{1/2}$ \wovp \cite[Corollary 5.50]{DBLP:journals/corr/abs-1011-3027}.
Since $A^{-1} = (\Id - B)^{-1} = \sum_{k=0}^\infty B^k$, the spectral norm of $A^{-1} - \Id - B$ is at most $\sum_{k=2}^\infty \lVert  B \rVert^k$ (whenever the series converges).
Hence, $\lVert A^{-1} - \Id - B\rVert \le \tilde O(d/n)$ \wovp.

It follows that it is enough to show that $\lvert \langle  x, B x \rangle\rvert \le \tilde O(1/n)^{1/2}$ \wovp.
The random variable $n - n \langle  x, B x \rangle=\sum_{i=1}^n \langle \sqrt n\cdot  a_i, x \rangle^2$ is $\chi^2$-distributed with $n$ degrees of freedom.
Thus, by standard concentration bounds, $n \lvert \langle  x, B x \rangle\rvert\le \tilde O(\sqrt n)$ \wovp \cite{laurent2000}.

We conclude that with overwhelming probability $1-d^{-\omega(1)}$,
\begin{displaymath}
  \Abs{ \iprod{ x, A^{-1} x} - \lVert  x \rVert^2}
  \le \Abs{\iprod{x, B x}} + \tilde O(d/n)
  \le \tilde O \left( \frac {d+\sqrt n} n \right)\mper
\end{displaymath}
\end{proof}

\begin{lemma}
\label{lem:ortho-2}
  Let $a_1,\ldots,a_n\in \R^d$ be independent random vectors from $N(0,\tfrac 1 n \Id)$ with $d\le n$ and let $A=\sum_{i=1}^n a_i a_i^\top$.
  Then for every index $i\in [n]$, with overwhelming probability $1-d^{\omega(1)}$,
  \begin{displaymath}
    \Abs{ \langle  a_j , A^{-1} a_j \rangle -\lVert  a_j \rVert^2 }
    \le \tilde O\left(\frac {d + \sqrt n } n \right)\cdot \lVert  a_j \rVert^2\mper
  \end{displaymath}
\end{lemma}

\begin{proof}
  Let $A_{-j} = \sum_{i\neq j} a_i a_i^\top$.
  By Sherman--Morrison,
  \begin{displaymath}
    A^{-1}
    = (A_{-j } + a_j a_j^\top)^{-1}
    =  A_{-j}^{-1} - \frac 1 {1+a_j^\top  A_{-j}^{-1} a_j }  A_{-j}^{-1} a_j a_j^\top A_{-j}^{-1}
  \end{displaymath}
  Thus, $\langle  a_j, A^{-1} a_j \rangle = \langle a_j, A_{-j}^{-1} a_j  \rangle - \langle a_j, A_{-j}^{-1} a_j  \rangle^2/ ( 1+ \langle a_j, A_{-j}^{-1} a_j  \rangle)$.
  Since $\lVert  \tfrac {n} {n-1} A_{-j} - \Id  \rVert = \tilde O(d/n)^{1/2}$ \wovp, we also have $\lVert  A_{-j}^{-1} \rVert\le 2$ with overwhelming probability.
  Therefore, \wovp,
  \begin{displaymath}
    \Abs{ \langle  a_j, A^{-1} a_j \rangle - \langle a_j, A_{-j}^{-1} a_j  \rangle } \le  \langle a_j, A_{-j}^{-1} a_j  \rangle^2 \le 4 \lVert  a_j \rVert^4 \le \tilde O(d / n) \cdot \lVert  a_j \rVert^2\mper
  \end{displaymath}
  At the same time, by \pref{lem:ortho-1}, \wovp,
  \begin{displaymath}
    \Abs{\langle a_j, \tfrac n {n-1} A_{-j}^{-1} a_j  \rangle -  \lVert  a_j \rVert^2 } \le \tilde O\left( \frac {d +\sqrt n} n \right) \cdot \lVert  a_j \rVert^2\mper
  \end{displaymath}
  We conclude that, \wovp,
  \begin{align*}
    \Abs{ \langle  a_j , A^{-1} a_j \rangle -\lVert  a_j \rVert^2 }
    &\le \Abs{ \langle  a_j , A^{-1} a_j \rangle - \langle a_j, A_{-j}^{-1} a_j  \rangle }
      + \Abs{ \langle a_j, A_{-j}^{-1} a_j \rangle - \tfrac {n-1} n \lVert  a_j \rVert^2 }
      + \tfrac 1 {n} \lVert  a_j \rVert^2
    \\
    & \le \tilde O \left ( \frac { d + \sqrt n } n \right )\mper
  \end{align*}
\end{proof}

\begin{lemma}
\label{lem:ortho-3}
  Let $A$ be a block matrix where one of the diagonal blocks is the $1 \times 1$ identity; that is,
  \[
    A = \Paren{ \begin{array}{cc} \|v\|^2 & c^\top \\ c & B \end{array} } = \Paren{ \begin{array}{cc} 1 & c^\top \\ c & B \end{array} } \mper
  \]
  for some matrix $B$ and vector $c$.
  Let $x$ be a vector which decomposes as $x = (x(1) \, \, x')$ where $x(1) = \iprod{x,e_1}$ for $e_1$ the first standard basis vector.

  Then
  \[
    \iprod{x, A^{-1} x} = \iprod{x', \Paren{ B^{-1} + \frac{B^{-1} cc^\top B^{-1}}{1 - c^\top B^{-1} c}} x'}
      + 2x(1) \iprod{ \Paren{ B^{-1} + \frac{B^{-1} cc^\top B^{-1}}{1 - c^\top B^{-1} c}} c, x'} + (1 - c^\top B^{-1} c)^{-1} x(1)^2\mper
  \]
\end{lemma}
\begin{proof}
  By the formula for block matrix inverses,
  \[
    A^{-1} = \Paren{ \begin{array}{cc} (1 - c^\top B^{-1} c)^{-1} & c^T(B - cc^\top)^{-1} \\ (B - cc^\top)^{-1} c & (B - cc^\top)^{-1} \end{array} }\mper
  \]
  The result follows by Sherman-Morrison applied to $(B - cc^\top)^{-1}$ and the definition of $x$.
\end{proof}

\begin{lemma}
\label{lem:ortho-4}
  Let $v \in \R^n$ be a unit vector and let $b_1,\ldots,b_n \in \R^{d-1}$ have iid entries from $\cN(0,1/n)$.
  Let $a_i \in \R^d$ be given by $a_i \seteq (v(i) \, \, b_i)$.
  Let $A \seteq \sum_i a_i a_i^T$.
  Let $c \in \R^{d-1}$ be given by $c \seteq \sum_i v(i) b_i$.
  Then for every index $i \in [n]$, \wovp,
  \[
    \left | \iprod{a_i, A^{-1} a_i} - \|a_i\|^2 \right | \leq \tO\Paren{\frac{ d + \sqrt{n}}{n}} \cdot \|a_i\|^2\mper
  \]
\end{lemma}
\begin{proof}
  Let $B \seteq \sum_i b_i b_i^T$.
  By standard concentration, $\|B^{-1} - \Id \| \leq \tO(d/n)^{1/2}$ \wovp \cite[Corollary 5.50]{DBLP:journals/corr/abs-1011-3027}.
  At the same time, since $v$ has unit norm, the entries of $c$ are iid samples from $\cN(0,1/n)$, and hence $n \|c\|^2$ is $\chi^2$-distributed with $d$ degrees of freedom.
  Thus \wovp $\|c\|^2 \leq \tfrac d n + \tO(dn)^{-1/2}$.
  Together these imply the following useful estimates, all of which hold \wovp:
  \begin{align*}
    |c^\top B^{-1} c| & \leq \|c\|^2 \|B^{-1}\|_{op}  \leq \frac d n + \tO\Paren{\frac d n}^{3/2}\\
    \|B^{-1} cc^\top B^{-1}\|_{op} & \leq \|c\|^2 \|B^{-1}\|_{op}^2 \leq \frac d n + \tO\Paren{\frac d n}^{3/2}\\
    \left \| \frac{B^{-1} cc^\top B^{-1}}{1 - c^\top B^{-1} c} \right \|_{op} & \leq \frac d n + \tO\Paren{\frac d n}^{3/2}\mcom
  \end{align*}
  where the first two use Cauchy-Schwarz and the last follows from the first two.

  We turn now to the expansion of $\iprod{a_i, A^{-1} a_i}$ offered by \pref{lem:ortho-3},
  \begin{align}
    \iprod{a_i, A^{-1} a_i} = & \iprod{b_i, \Paren{ B^{-1} + \frac{B^{-1} cc^\top B^{-1}}{1 - c^\top B^{-1} c}} b_i} \label{eqn:ortho-1}\\
      & + 2v(i) \iprod{ \Paren{ B^{-1} + \frac{B^{-1} cc^\top B^{-1}}{1 - c^\top B^{-1} c}} c, b_i} \label{eqn:ortho-2}\\
      & + (1 - c^\top B^{-1} c)^{-1} v(i)^2 \label{eqn:ortho-3}\mper
  \end{align}
  Addressing \pref{eqn:ortho-1} first, by the above estimates and \pref{lem:ortho-2} applied to $\iprod{b_i, B^{-1} b_i}$,
  \[
    \left | \iprod{b_i, \Paren{ B^{-1} + \frac{B^{-1} cc^\top B^{-1}}{1 - c^\top B^{-1} c}} b_i} - \|b_i\|^2 \right |
    \leq \tO\Paren{\frac {d + \sqrt n}{n}} \cdot \|b_i\|^2
  \]
  \wovp.
  For \pref{eqn:ortho-2}, we pull out the important factor of $\|c\|$ and separate $v(i)$ from $b_i$: \wovp,
  \begin{align*}
    \left | 2v(i) \iprod{ \Paren{ B^{-1} + \frac{B^{-1} cc^\top B^{-1}}{1 - c^\top B^{-1} c}} c, b_i} \right |
    & = \left | 2 \|c\| v(i) \iprod{ \Paren{ B^{-1} + \frac{B^{-1} cc^\top B^{-1}}{1 - c^\top B^{-1} c}} \frac c {\|c\|}, b_i} \right |\\
    & \leq \left | \|c\|^2 \Paren{ v(i)^2 + \iprod{ \Paren{ B^{-1} + \frac{B^{-1} cc^\top B^{-1}}{1 - c^\top B^{-1} c}} \frac c {\|c\|}, b_i}^2} \right |\\
    & \leq \tO\Paren{\frac d n } (v(i)^2 + \|b_i\|^2)\\
    & = \tO\Paren{\frac d n} \|a_i\|^2\mcom
  \end{align*}
  where the last inequality follows from our estimates above and Cauchy-Schwarz.

  Finally, for \pref{eqn:ortho-3}, since $(1 - c^\top B^{-1} c) \geq 1 - \tO(d/n)$ \wovp, we have that
  \[
    | (1 - c^\top B^{-1} c)^{-1} v(i)^2 - v(i)^2 | \leq \tO\Paren{\frac d n} v(i)^2\mper
  \]

  Putting it all together,
  \begin{align*}
     \left | \iprod{a_i, A^{-1} a_i} - \|a_i\|^2 \right | & \leq \left | \iprod{b_i, \Paren{ B^{-1} + \frac{B^{-1} cc^\top B^{-1}}{1 - c^\top B^{-1} c}} b_i} - \|b_i\|^2 \right |\\
     & + \left | 2v(i) \iprod{ \Paren{ B^{-1} + \frac{B^{-1} cc^\top B^{-1}}{1 - c^\top B^{-1} c}} c, b_i} \right |\\
     & + | (1 - c^\top B^{-1} c)^{-1} v(i)^2 - v(i)^2 |\\
     & \leq \tO\Paren{\frac{d + \sqrt n}{n}} \cdot \|a_i\|^2\mper\qedhere
  \end{align*}
\end{proof}

\section{Concentration bounds for overcomplete tensor decomposition}
\label{sec:tensor-decomp-aux}
We require some facts about the concentration of certain scalar-and matrix-valued random variables, which generally follow from standard concentration arguments.
We present proofs here for completeness.

The first lemma captures standard facts about random Gaussians.

\begin{fact}
    \label{fact:SIP}
    Let $a_1,\ldots, a_n\in\R^d$ be sampled $a_i \sim \cN(0,\tfrac 1 d \Id)$.
    \begin{enumerate}
      \item \label{item:aux-1}
Inner products $|\iprod{a_i, a_j}|$ are all $\approx 1/\sqrt d$:
    \[
	\Pr \left\{ \iprod{a_i, a_j}^2 \le \tO\left(\frac{1}{d}\right)
        \tallsep \forall i,j\in[n],\ i\neq j \right\} \geq 1 - n^{-\omega(1)}.
    \]
  \item
    \label{item:aux-2}
    Norms are all about $\|a_i\| \approx 1 \pm \tO(1/\sqrt d)$:
    \[
	\Pr \left\{ 1 - \tO(1/\sqrt{d}) \le \|a_i\|_2^2 \le 1 +
	    \tO(1/\sqrt{d}) \tallsep
	\forall i \in[n] \right\} \geq 1 - \owl\mper
    \]
  \item
    \label{item:aux-3}
    \label{fact:gaussian-iprods}
    Fix a vector $v \in \R^d$.
    Suppose $g\in\R^d$ is a vector with entries identically distributed $g_i \sim
    \cN(0,\sigma)$.
    Then $\iprod{g,v}^2 \approx \sigma^2\cdot \|v\|_2^2$:
    \[
	\Pr \left\{ \bigg|\iprod{g,v}^2 - \sigma^2\cdot \|v\|_2^4\bigg| \le \tO(\sigma^2\cdot \|v\|_4^2)
	\right\} \geq 1 - \owl\mper
    \]
\end{enumerate}
\end{fact}

\begin{proof}[Proof of \pref{fact:SIP}]
  We start with \pref{item:aux-1}. Consider the quantity $\iprod{a_i,a_j}^2$. We calculate the expectation,
    \[
	\Ex{\iprod{a_i,a_j}^2}
	= \sum_{k,\ell\in[d]} \Ex{a_i(k) a_i(\ell) a_j(k) a_j(\ell)}
	= \sum_{k\in[d]} \Ex{a_i(k)^2}\cdot \Ex{a_j(k)^2}
	= d \cdot \frac{1}{d^2}
	= \frac{1}{d}.
    \]
    Since this is a degree-4 square polynomial in the entries of $a_i$ and
    $a_j$, we may apply \pref{lem:gaus-polys} to conclude that
    \[
	\Pr\left( \iprod{a_i,a_j}^2 \ge t \cdot \frac{1}{d} \right) \le
	\exp\left( - O(t^{1/2})\right).
    \]
    Applying this fact with $t = \polylog(n)$ and taking a union bound over
    pairs $i,j\in[n]$ gives us the desired result.

    Next is \pref{item:aux-2}.
    Consider the quantity $\|a_i\|_2^2$.
    We will apply \pref{lem:gaus-polys} in order to
    obtain a tail bound for the value of the polynomial $(\|a_i\|_2^2 - 1)^2$.
    We have
    \[
	\Ex{(\|a_i\|_2^2 - 1)^2} = O\Paren{\frac{1}{d}},
    \]
    and now applying \pref{lem:gaus-polys} with the square root of this
    expectation, we have
    \[
	\Pr\Paren{ \left| \|a_i\|_2^2 - 1 \right| \ge \tO(\tfrac{1}{\sqrt{d}}) } \le n^{-\log n}
	\mper
    \]
    This gives both bounds for a single $a_i$.
    The result now follows from taking a union bound over all $i$.

    Moving on to \pref{item:aux-3},
    we view the expression $f(g) := (\iprod{g,v}^2 - \sigma^2 \|v\|^2)^2$ as a polynomial in the gaussian
    entries of $g$. The degree of $f(g)$ is 4, and $\E[|f(g)|] = 3\sigma^4 \cdot \|v\|_4^4$,
    and so we may apply
    \pref{lem:gaus-polys} to conclude that
    \[
	\Pr\Paren{|f(g)| \ge t\cdot 3\sigma^4 \cdot \|v\|_4^4} \le \exp(-c_4t^{1/2}),
    \]
    and taking $t = \polylog(n)$ the conclusion follows.
    \end{proof}
We also use the fact that the covariance matrix of a sum
of sufficiently many gaussian outer products concentrates about its expectation.
\begin{fact}
    \label{fact:GC}
    Let $a_1,\ldots, a_n\in\R^d$ be vectors with iid gaussian entries such that
    $\Ex{\|a_i\|_2^2} = 1$, and $n = \Omega(d)$.
    Let $\Event$ be the event that the sum $\sum_{i\in[n]} a_i a_i^\top$ is close
    to $\tfrac{n}{d} \cdot \Id$, that is
    \[
	\Pr \left\{ \tOmega(n/d) \cdot \Id \le \sum_{i\in[n]} a_i a_i^\top
	\le \tO(n/d)\cdot \Id \right\} \geq 1 - \owl\mper
    \]
\end{fact}
\begin{proof}[Proof of \pref{fact:GC}]
    We apply a truncated matrix bernstein inequality. For convenience,
    $A := \sum_{i \in [n]} a_i a_i^\top$ and
    let $A_i := a_i a_i^\top$ be a single summand.
    To begin, we calculate the first and second moments of the summands,
    \begin{align*}
	\Ex{A_i} &= \frac{1}{d} \cdot \Id \\
	\Ex{A_iA_i^\top} &= O\Paren{\frac{1}{d}} \cdot \Id.
    \end{align*}
    So we have $\Ex{A} = \frac{n}{d} \cdot \Id$ and $\sigma^2(A) =
    O\Paren{\frac{n}{d}}$.

    We now show that each summand is well-approximated by a truncated variable.
    To calculate the expected norm $\|A_i\|_{op}$, we observe that $A_i$ is
    rank-$1$ and thus $\Ex{\|A_i\|_{op}} = \Ex{\| a_i\|_2^2} = 1$.
    Applying \pref{lem:truncate}, we have
    \[
	\Pr\Paren{\|A_i\|_{op} \ge \tO(1)} \le n^{-\log n},
    \]
    and also
    \[
	\Ex{\|A_i\|_{op}\cdot \Ind\{\|A_i\|_{op} \ge \tO(1)\}} \le n^{-\log n}.
    \]

    Thus, applying the truncated matrix bernstein inequality from
    \pref{prop:truncated-bernstein} with
    $\sigma^2 = O(\tfrac{n}{d})$,
    $\beta = \tO(1)$,
    $p = n^{-\log n}$,
    $q = n^{-\log n}$, and
    $t = \tO\Paren{ \frac{n^{1/2}}{d^{1/2}}}$,
    we have that with overwhelming probability,
    \[
	\left\| A - \frac{n}{d}\cdot \Id \right\|_{op} \le \tO\Paren{\frac{ n^{1/2}}{
	d^{1/2}}} \mper
    \]
\end{proof}

We now show that among the terms of the polynomial $\iprod{g,Ta_i^{\tensor 2}}$,
those that depend on $a_j$ with $j \neq i$ have small magnitude.
This polynomial appears in the proof that $\Mdiag$ has a noticeable spectral gap.
\restatelemma{lem:signal-negligible}
\begin{proof}
    Fixing $a_i$ and $g$, the terms in the summation are independent, and we may apply a Bernstein inequality.
    A straightforward calculation shows that the expectation of the sum is 0 and the variance is $\tO(\tfrac{n}{d^2})\cdot \|g\|^2\|a_i\|^4$.
    Additionally, each summand is a polynomial in Gaussian variables, the square of which has expectation
    $\tO(\tfrac{1}{d^2}\cdot\|g\|^2\|a_i\|^4)$.
    Thus \pref{lem:gaus-polys} allows us to truncate each summand appropriately so as to employ \pref{prop:truncated-bernstein}.
    An appropriate choice of logarithmic factors and the concentration of $\|g\|^2$ and $\|a_i\|^2$ due to \pref{fact:SIP}
    gives the result for each $i \in [n]$.
    A union bound over each choice of $i$ gives the final result.
\end{proof}

Finally, we prove that a matrix which appears in the expression for $\Msame$ has bounded norm \wovp

\begin{lemma}
  \label{lem:Mj-bound}
  Let $a_1,\ldots,a_n$ be independent from $\cN(0,\tfrac 1 d \Id_d)$.
  Let $g \sim \cN(0,\Id_d)$.
  Fix $j \in [n]$.
  Then \wovp
  \[
    \Norm{\sum_{\substack{ \inn \\ i \neq j}} \iprod{g,a_i} \|a_i\|^2 \iprod{a_i, a_j} \cdot a_i a_i^\top} \leq \tO(n/d^2)^{1/2}\mper
  \]
\end{lemma}

\begin{proof}
  The proof proceeds by truncated matrix Bernstein, since the summands are independent for fixed $g, a_j$.
  For this we need to compute the variance:
  \begin{align*}
    \sigma^2 & = \Norm{\sum_{\substack{\inn \\ i \neq j}} \E \iprod{g,a_i}^2 \|a_i\|^6 \iprod{a_i, a_j}^2 \cdot a_i a_i^\top}
    \leq O(1/d) \cdot \Norm{\sum_{\substack{\inn \\ i \neq j}} \E a_i a_i^\top}
    \leq O(1/d) \cdot n/d
    \leq O(n/d^2)\mper
  \end{align*}
  The norm of each term in the sum is bounded by a constant-degree polynomial of Gaussians.
  Straightforward calculations show that in expectation each term is $O(\tfrac 1 d \iprod{g, a_i})$ in norm; \wovp this is $O(\sigma)$.
  So \pref{lem:gaus-polys} applies to establish the hypothesis of truncated Bernstein \pref{prop:truncated-bernstein}.
  In turn, \pref{prop:truncated-bernstein} yields that \wovp
  \[
    \Norm{\sum_{\substack{ \inn \\ i \neq j}} \iprod{g,a_i} \|a_i\|^2 \iprod{a_i, a_j} \cdot a_i a_i^\top} \leq \tO(\sigma) = \tO(n/d^2)^{1/2}\mper
  \]
\end{proof}

\subsubsection{Proof of \pref{fact:sigma}}
Here we prove the following fact.
\begin{fact}
  \label{fact:sigma}
  Let $\Sigma = \E _{x \sim \cN(0,\Id_d)} (x x^\top)^{\otimes 2}$ and let $\tilde \Sigma = \E_{x \sim \cN(0,\Id_d)} (xx^\top)^{\otimes 2} / \|x\|^4$.
  Let $\Phi = \sum_i e_i^{\otimes 2} \in \R^{d^2}$ and let $\Pisym$ be the projector to the symmetric subspace of $\R^{d^2}$ (the span of vectors of the form $x^{\otimes 2}$ for $x\in \R^d$).
    Then
    \begin{align*}
      \Sigma\hphantom{^+} &= 2\, \Pisym +  \Phi \Phi^\top\mcom
      &
      \tilde\Sigma\hphantom{^+} &= \tfrac{2}{d^2 + 2d} \Pisym + \tfrac{1}{d^2 + 2d} \Phi \Phi^\top\mcom
      \\
      \Sigma^+ &= \tfrac{1}{2} \Pisym - \tfrac{1}{2(d+2)} \Phi \Phi^\top\mcom
      &
      {\tilde\Sigma}^+ &= \tfrac{d^2 + 2d}{2} \Pisym - \tfrac{d}{2} \Phi \Phi^\top\mper
    \end{align*}
    In particular,
    \[
      R = \sqrt{2}\, (\Sigma^+)^{1/2} = \Pisym - \tfrac{1}{d} \Paren{1 - \sqrt{\tfrac{2}{d+2}}} \, \Phi \Phi^\top \quad \text{has} \quad \|R \| = 1
    \]
    and for any $v \in \R^d$,
    \[
      \|\signh (v\tensor v) \|_2^2 = \left(1-\tfrac{1}{d+2}\right)\cdot \|v\|^4.
    \]
\end{fact}

We will derive \pref{fact:sigma} as a corollary of a more general claim about rotationally symmetric distributions.
\begin{lemma}
\label{lem:cov-rot-sym}
  Let $\cD$ be a distribution over $\R^d$ which is rotationally symmetric; that is, for any rotation $R$, $x \sim \cD$ is distributed identically to $Rx$.
  Let $\Sigma = \E_{x \sim \cD} (x x^\top)^{\otimes 2}$, let $\Phi = \sum_i e_i^{\otimes 2} \in \R^{d^2}$ and let $\Pisym$ be the projector to the symmetric subspace of $\R^{d^2}$ (the span of vectors of the form $x^{\otimes 2}$ for $x\in \R^d$).
  Then there is a constant $r$ so that
  \begin{displaymath}
    \Sigma = 2r \, \Pisym + r \, \Phi \Phi^\top \mper
  \end{displaymath}
  Furthermore, $r$ is given by
  \[
  r = \E \iprod{x, a}^2 \iprod{x, b}^2  = \tfrac{1}{3} \E \iprod{x, a}^4
  \]
  where $a,b$ are orthogonal unit vectors.
\end{lemma}
\begin{proof}
  First, $\Sigma$ is symmetric and operates nontrivially only on the symmetric subspace
  (in other words $\ker \Pisym \subseteq \ker \Sigma$).
  This follows from $\Sigma$ being an expectation over symmetric matrices whose
  kernels always contain the complement of the symmetric subspace.

  Let $\hat{a},\hat{b},\hat{c},\hat{d} \in \R^d$ be any four orthogonal unit vectors.
  Let $R$ be any rotation of $\R^d$ that takes $\hat{a}$ to $-\hat{a}$, but fixes $\hat{b}$, $\hat{c}$, and $\hat{d}$
  (this rotation exists for $d \ge 5$, but a different argument holds for $d \le 4$) \Jnote{}.
  By rotational symmetry about $R$, all of these quantities are $0$:
  \[ \E \iprod{\hat{a},x}\iprod{\hat{b},x}\iprod{\hat{c},x}\iprod{\hat{d},x} = 0, \]
  \[ \E \iprod{\hat{a},x}\iprod{\hat{b},x}\iprod{\hat{c},x}^2 = 0, \qquad
     \E \iprod{\hat{a},x}\iprod{\hat{b},x}^3 = 0. \]

  Furthermore, let $Q$ be a rotation of $\R^d$ that takes $\hat{a}$ to $(\hat{a}+\hat{b})/\sqrt{2}$.
  Then by rotational symmetry about $Q$,
  \begin{align*} \E \iprod{\hat{a},x}^4 = \E \iprod{\hat{a},Qx}^4 = \E \tfrac{1}{4} \iprod{\hat{a}+\hat{b},x}^4
    = \E \tfrac{1}{4}[\iprod{\hat{a},x}^4 + \iprod{\hat{b},x}^4 + 6 \iprod{\hat{a},x}^2\iprod{\hat{b},x}^2]
  \end{align*}
  Thus, since $\E \iprod{\hat{a},x}^4 = \E \iprod{\hat{b},x}^4$ by rotational symmetry, we have
  \[ \E \iprod{\hat{a},x}^4 = 3 \E \iprod{\hat{a},x}^2\iprod{\hat{b},x}^2. \]

  So let $r \seteq \E \iprod{\hat{a},x}^2 \iprod{\hat{b},x}^2 = \tfrac{1}{3} \E \iprod{\hat{a},x}^4$.
  By rotational symmetry, $r$ is constant over choice of orthogonal unit vectors $\hat{a}$ and $\hat{b}$.

  Since $\Sigma$ operates only on the symmetric subspace,
  let $u \in \R^{d^2}$ be any unit vector in the symmetric subspace.
  Such a $u$ unfolds to a symmetric matrix in $\R^{d \times d}$,
  so that it has an eigendecomposition $u = \sum_{i=1}^d \lambda_i u_i \tensor u_i$.
  Evaluating $\iprod{u, \Sigma u}$,
  \begin{align*}
    \iprod{u, \Sigma u} & = \sum_{i,j = 1}^d \E \lambda_i \lambda_j \iprod{x,u_i}^2 \iprod {x,u_j}^2 \quad \text{ other terms are $0$ by above}\\
                        & = 3r\sum_{i=1}^d \lambda_i^2 + r\sum_{i \ne j} \lambda_i \lambda_j \\
                        & = 2r \sum_{i=1}^d \lambda_i^2 + r\Paren{\sum_{i = 1}^d \lambda_i }^2 \\
                        & = 2r \, \|u\|^2 + r \Paren{\sum_{i = 1}^d \lambda_i }^2\quad \text{Frobenious norm is sum of squared eigenvalues}\\
                        & = 2r \, \|u\|^2 + r \Paren{\sum_i u_{i,i}}^2 \quad \text{trace is sum of eigenvalues}\\
                        & = 2r \, \iprod{u, \Pisym u} + r \, \iprod{u, \Phi \Phi^\top u}\mcom
  \end{align*}
  so therefore $\Sigma = 2r \,\Pisym + r \, \Phi\Phi^{\top}$.
\end{proof}
\begin{proof}[Proof of \pref{fact:sigma}]
  When $x \sim \cN(0,\Id_d)$, the expectation $\E \iprod{x,a}^2 \iprod{x,b}^2 = 1$ is just a product of
  independent standard Gaussian second moments.
  Therefore by \pref{lem:cov-rot-sym}, $\Sigma = 2 \, \Pisym + \Phi\Phi^{\top}$.

  To find $\tilde \Sigma$ where $x$ is uniformly distributed on the unit sphere,
  we compute
  \[
    1 = \E \|x\|^4 = \sum_{i,j} \E x_i^2 x_j^2 = d \, \E x_1^4 + (d^2 - d) \, \E x_1^2 x_2^2
  \]
  and use the fact that $\E x_1^4 = 3\, \E x_1^2$ (by \pref{lem:cov-rot-sym}) to find that
  $\E x_1^2 x_2^2 = \tfrac{1}{d^2 + 2d}$, and therefore by \pref{lem:cov-rot-sym},
  $\tilde \Sigma = \frac{2}{d^2 + 2d} \Pisym + \frac{1}{d^2 + 2d} \Phi\Phi^{\top}$.

  To verify the pseudoinverses, it is enough to check that $MM^+ = \Pisym$ for each matrix $M$
  and its claimed pseudoinverse $M^+$.

  To show that
    \[
      \|\signh (v\tensor v) \|_2^2 = \left(1-\tfrac{1}{d+2}\right)\cdot \|v\|^4 \mcom
    \]
  for any $v \in \R^d$, we write
  $\|\signh(v \tensor v)\|_2^2 = (v\tensor v)^\top R^2(v\tensor v)$
  and use the substitution $R^2 = 2\Sigma^+$, along with the facts that
  $\Pisym (v \tensor v) = v \tensor v$ and $\iprod{\Phi, v \tensor v} = \|v\|^2$.

\end{proof}

Now we can prove some concentration claims we deferred:
\restatelemma{lem:gap-scalar}
\begin{proof}[Proof of \pref{lem:gap-scalar}]
We prove the first item:
\begin{align*}
  \sum_{i \neq j} \iprod{u_j, R^2 u_i}^2 & = \sum_{i \neq j} \iprod{u_j, 2 \Sigma^+ u_i}^2\nonumber \\
  & = \sum_{i \neq j} \iprod{u_j, (\Pisym - \tfrac{1}{d+2} \Phi \Phi^\top)u_i}^2 \quad \text{by \pref{fact:sigma}}\nonumber\\
  & = \sum_{i \neq j} (\iprod{a_j, a_i}^2 - \tfrac{1}{d+2} \|u_j\|^2 \|u_i\|^2)^2 \nonumber\\
  & = \sum_{i \neq j} \tO(1/d)^2 \quad \text{\wovp by \pref{fact:SIP}}\nonumber\\
  & = \tO(n/d^2)\mper
\end{align*}
And one direction of the second item, using \pref{fact:sigma} and \pref{fact:SIP} (the other direction is similar):
\begin{align*}
  \|Ru_j\|^2 = \iprod{u_j, R^2 u_j} = \iprod{u_j, (\Pisym + \tfrac{1}{d+2} \Phi \Phi^\top)u_j} = (1 - \Theta(1/d))\|a_j\|^4 = 1 - \tO(1/\sqrt{d})
\end{align*}
where the last equality holds \wovp.
\end{proof}

\subsubsection{Proof of \pref{lem:bounded-signal}}
\label{sec:bounded-signal}

To prove \pref{lem:bounded-signal} we will begin by reducing to the case $S = [n]$ via the following.
\begin{lemma}
  \label{lem:p-pi-1}
  Let $v_1,\ldots,v_n \in \R^d$.
  Let $A_S$ have columns $\{v_i\}_{i \in S}$.
  Let $\Pi_S$ be the projector to $\Span \{ v_i \}_{i \in S}$.
  Suppose there is $c \geq 0$ so that $\|A_{[n]}^\top A_{[n]} - \Id_n\| \leq c$.
  Then for every $S \subseteq [n]$, $\|A_{S} A_{S}^\top - \Pi_S\| \leq c$
\end{lemma}
\begin{proof}
  If the hypothesized bound $\|A_{[n]}^\top A_{[n]} - \Id_n\| \leq c$ holds then for every $S \subseteq [n]$ we get $\|A_S^\top A_S - \Id_{|S|} \| \leq c$ since $A_S^\top A_S$ is a principal submatrix of $A_{[n]}^\top A_{[n]}$.
  If $\|A_S^\top A_S - \Id_{|S|}\| \leq c$, then because $A_S A_S^\top$ has the same nonzero eigenvalues as $A_S^\top A_S$,
  we must have also $\|A_S A_S^\top - \Pi_S\| \leq c$.
\end{proof}

It will be convenient to reduce concentration for matrices involving $a_i \tensor a_i$ to analogous matrices where the vectors $a_i \tensor a_i$ are replaced by isotropic vectors of constant norm.
The following lemma shows how to do this.
\begin{lemma}
  \label{lem:isotropy}
  Let $a \sim \cN(0, \tfrac 1 d \Id_d)$.
  Let $\tilde \Sigma \seteq \E_{x \sim \cN(0, \Id_d)} (xx^\top)^{\tensor 2} / \|x\|^4 $.
  Then $u \seteq (\tilde \Sigma^{+})^{1/2} a \tensor a / \|a\|^2$ is an isotropic random vector in the symmetric subspace $\Span \{ y \tensor y \mid y \in \R^d \}$
  with $\|u\| = \sqrt{\dim \Span \{ y \tensor y \mid y \in \R^d \} }$.
\end{lemma}
\begin{proof}
  The vector $u$ is isotropic by definition so we prove the norm claim.
  Let $\tilde \Phi = \Phi / \|\Phi \|$.
  By \pref{fact:sigma},
  \[
    {\tilde \Sigma}^+ = \tfrac{d^2 + 2d}{2} \Pisym - \tfrac{d}{2} \Phi \Phi^\top
  \]
  Thus,
  \[
  \|u \|^2 = \iprod{ \tfrac{a \tensor a}{\|a\|^2}, {\tilde \Sigma}^+ \tfrac{a \tensor a}{\|a\|^2}}
    = \tfrac{d^2 + 2d}{2} - \tfrac{d}{2} = \tfrac{d^2 + d}{2} = \dim \Span \{ y \tensor y \mid y \in \R^d \}\mper\qedhere
  \]
  \end{proof}
  The last ingredient to finish the spectral bound is a bound on the incoherence of independent samples from $(\tilde \Sigma^+)^{1/2}$.
  \begin{lemma}
    \label{lem:incoherence}
    Let $\tilde \Sigma = \E_{a \sim \cN(0,\Id_d)} (aa^\top \tensor aa^\top)/\|a\|^4$.
    Let $a_1,\ldots,a_n \sim \cN(0,\Id_d)$ be independent, and let $u_i = (\tilde \Sigma^+)^{1/2} (a_i \tensor a_i) / \|a_i \|^2$.
    Let $d' = \dim \Span \{ y \tensor y \mid y \in \R^d \} = \tfrac{1}{2}(d^2 + d)$.
    Then
    \[
       \tfrac 1 {d'} \E  \max_i \sum_{j \neq i} \iprod{u_i, u_j}^2 \leq\tO(n)\mper
    \]
  \end{lemma}
  \begin{proof}
    Expanding $\iprod{u_i, u_j}^2$ and using ${\tilde \Sigma}^+ = \tfrac{d^2 + 2d}{2} \Pisym - \tfrac d 2 \Phi \Phi^\top$, we get
    \[
      \iprod{u_i, u_j}^2 = \Paren{\tfrac{d^2 + 2d}{2} \iprod{\tfrac{a_i \tensor a_i}{\|a_i\|^2}, \tfrac{a_j \tensor a_j}{\|a_j\|^2}} - \tfrac d 2}^2
        = \Paren{\tfrac{d^2 + 2d}{2} \cdot \tfrac{\iprod{a_i, a_j}^2}{\|a_i\|^2 \|a_j\|^2} - \tfrac d 2}^2
    \]
    From elementary concentration, $\E \max_{i \neq j} \iprod{a_i, a_j}^2 / \|a_i\|^2 \|a_j\|^2 \leq \tO(1/d)$, so the lemma follows by elementary manipulations.
  \end{proof}

We need the following bound on the deviation from expectation of a tall matrix with independent columns.
\begin{theorem}[Theorem 5.62 in \cite{DBLP:journals/corr/abs-1011-3027}]
  \label{thm:indep-columns}
  Let $A$ be an $N \times n$ matrix $(N \geq n)$ whose columns $A_j$ are independent isotropic random vectors in $\R^N$ with $\|A_j \|_2 = \sqrt N$ almost surely.
  Consider the incoherence parameter
  \[
    m \defeq \frac 1 N \E \max_{\inn} \sum_{\jni} \iprod{A_i, A_j}^2 \mper
  \]
  Then $\E \|\frac 1 N A^T A - \Id \| \leq C_0 \sqrt{\frac{m \log n}{N}}$.
\end{theorem}

We are now prepared to handle the case of $S = [n]$ via spectral concentration for matrices with independent columns, \pref{thm:indep-columns}.
\restatelemma{lem:bounded-signal}

\begin{proof}[Proof of \pref{lem:bounded-signal}]
  By \pref{lem:p-pi-1} it is enough to prove the lemma in the case of $S = [n]$.
  For this we will use \pref{thm:indep-columns}.
  Let $A$ be the matrix whose columns are given by $a_i \tensor a_i$, so that $P_{[n]} = P = AA^\top$.
  Because $RAA^\top R$ and $A^{\top}RRA$ have the same nonzero eigenvalues,
  it will be enough to show that $\|A^\top R^2 A - \Id\| \leq \tO(\sqrt{n}/d) + \tO(n/d^{3/2})$ with probability $1 - o(1)$.
  (Since $n \leq d$ we have $\sqrt n / d = \tO(n/d^{3/2})$ so this gives the theorem.)

  The columns of $RA$ are independent, given by $R (a_i \tensor a_i)$.
  However, they do not quite satisfy the normalization conditions needed for \pref{thm:indep-columns}.
  Let $D$ be the diagonal matrix whose $i$-th diagonal entry is $\|a_i\|^2$.
  Let $\tilde \Sigma = \E_{x \sim \cN(0,\Id)} (xx^\top)^{\tensor 2} / \|x\|^4$.
  Then by \pref{lem:isotropy} the matrix $(\tilde \Sigma^+)^{1/2} D^{-1} A$ has independent columns from an isotropic distribution with a fixed norm $d'$.
  Together with \pref{lem:incoherence} this is enough to apply \pref{thm:indep-columns} to conclude that $\E \|\tfrac 1 {(d')^2} A^\top D^{-1} \tilde \Sigma^+ D^{-1} A - \Id\| \leq \tO(\sqrt n /d)$.
  By Markov's inequality, $\|\tfrac 1 {(d')^2} A^\top D^{-1} \tilde \Sigma^+ D^{-1} A - \Id\| \leq \tO(\sqrt n /d)$ with probability $1 - o(1)$.

  We will show next that $\|A^\top R^2 A - \tfrac 1 {(d')^2} A^\top D^{-1} \tilde \Sigma^+ D^{-1} A \| \leq \tO(n /d^{3/2})$ with probability $1 - o(1)$; the lemma then follows by triangle inequality.
  The expression inside the norm expands as
  \[
    A^{\top}(R^2 - \tfrac 1 {(d')^2} D^{-1} \tilde \Sigma^+ D^{-1} )A\mper
  \]
  and so
  \[
   \|A^\top R^2 A - \tfrac 1 {(d')^2} A^\top D^{-1} \tilde \Sigma^+ D^{-1} A \| \leq \|A\|^2 \|R^2 - \tfrac 1 {(d')^2} D^{-1} \tilde \Sigma^+ D^{-1}\|
  \]
  By \pref{fact:SIP}, with overwhelming probability $\|D - \Id\| \leq \tO(1/\sqrt{d})$.
  So $\|(1/d')^2 D^{-1} \tilde \Sigma^+ D^{-1} - (1/d')^2 \tilde \Sigma^+ \| \leq \tO(1/\sqrt{d})$ \wovp.
  We recall from \pref{fact:sigma}, given that $R = \sqrt{2}\cdot (\Sigma^+)^{1/2}$, that
  \[
    R^2 = \Pisym - \tfrac{1}{d + 2} \Phi \Phi^\top \quad \text{and} \quad \tfrac{1}{(d')^2}\tilde \Sigma^+ = \tfrac{d+2}{d+1} \Pisym - \tfrac 1 {d+1} \Phi \Phi^\top\mper
  \]
  This implies that $\|R^2 - (1/d')^2 \tilde \Sigma^+\| \leq O(1/d)$.
  Finally, by an easy application of \pref{prop:truncated-bernstein}, $\|A\|^2 = \|\sum_i (a_i a_i^\top)^{\tensor 2} \| \leq \tO(n/d)$ \wovp.
  All together, $\|A^\top R^2 A - \tfrac 1 {(d')^2} A^\top D^{-1} \tilde \Sigma^+ D^{-1} A \| \leq \tO(n/d^{3/2})$.
\end{proof}
\section{Concentration bounds for tensor principal component analysis}
\label{app:tensor-pca}
For convenience, we restate \pref{lem:tpca-concentration} here.
\begin{lemma}[Restatement of \pref{lem:tpca-concentration}]
  For any $v$, with high probability over $\bA$, the following occur:
  \begin{align*}
      \left \| \sum_i \Tr(A_i) \cdot A_i \right \| & \leq O(n^{3/2} \log^2 n)\\
    \left \| \sum_i v(i) \cdot A_i \right \| & \leq O(\sqrt n \log n)\\
    \left \|\sum_i \Tr(A_i) v(i) \cdot vv^T \right \| & \leq O(\sqrt n \log n)\mper
  \end{align*}
\end{lemma}
\begin{proof}[Proof of \pref{lem:tpca-concentration}]
  We begin with the term $\sum_i \Tr(A_i) \cdot A_i$.
  It is a sum of iid matrices $\Tr(A_i) \cdot A_i$.
  A routine computation gives $\E \Tr(A_i) \cdot A_i = \Id$.
  We will use the truncated matrix Bernstein's inequality (\pref{prop:truncated-bernstein}) to bound $\| \sum_i \Tr(A_i) A_i \|$.

  For notational convenience, let $A$ be distributed like a generic $A_i$.
  By a union bound, we have both of the following:
  \begin{align*}
    \Pr\Big( \| \Tr(A) \cdot A\| \geq tn\Big)
      &\leq \Pr\Big(|\Tr(A)| \geq \sqrt{tn}\Big)
      + \Pr\Big(\|A\| \geq \sqrt{tn}\Big)\\
    \Pr\Big( \| \Tr(A) \cdot A - \Id\| \geq (t + 1)n\Big)
      &\leq \Pr\Big(|\Tr(A)| \geq \sqrt{tn}\Big)
      + \Pr\Big(\|A\| \geq \sqrt{tn}\Big)\mper
  \end{align*}

  Since $\Tr(A)$ the sum of iid Gaussians, $\Pr(|\Tr(A)| \geq \sqrt{tn}) \leq e^{-c_1t}$ for some constant $c_1$.
  Similarly, since the maximum eigenvalue of a matrix with iid entries has a subgaussian tail, $\Pr(\|A\| \geq \sqrt{tn}) \leq e^{-c_2t}$
  for some $c_2$.
  All together, for some $c_3$, we get
  $\Pr( \|\Tr(A) \cdot A \| \geq tn) \leq e^{-c_3t}$ and
  $\Pr( \|\Tr(A) \cdot A - \Id \| \geq (t + 1)n) \leq e^{-c_3t}$.

  For a positive parameter $\beta$, let $\Ind_{\beta}$ be the indicator variable
  for the event $\|\Tr(A) \cdot A \| \leq \beta$.
  Then
  \begin{align*}
    \E \|\Tr(A)\cdot A\| - \E \|\Tr(A)\cdot A\| \Ind_{\beta}
        & = \int_0^{\infty} \big[\Pr(\|\Tr A \cdot A\| > s)
                            -\Pr(\|\Tr A \cdot A\|\Ind_{\beta} > s)\big]\d s\\
        & = \beta \Pr(\|\Tr A \cdot A\| > \beta)
              + \int_{\beta}^{\infty} \Pr(\|\Tr A \cdot A\| > s) \d s \\
        & \le \beta e^{-c_3\beta/n}
              + \int_{\beta}^{\infty} \Pr(\|\Tr A \cdot A\| > s) \d s \\
        & = \beta e^{-c_3\beta/n}  + \int_{\beta/n}^\infty
            \Pr(\|\Tr A \cdot A\| \geq tn) \, n \d t\\
        & \leq \beta e^{-c_3\beta/n} + \int_{\beta/n}^\infty n e^{-c_3t} \d t\\
        & = \beta e^{-c_3\beta/n} + \tfrac{n}{c_3} e^{-c_3\beta/n} \mper
  \end{align*}
  Thus, for some $\beta = O(n \log n)$ we may take the parameters $p,q$ of \pref{prop:truncated-bernstein} to be $O(n^{-150})$.
  The only thing that remains is to bound the parameter $\sigma^2$.
  Since $(\E \Tr(A) \cdot A)^2 = \Id$, it is enough just to bound $\| \E \Tr(A)^2 AA^T \|$.
  We use again a union bound:
  \[
    \Pr( \| \Tr(A)^2 AA^T \| > tn^2 ) \leq \Pr (|\Tr(A)| > t^{1/4} \sqrt{n}) + \Pr(\|A\| > t^{1/4} \sqrt{n}) \mper
  \]
  By a similar argument as before, using the Gaussian tails of $\Tr A$ and $\|A\|$,
  we get $\Pr( \| \Tr(A)^2 AA^T \| > tn^2) \le e^{-c_4\sqrt{t}}$.
  Then starting out with the triangle inequality,
  \begin{align*}
    \sigma^2
    &=   \| n \cdot \E \Tr(A)^2AA^T \| \\
    &\le n \cdot \E \| \Tr(A)^2AA^T \| \\
    &=   n\cdot \int_0^{\infty} \Pr(\Tr(A)^2AA^T > s) \d s \\
    &=   n \cdot \int_0^{\infty} \Pr(\Tr(A)^2AA^T > tn^2) \, n^2 \d t \\
    &\le n \cdot \int_0^{\infty} e^{-c_4\sqrt{t}} \, n^2 \d t \\
    &= n \cdot \left[ -\frac{2n^2(c_4\sqrt{t} + 1)}{c_4^2} e^{-c_4\sqrt{t}}
        \right]_{t = 0}^{t = \infty} \\
    &\leq O(n^3) \mper
  \end{align*}
  This gives that with high probability,
  \[
    \left \| \sum_i \Tr(A_i) \cdot A_i \right \| \leq O(n^{3/2} \log^2 n)\mper
  \]

  The other matrices are easier.
  First of all, we note that the matrix $\sum_i v(i) \cdot A_i$ has independent standard Gaussian entries,
  so it is standard that with high probability $\|\sum_i v(i) \cdot A_i \| \leq O(\sqrt n \log n)$.
  Second, we have
  \[
    \sum_i v(i) \Tr(A_i) vv^T = vv^T \sum_i v(i) \Tr(A_i).
  \]
  The random variable $\Tr(A_i)$ is a centered Gaussian with variance $n$,
  and since $v$ is a unit vector, $\sum_i v(i) \Tr(A_i)$ is also a centered
  Gaussian with variance $n$.
  So with high probability we get
  \[
    \left\| vv^T \sum_i v(i) \Tr(A_i) \right\|
      = \left|\sum_i v(i) \Tr(A_i)\right|
      \leq O(\sqrt{n} \log n)
  \]
  by standard estimates.
  This completes the proof.
\end{proof}

\end{document}

